\documentclass[12pt, conference, final,
onecolumn,technote,a4paper]{IEEEtran}
\usepackage{latexsym}
\usepackage{tcolorbox}
\usepackage{enumitem}
\usepackage{todonotes}
\usepackage[numbers,square,comma,compress]{natbib}
\usepackage{mathrsfs}
\usepackage[normalem]{ulem}
\usepackage{url}
\usepackage{amsfonts,amssymb,amsmath,amsthm,bm}
\usepackage{xspace}
\usepackage{xcolor}
\usepackage{hyperref}
\usepackage{cleveref}
\newtheorem{theorem}{Theorem}
\newtheorem{assumption}{Assumption}

\newtheorem{definition}[theorem]{Definition}

\newtheorem{corollary}{Corollary}
\newtheorem{lemma}[theorem]{Lemma}

\newtheorem{construction}{Construction}

\newcommand{\cproblem}[1]{\ensuremath{\mathsf{#1}}\xspace}
\newcommand{\LWE}{\cproblem{LWE}}

\newcommand{\rand}{\xleftarrow{\$}}

\newcommand{\ABE}{\cproblem{ABE}}
\newcommand{\lhltr}{\cproblem{LHL}\text{-}\cproblem{Trap}}
\newcommand{\EXP}{\cproblem{Exp}}

\newcommand{\NC}{\cproblem{NC}}

\newcommand{\LSSS}{\cproblem{LSSS}}

\newcommand{\gameitem}[1]{\vspace{0.5em}\textbf{\underline{#1}}\vspace{0.5em}}

\newcommand{\pk}{\cproblem{pk}}
\newcommand{\msk}{\cproblem{msk}}

\newcommand{\keygen}{\cproblem{Keygen}}
\newcommand{\sk}{\cproblem{sk}}

\newcommand{\pp}{\cproblem{pp}}
\newcommand{\setup}{\cproblem{Setup}}

\newcommand{\msg}{\cproblem{msg}}

\newcommand{\enc}{\cproblem{Enc}}
\newcommand{\ct}{\cproblem{ct}}
\newcommand{\dec}{\cproblem{Dec}}
\renewcommand{\H}{\cproblem{H}}
\newcommand{\ppt}{\cproblem{PPT}}

\newcommand{\A}{\mathcal{A}}
\newcommand{\pr}{\mathrm{Pr}}

\newcommand{\trapgen}{\cproblem{TrapGen}}
\newcommand{\samplepre}{\cproblem{SamplePre}}
\newcommand{\td}{\cproblem{td}}

\newcommand{\Adv}{\cproblem{Adv}}

\newcommand{\ver}{\cproblem{Ver^{mx}}}
\newcommand{\up}{\cproblem{up}}
\newcommand{\down}{\cproblem{down}}
\newcommand{\Com}{\cproblem{Com^{mx}}}
\newcommand{\Open}{\cproblem{Open^{mx}}}

\newcommand{\iitem}[1]{\bigskip\noindent\textbf{#1.}\ }

\newcommand{\KPABE}{\cproblem{KP}\text{-}\cproblem{ABE}}
\newcommand{\CPABE}{\cproblem{CP}\text{-}\cproblem{ABE}}

\DeclareMathOperator{\poly}{poly}
\DeclareMathOperator{\negl}{negl}

\newcommand{\matA}{\ensuremath{\mathbf{A}}}
\newcommand{\matB}{\ensuremath{\mathbf{B}}}
\newcommand{\matC}{\ensuremath{\mathbf{C}}}
\newcommand{\matD}{\ensuremath{\mathbf{D}}}

\newcommand{\matG}{\ensuremath{\mathbf{G}}}
\newcommand{\matH}{\ensuremath{\mathbf{H}}}
\newcommand{\matI}{\ensuremath{\mathbf{I}}}

\newcommand{\matM}{\ensuremath{\mathbf{M}}}
\newcommand{\matN}{\ensuremath{\mathbf{N}}}

\newcommand{\matQ}{\ensuremath{\mathbf{Q}}}
\newcommand{\matR}{\ensuremath{\mathbf{R}}}
\newcommand{\matS}{\ensuremath{\mathbf{S}}}
\newcommand{\matT}{\ensuremath{\mathbf{T}}}
\newcommand{\matU}{\ensuremath{\mathbf{U}}}
\newcommand{\matV}{\ensuremath{\mathbf{V}}}
\newcommand{\matW}{\ensuremath{\mathbf{W}}}

\newcommand{\matZ}{\ensuremath{\mathbf{Z}}}

\newcommand{\vecb}{\ensuremath{\mathbf{b}}}
\newcommand{\vecc}{\ensuremath{\mathbf{c}}}
\newcommand{\vecd}{\ensuremath{\mathbf{d}}}
\newcommand{\vece}{\ensuremath{\mathbf{e}}}

\newcommand{\vecg}{\ensuremath{\mathbf{g}}}

\newcommand{\veck}{\ensuremath{\mathbf{k}}}
\newcommand{\vecn}{\ensuremath{\mathbf{n}}}

\newcommand{\vecr}{\ensuremath{\mathbf{r}}}
\newcommand{\vecs}{\ensuremath{\mathbf{s}}}
\newcommand{\vect}{\ensuremath{\mathbf{t}}}
\newcommand{\vecu}{\ensuremath{\mathbf{u}}}
\newcommand{\vecv}{\ensuremath{\mathbf{v}}}

\newcommand{\vecx}{\ensuremath{\mathbf{x}}}
\newcommand{\vecy}{\ensuremath{\mathbf{y}}}
\newcommand{\vecz}{\ensuremath{\mathbf{z}}}
\newcommand{\veczero}{\ensuremath{\mathbf{0}}}

\newcommand{\N}{\ensuremath{\mathbb{N}}}

\newcommand{\R}{\ensuremath{\mathbb{R}}}

\newcommand{\Z}{\ensuremath{\mathbb{Z}}}

\begin{document}
\title{Ciphertext-Policy $\ABE$ for $\NC^1$ Circuits with Constant-Size Ciphertexts from Succinct $\LWE$}
\author{
  Jiaqi Liu,
  Yuanyi Zhang, 
  Fang-Wei Fu\thanks{Jiaqi Liu, Yuanyi Zhang and Fang-Wei Fu are with Chern Institute of Mathematics and LPMC, Nankai University, Tianjin 300071, P. R. China, Emails: ljqi@mail.nankai.edu.cn, yuanyiz@mail.nankai.edu.cn, fwfu@nankai.edu.cn.}}
  
\maketitle
\begin{abstract}
    We construct a lattice-based ciphertext-policy attribute-based encryption ($\CPABE$) scheme for $\NC^1$ access policies with constant-size ciphertexts.  Let $\lambda$ be the security parameter. For an $\NC^1$ circuit of depth $d$ and size $s$ on $\ell$-bit inputs, our scheme has the public-key and ciphertext sizes $O(1)$  (independent of $d$), and secret-key size $O(\ell)$, where the $O(\cdot)$ hides $\poly(\lambda)$ factors. As an application, we obtain a broadcast encryption scheme for $N$ users with ciphertext size $\poly(\lambda)$ independent of $\log N$ and key sizes $\poly(\lambda,\log N)$. Our construction is selectively secure in the standard model under the $\poly(\lambda)$-succinct \LWE\ assumption introduced by Wee (CRYPTO~2024). 
\end{abstract}

\section{Introduction}
In a traditional public-key encryption scheme, only a user who possesses the secret key corresponding to the public key used for encryption can decrypt the ciphertext. As a generalization of public-key encryption, \emph{attribute-based encryption} (\ABE) realizes fine-grained access control over encrypted data. There are two types of $\ABE$ schemes: \emph{ciphertext-policy} \ABE ($\CPABE$) and \emph{key-policy} \ABE ($\KPABE$). In a $\CPABE$ scheme, each secret key $\sk_{\vecx}$ is associated with an attribute $\vecx\in\{0,1\}^{\ell}$, and each ciphertext $\ct_f$ is associated with an access policy $f$ that specifies which attribute vectors are authorized to decrypt it. In contrast, in a $\KPABE$ scheme, each secret key $\sk_f$ is associated with an access policy $f$, and each ciphertext $\ct_{\vecx}$ is associated with an attribute $\vecx$. In both cases, decryption succeeds only if the attribute satisfies the access policy, i.e., $f(\vecx)=0$; otherwise, the decryptor learns no information about the underlying plaintext.

The concept of \ABE was first proposed independently by Sahai and Waters~\cite{SW05} and by Goyal et al.~\cite{GPSW06}. Since then, \ABE has become a fundamental cryptographic primitive, and many concrete schemes have been proposed. A large fraction of existing \ABE constructions rely on assumptions over bilinear groups. However, due to known quantum attacks on bilinear-group-based assumptions, there has been substantial interest in constructing schemes from alternative assumptions. In particular, many subsequent \ABE schemes base their security on the \LWE assumption or its variants, which are believed to be resistant to quantum attacks. In \cite{AFV11,Boy13,GVW13,BGG+14,GV15,BV16,Tsa19,AMY20,CW23}, a series of $\KPABE$ schemes were proposed, all based on the $\LWE$ assumption.  %The first $\LWE$-based $\CPABE$ scheme avoiding the generic tranformation was proposed much later than the first $\LWE$-based $\KPABE$ scheme. Schemes such as \cite{}

A generic transformation can convert a $\KPABE$ scheme into a $\CPABE$ one via a universal circuit. However, this transformation typically leads to significantly worse parameters. To optimize the parameters of $\CPABE$ scheme, Agrawal and Yamada \cite{AY20b} proposed a $\CPABE$ scheme for all $\NC^{1}$ circuits with ciphertext size independent of the circuit size $s$, relying on the $\LWE$ and generic bilinear groups. Datta et al. \cite{DKW21a} proposed a $\CPABE$ scheme for all $\NC^{1}$ circuits via linear secret sharing schemes, relying on the hardness of the $\LWE$ assumption. Subsequently, there are many $\CPABE$ schemes relying on the nonstandard $\LWE$-related assumptions. Wee \cite{Wee22} proposed a $\CPABE$ scheme based on the evasive $\LWE$ assumption and tensor $\LWE$ assumption. Hsieh et al.~\cite{HLL24} presented a $\CPABE$ scheme based on the (structured) evasive $\LWE$ assumption. Agrawal et al.~\cite{AKY24} proposed a $\CPABE$ scheme for all Turing machines based on the circular evasive $\LWE$ and tensor $\LWE$. Wee \cite{Wee24,Wee25} presents $\CPABE$ scheme for all circuits based on the weaker $\poly(\lambda,d)$-succinct $\LWE$ assumption with ciphertext size independent of the attribute length $\ell$.

\subsection{Our results}
We list our results as below:

1. In this work, we present a $\CPABE$ scheme for access policies represented by $\NC^1$ circuits. For circuits  on $\ell$-bit inputs with depth $d$ and size $s$ (hence $s=\poly(\ell)$ and $d=\log(\ell)$) where $s,d,\ell$ are fixed at setup.  We construct a $\CPABE$ scheme with parameters:
$$|\pk|=O(s),\quad |\sk|=O(\ell),\quad |\ct|= O(1),$$ where the notation $O(\cdot)$ hides $\poly(\lambda)$ factors. It proves to be selectively secure against unbounded collusions. The security relies on the $\poly(\lambda)$-succinct $\LWE$ assumptions with sub-exponential modulus-to-noise ratio.

\begin{itemize}
    \item The scheme we construct is a \emph{succinct} $\CPABE$. Namely, the ciphertext size in our scheme relies only on the security parameter $\lambda$ and independent of the size $s$, the input length $\ell$, and the depth $d$, which is asymptotically optimal up to $\poly(\lambda)$ factors.
    \item The size of the public parameters scales with the circuit size $s$. Though it is seen as not competitive compared with previous schemes (many schemes have public key size $O(\ell)$, such as \cite{BV22,Wee22,HLL24}), we observe that in the public parameters, all elements are sampled uniformly except for the public parameters used in the succinct $\LWE$ instance (i.e., $\pp=(\matB,\matW,\matT)$), which is of size $\poly(\lambda)$. Therefore, we can generate the purely uniform part with a pseudorandom generator instead. We can then compress the size of the public parameters into $O(1)$ size.

    \item The bound on the size of $\sk$ is not tight. It depends on the number of attributes associated with the secret key rather than strictly on $\ell$. In scenarios where the users possesses only $O(1)$ attributes, the secret-key size can be reduced to $O(1)$.
\end{itemize}

2. $\CPABE$ for $\NC^1$ circuits gives a broadcast encryption for $N=\ell$ users as a special case. In the broadcast encryption, we can use a circuit of size $O(N\log N)$ and depth $O(\log N)$ to check the membership of the users. We obtain a broadcast encryption scheme for $N$ users with parameters 
$$|\pk|=N\cdot \poly(\lambda,\log N),\quad |\sk|=\poly (\lambda,\log N),\quad |\ct|=\poly(\lambda).$$
By using a pseudorandom generator to compress the public parameter, we can reduce $|\pk|$ to $\poly(\lambda)$. Since $|\pk|+|\ct|+|\sk|=\poly(\lambda,\log N)$, we obtain an optimal broadcast encryption (under $\poly(\lambda)$-succinct $\LWE$ assumption).

3. In the proof of the selective security of our construction, we use the technique of using the trapdoor $\matT$ for $[\matI_{\ell}\otimes\matB\mid\matW]$ in the $\ell$-succinct $\LWE$ assumption to generate other trapdoors for the matrix of the form $[\matI_{\ell'}\otimes \matB\mid\matW']$, where $\matW'$ is statistically close to uniform.
 
\subsection{Related Works}
\iitem{Comparison with \cite{Wee25}} We use the $\{0,1\}$-$\LSSS$ technique from \cite{DKW21a} and the matrix commitment scheme from \cite{Wee25} to construct our $\CPABE$ scheme for $\NC^1$ scheme and obtain an optimal broadcast encryption scheme. Compared with \cite{Wee25}, our scheme supports a narrower class of access policies and incurs a larger secret-key size. Nevertheless, we improve the parameters in the following respects: (1) $|\pk|$, $|\sk|$, and $|\ct|$ are independent of the circuit depth $d$; (2) the resulting broadcast encryption has ciphertext size $|\ct|$ independent of $\log N$, where $N$ is the number of users in the system; (3) during encryption, we execute only a single matrix commitment on $O(\ell)$-width inputs, rather than running a circuit commitment scheme on an $\ell$-input circuit that composes a sequence of matrix commitments.

\subsection{Paper Organization}
The paper is organized as follows. Section~\ref{technique} provides a technical overview of our main construction and ideas. Section~\ref{sec:pre} introduces necessary notations, basic concepts of lattices, the \LWE assumption, and its variants. In Section~\ref{sec:CPABEnotion}, we formalize the linear secret sharing scheme ($\LSSS$) and ciphertext-policy attribute-based encryption ($\CPABE$) for access structures realized by $\LSSS$. Section~\ref{sec:scheme} presents our construction of an $\CPABE$ scheme and its security analysis. 

\section{Technical Overviews}\label{technique}
In this section, we demonstrate a high-level overview of our construction. Throughout this section, we adopt the wavy underline to indicate that a term is perturbed by some noise. For example, we use the notation $\smash{\uwave{\vecs^\top\matA}}$ to denote the term $\vecs^\top\matA+\vece$ for an unspecified noise vector $\vece$.  For a discrete set $S$, we write $\vecx\rand S$ to denote that $\vecx$ is uniformly sampled from the set $S$. The notation $\matA^{-1}(\vecy)$ means sampling a short preimage $\vecx$ from discrete Gaussian distribution such that $\matA\vecx=\vecy$. For a positive integer $n\in \N$, denote $[n]$ as the set $\{k\in\N:1\leq k\leq n\}$. We use $\overset{c}{\approx}$ to denote computational indistinguishability. We defer the formal definitions of all notations to Section~\ref{sec:pre}, and here use them informally for simplicity.

\iitem{$\ell$-succinct $\LWE$} The security of our construction relies on the $\ell$-succinct $\LWE$ assumption introduced by~\cite{Wee24}. More formally, let $n,q,m=O(n\log q)$ be lattice parameters. First, sample $\matB\rand\Z_{q}^{n\times m},\matW\rand\Z_{q}^{\ell n\times m}$, and a Gaussian $\matT\in\Z^{(\ell+1)m\times \ell m}$ such that $[\matI_{\ell}\otimes \matB\mid \matW]\cdot\matT=\matI_{\ell}\otimes \matG$. The $\ell$-succinct $\LWE$ assumption asserts that $$(\matB,\uwave{\vecs^\top\matB},\matW,\matT)\overset{c}{\approx}(\matB,\vecc,\matW,\matT),$$ where $\vecs\rand\Z_q^n,\vecc\rand\Z_q^m$. Throughout the paper, we denote $\pp:=(\matB,\matW,\matT)$. The $\ell$-succinct $\LWE$ assumption is a falsifiable assumption implied by evasive $\LWE$ introduced in~\cite{Wee22}. As the parameter $\ell$ increases, the assumption becomes stronger. We note that the $1$-succinct $\LWE$ assumption is equivalent to the standard $\LWE$ assumption. 

\subsection{Our schemes}
We begin by describing the $\CPABE$ scheme in \cite{DKW21a}, which supports access structures realized by a linear secret sharing scheme ($\LSSS$). 

\iitem{DKW21 scheme} Let $\mathbb{U}$ denote the attribute universe, let $N:=|\mathbb{U}|$, and identify $\mathbb{U}$ with the set $[N]$.
\begin{itemize}
    \item The public key is given by $(\{\matA_u\}_{u\in \mathbb{U}},\{\matQ_u\}_{u\in\mathbb{U}},\vecy)$, where $\matA_u,\matQ_u\rand\Z_q^{n\times m}$ and $\vecy\rand\Z_q^n$.
    \item The master secret key is given by $\{\td_{\matA_{u}}\}_{u\in \mathbb{U}}$, where $\td_{\matA_u}$ is the trapdoor for the matrix $\matA_u$ associated with the attribute $u$. This trapdoor allows efficient sampling of short preimages for $\matA_u$.
    \item Given an attribute set $U\subseteq\mathbb{U}$, the secret key associated with $U$ is sampled as follows: Sample $\hat{\vect}\leftarrow \chi^{m-1}$, where $\chi$ is some error distribution, set $\vect=(1,\hat{\vect}^\top)^\top$ and sample $\veck_u\leftarrow \matA_u^{-1}(-\matQ_u\vect)$. Then the secret key is given by $(\{\veck_u\}_{u\in U},\vect)$.
    \item To encrypt a message $\msg\in\{0,1\}$ under the access structure determined by matrix $\matM\in\{-1,0,1\}^{\ell\times s_{\max}}$, the ciphertext is given as follows: Assume that $\rho$ is an injective function that maps the row indices of $\matM$ to the attribute. First sample $\vecs\rand\Z_q^n$ and $\vecv_2,\ldots,\vecv_{s_{\max}}\in\Z_q^m$. The ciphertext components are computed as \begin{align*}
\vecc_{1,i}^\top&=\uwave{\vecs^\top\matA_{\rho(i)}},\\
\vecc_{2,i}^\top&=M_{i,1}(\vecs^\top\vecy,\underbrace{0,\ldots,0}_{m-1})+\sum_{j\in\{2,\ldots,s_{\max}\}}M_{i,j}\vecv_j^\top+\uwave{\vecs^\top\matQ_{\rho(i)}},\\
c_3&=\uwave{\vecs^\top\vecy}+\msg\cdot\lceil q/2\rfloor.
    \end{align*}
    Then the ciphertext is given as $(\{\vecc_{1,i}\}_{i\in [\ell]},\{\vecc_{2,i}\}_{i\in [\ell]},c_3)$.

\item Suppose that the secret key is associated with an authorized attribute set. Let $I$ be the set of row indices of $\matM$ corresponding to the attributes associated with the secret key. Let $\{w_i\}_{i\in I}\in\{0,1\}$ be the reconstruction coefficients in the $\LSSS$ scheme with respect to the matrix $\matM$. The decryption algorithm proceeds as follows: \begin{align*}
    \sum_{i\in I}w_i(\vecc_{1,i}^\top\veck_{\rho(i)}+\vecc_{2,i}^\top\vect)&\approx\sum_{i\in I}w_i(\vecs^\top\matA_{\rho(i)}\veck_{\rho(i)})+\sum_{i\in I}w_iM_{i,1}
(\vecs^\top\vecy,0,\ldots,0)\vect\\
&+\sum_{i\in I,j\in \{2,\ldots,s_{\max}\}}w_iM_{i,j}\vecv_j^\top\vect+\sum_{i\in I}w_i\vecs^\top\matQ_{\rho(i)}\vect\\
&=\sum_{i\in I}w_iM_{i,1}\vecs^\top\vecy+\sum_{i\in I,j\in \{2,\ldots,s_{\max}\}}w_iM_{i,j}\vecv_j^\top\vect
\end{align*}
Since the reconstruction coefficients $\{w_i\}_{i\in I}$ are such that $\sum_{i\in I}w_iM_{i,1}=1$ and $\sum_{i\in I}w_i M_{i,j}=0$ for $j\in \{2,\ldots,s_{\max}\}$. Hence $\sum_{i\in I}w_i(\vecc_{1,i}^\top\veck_{\rho(i)}+\vecc_{2,i}^\top\vect)\approx \vecs^\top\vecy$, and we can recover $\msg$ by subtracting it from $c_3$ and decoding.
\end{itemize}

From the above description, the ciphertext size scales with $\ell$, i.e., the number of rows of the matrix $\matM$. We expect the ciphertext size to be independent of $\ell$, and we use the matrix commitment technique in~\cite{Wee25} to compress the ciphertext. Precisely, we replace the matrices $\{\matA_{u}\}_{u\in\mathbb{U}}\rand\Z_q^{n\times m}$ with a matrix $\matB\rand\Z_q^{n\times m}$ of the same size. Instead of generating $\vecv_2,\ldots,\vecv_{s_{\max}}\rand \Z_q^m$, we select uniformly random matrices $\matB_2,\ldots,\matB_{s_{\max}}\rand\Z_q^{n\times m}$. We modify the second component of the ciphertext as follows: $$\vecc_{2,i}'=M_{i,1}(\vecs^\top\vecy,0,\ldots,0)+\sum_{j\in\{2,\ldots,s_{\max}\}}M_{i,j}\vecs^\top\matB_j+\uwave{\vecs^\top\matQ_{\rho(i)}}$$ for our commitment purposes (with this modification, the scheme remains selectively secure in the same model). To compress the $\{\vecc_{2,i}'\}$, for $i\in [\ell]$, let $$\matU_{\rho(i)}\leftarrow M_{i,1}(\vecy\mid\veczero\mid\cdots\mid\veczero)+\sum_{j\in\{2,\ldots,s_{\max}\}}M_{i,j}\matB_j+\matQ_{\rho(i)}$$ and for $u\in [N]\setminus\rho([\ell])$ we sample $\matU_{u}\rand\Z_q^{n\times m}$. For $i\in [\ell]$, we have $\vecc_{2,i}'\approx\vecs^\top\matU_{\rho(i)}$. Then we commit to $[\matU_1\mid \cdots\matU_{N}]\in \Z_q^{n\times mN}$, and denote the commitment matrix by $\matC\in \Z_q^{n\times m}$. Roughly, the matrix $\matC$ contains partial information about all $\matU_u$.

\iitem{Modification} With the idea given above, we modify the DKW21 $\CPABE$ scheme as follows:
\begin{itemize}
    \item The public key is given by $(\pp,\matA,\{\matB_i\}_{i\in[s_{\max}]},\{\matQ_u\}_{u\in\mathbb{U}},\vecy)$, where $\pp$ is the public parameter of the $2m^2$-succinct $\LWE$ problem, $\matA,\matB_i,\matQ_u\rand\Z_q^{n\times m}$ and $\vecy\rand\Z_q^n$.
    \item The master secret key is given by $\td_{\matB}$.
    
    \item Given an attribute set $U\subseteq\mathbb{U}$, the secret key associated with $U$ is sampled as follows: Sample $\hat{\vect}\leftarrow \chi^{m-1}$, where $\chi$ is some error distribution, set $\vect=(1,\hat{\vect}^\top)^\top$ and sample $$\veck_u\leftarrow \matB^{-1}((\matA\matV_u+\matQ_u)\vect),$$ where $\matV:=[\matV_1\mid\cdots\mid\matV_N]$ is the verification matrix of a commitment of a matrix of width $Nm$ (which can be computed without the knowledge of the matrix committed). Then the secret key is given by $(\{\veck_u\}_{u\in U},\vect)$.
    \item To encrypt a message $\msg\in\{0,1\}$ under the access structure determined by matrix $\matM\in\{-1,0,1\}^{\ell\times s_{\max}}$, the ciphertext is given as follows: Assume that $\rho$ is an injective function that maps the row indices of $\matM$ to the attribute. First sample $\vecs\rand\Z_q^n$. The ciphertext components are computed as \begin{align*}
\vecc_{1}^\top=\uwave{\vecs^\top\matB},\quad
\vecc_{2}^\top=\uwave{\vecs^{\top}(\matA+\matC)},\quad
c_3=\uwave{\vecs^\top\vecy}+\msg\cdot\lceil q/2\rfloor.
    \end{align*}
    Then the ciphertext is given as $(\vecc_1,\vecc_2,c_3)$.

    In the second component $\vecc_2$, the commitment matrix $\matC$ is such that $$\matC\matV=[\matU_1\mid\cdots\mid\matU_{N}]-\matB\matZ,$$ where for each $i\in [\ell]$, $$\matU_{\rho(i)}\leftarrow M_{i,1}(\vecy\mid\underbrace{\veczero\mid\cdots\mid\veczero}_{m-1})+\sum_{j\in\{2,\ldots,s_{\max}\}}M_{i,j}\matB_j+\matQ_{\rho(i)}$$ and for $u\in [N]\setminus\rho([\ell])$,  $\matU_{u}\rand\Z_q^{n\times m}$.

\item When decrypting the ciphertexts, one can first compute $\matU_{\rho(i)}$ for all $i\in I$, where $I$ is the set of the row indices related to the attributes associated with the secret keys. If we denote $\matV=[\matV_1\mid\cdots\mid\matV_{N}]$ and $\matZ=[\matZ_1\mid\cdots\mid\matZ_N]$, we obtain that $$\matC\matV_{i}=\matU_{i}-\matB\matZ_i,\qquad \forall i\in [N].$$

Therefore, we can recover $\vecc_{2,i}'$ in the DKW21 scheme by $$\vecc_2^\top\matV_{\rho(i)}\approx\vecs^\top(\matA+\matC)\matV_{\rho(i)}=\vecs^\top\matA\matV_{\rho(i)}+\underbrace{\vecs^\top\matU_{\rho(i)}}_{\approx\vecc_{2,i}'}-\underbrace{\vecs^\top\matB\matZ_{\rho(i)}}_{\approx\vecc_1^\top\matZ_{\rho(i)}}.$$
Combining it with the reconstruction coefficients $\{w_i\}_{i\in I}$, we obtain that \begin{align*}
    &\sum_{i\in I}w_i(\vecc_2^\top\matV_{\rho(i)}+\vecc_1^\top\matZ_{\rho(i)})\\
    \approx&\sum_{i\in I}w_i\vecs^\top\matA\matV_{\rho(i)}+\sum_{i\in I}w_i\vecs^\top\matU_{\rho(i)}\\
    =&\sum_{i\in I}w_i\vecs^\top(\matA\matV_{\rho(i)}+\matQ_{\rho(i)})+\sum_{i\in I}w_iM_{i,1}\vecs^\top[\vecy\mid\veczero\mid\cdots\mid\veczero]+\sum_{i\in I,j\in\{2,\ldots,s_{\max}\}}w_iM_{i,j}\vecs^\top\matB_j\\
    =&\sum_{i\in I}w_i\vecs^\top(\matA\matV_{\rho(i)}+\matQ_{\rho(i)})+\vecs^\top[\vecy\mid\veczero\mid\cdots\mid\veczero].
\end{align*}
Then using the secret key equation $\matB\veck_{\rho(i)}=(\matA\matV_{\rho(i)}+\matQ_{\rho(i)})\vect$, we can recover the message by computing $$c_3-\sum_{i\in I}w_i(\vecc_2^\top\matV_{\rho(i)}\vect+\vecc_1^\top\matZ_{\rho(i)}\vect-\vecc_1^\top\veck_{\rho(i)}).$$
\end{itemize}

\subsection{Static Security}
We prove that our scheme is selectively secure under the $2m^2$-succinct $\LWE$ assumption in \ref{subsec:security}. We achieve our goal by defining a sequence of hybrid games. The first hybrid game corresponds to the real game as defined in~\ref{def:selectsecurity}, while the last hybrid game is independent of the messages, where the advantage of the adversary is exactly 0. By arguing each adjacent pair of the hybrid games is indistinguishable, the adversary in the real game can win the game with negligible advantage. 

More precisely, in the reduction, we attempt to prove the ciphertext components $\vecc_1,\vecc_2,c_3$ are pseudorandom, since they resemble $\LWE$ samples intuitively. However, in the key query phase, the challenger uses the trapdoors $\td_{\matB}$ for $\matB$ to generate corresponding secret keys, which may result in the additional information leakage of the $\LWE$ samples. To overcome this problem, we notice that using the matrix $\matT$ in the public parameters in the $2m^2$-succinct $\LWE$ instance, we can construct trapdoors for the matrix of the form $[\matI_{s}\otimes \matB\mid \matW']$, where $\matW'$ is statistically close to uniform. Using this technique, the challenger can respond to the secret-key queries with the newly generated trapdoor when the secret-key corresponds to an unauthorized attribute set. Since the new trapdoor is efficiently constructed from the public parameters without using the trapdoor $\td_\matB$, this avoids giving the adversary additional information correlated with the $\LWE$ samples in the hybrids where the ciphertext components are replaced by uniform.   

During the reduction, suppose that the adversary submits a secret-key query with respect to the attribute set $U$. The challenger needs to sample the preimage of $(\matA\matV_{u}+\matQ_u)\vect$ for each $u\in U$ with respect to $\matB$. We first simulate the uniform matrix $\matA=\matB\matR-\matC,\matQ_u=\matB\matR_u+M_{\rho^{-1}(u),1}(\vecy\mid\veczero\mid\cdots\mid\veczero)$ with $\matR,\matR_u\rand\{-1,1\}^{m\times m}$ for all $u\in \mathbb{U}$, which is guaranteed by the leftover hash lemma under appropriate parameters. Parsing $\matB_i$ as $[\vecb_i'\mid\matB_i']$, where $\vecb_i'\in\Z_q^{n}$, then 
\begin{align*}
    (\matA\matV_{u}+\matQ_u)\vect&=((\matB\matR-\matC)\matV_u+\matQ_u)\vect\\
   & =(\matB\matR\matV_u-\matU_u+\matB\matZ_u+\matQ_u)\vect\\
   & =\matB\matR\matV_u\vect+\matB\matZ_u\vect+\matQ_u\vect-M_{\rho^{-1}(u),1}(\vecy\mid\underbrace{\veczero\mid\cdots\mid\veczero}_{m-1})\vect-\sum_{j\in\{2,\ldots,s_{\max}\}}M_{\rho^{-1}(u),j}\matB_j\vect\\
   &=\matB(\matR\matV_u\vect+\matZ_u\vect+\matR_u\vect)-\sum_{j\in\{2,\ldots,s_{\max}\}}M_{\rho^{-1}(u),j}\vecb_j'-\sum_{j\in\{2,\ldots,s_{\max}\}}M_{\rho^{-1}(u),j}\matB_j'\hat{\vect}.
\end{align*}
Then the challenger can instead first sample short $\hat{\vect}$ and $\tilde{\veck}_u$ as: \begin{align}\label{tech:trapdoor}
    \left[\matB\Bigg| \sum_{j\in\{2,\ldots,s_{\max}\}}M_{\rho^{-1}(u),j}\matB_j'\right]\left(\begin{array}{c}
     \tilde{\veck}_u \\
     \hat{\vect}
\end{array}\right)=-\sum_{j\in\{2,\ldots,s_{\max}\}}M_{\rho^{-1}(u),j}\vecb_j'
\end{align} for each $u\in \rho([\ell])\cap U$. Then the challenger sets $$\veck_u\leftarrow \tilde{\veck}_u+\matR\matV_u\vect+\matZ_u\vect+\matR_u\vect$$ with $\vect=(1,\hat{\vect}^\top)^\top$, and hence $\{\veck_u\}$ are distributed the same as in the real game up to negligible statistical distance by the noise-smudging technique. By simulating $\matB_i:=\matB\matR_i+\matC_i$ with $\matR_i\rand\{-1,1\}$ and $\matC_i$ the commitment matrix of some matrix of specific form, we can use $\matT$ to generate a trapdoor $\td'$ for the matrix 
\begin{align}\label{eq:construct trap}
    \left[\matI_{s_{\max}-1}\otimes \matB\Bigg|\begin{array}{c}
     \matB_2'  \\
     \vdots\\
     \matB_{s_{\max}}'
\end{array}\right].
\end{align}
 Moreover, we can prove that as long as $U$ is unauthorized to decrypt the ciphertext, $\td'$ can then be used to generate the trapdoor for the linear combination of the row block matrices of \eqref{eq:construct trap}, i.e., $$(\matM_U\otimes \matI_n)\left[\matI_{s_{\max}-1}\otimes \matB\Bigg|\begin{array}{c}
     \matB_2'  \\
     \vdots\\
     \matB_{s_{\max}}'
\end{array}\right],$$ where $\matM_{U}$ are the row vectors of $\matM$ corresponding to the attributes in $U$. Therefore, without using the trapdoor $\td_{\matB}$, the challenger can sample $\{\tilde{\veck}_u\}_{u\in U}$ in \eqref{tech:trapdoor}. For details, please refer to~\ref{subsec:security}.

\section{Preliminaries}\label{sec:pre}
\iitem{Notations}
Let $\lambda\in\N$ denote the security parameter used throughout this paper. For a positive integer $n\in \N$, denote $[n]$ as the set $\{k\in\N:1\leq k\leq n\}$. For a positive integer $q\in\N$, let $\Z_q$ denote the ring of integers modulo $q$. Throughout this paper, vectors are assumed to be column vectors by default. We use bold lower-case letters (e.g., $\vecu,\vecv$) for vectors and bold upper-case letters (e.g., $\matA,\matB$) for matrices. Let $\vecv[i]$ denote the $i$-th entry of the vector $\vecv$, and let $\matU[i,j]$ denote the $(i,j)$-entry of the matrix $\matU$. 
 We write $\veczero_n$ for the all-zero vector of length $n$, and $\veczero_{n\times m}$ for the all-zero matrix of dimension $n\times m$. The \emph{infinity norm} of a vector $\vecv$ and the corresponding operator norm of a matrix $\matU$ are defined as: $$\|\vecv\|=\max_i|\vecv[i]|,\ \|\matU\|=\max_{i,j}|\matU[i,j]|.$$ For two matrices
$\matA,\matB$ of dimensions $n_1\times m_1$ and $n_2\times m_2$, respectively, their \emph{Kronecker product} is an $n_1n_2\times m_1m_2$ matrix given by $$\matA\otimes\matB=\left[\begin{array}{ccc}
     \matA[1,1]\matB&\cdots&\matA[1,m_1]\matB  \\
     \vdots&\ddots&\vdots \\
     \matA[n_1,1]\matB&\cdots&\matA[n_1,m_1]\matB
\end{array}\right].$$ 
We have $(\matA\otimes \matB)(\matC\otimes\matD)=(\matA\matC)\otimes(\matB\matD)$ if the multiplication is compatible.

\iitem{Discrete Gaussians} Let $D_{\Z,\sigma}$ represent the centered \emph{discrete Gaussian distribution} over $\Z$ with standard deviation $\sigma\in\R^{+}$ (e.g.,  \cite{banaszczyk1993new}). For a matrix $\matA\in\Z_q^{n\times m}$ and a vector $\vecv\in\Z_q^n$, let $\matA_\sigma^{-1}(\vecv)$ denote the distribution of a random variable $\vecu\leftarrow D_{\Z,\sigma}^m$ conditioned on $\matA\vecu=\vecv\mod q$. 
When $\matA_{\sigma}^{-1}$ is applied to a matrix, it is understood as being applied independently to each column of the matrix.

The following lemma (e.g., \cite{micciancio2007worst}) shows that for a given Gaussian width parameter $\sigma=\sigma(\lambda)$, the probability of a vector drawn from the discrete Gaussian distribution having norm greater than $\sqrt{n}\sigma$ is negligible.

\begin{lemma}\label{Gaussianbound}
    Let $\lambda$ be a security parameter and $\sigma=\sigma(\lambda)$ be a Gaussian width parameter. Then for all polynomials $n=n(\lambda)$, there exists a negligible function $\negl(\lambda)$ such that for all $\lambda\in\N$,
$$\Pr\left[\|\vecv\|>\sqrt{n}\sigma:\vecv\leftarrow D_{\Z,\sigma}^n\right]=\negl(\lambda). $$
\end{lemma}

\iitem{Smudging Lemma} In the following, we present the standard smudging lemma, which formalizes the intuition that sufficiently large standard deviation can ``smudge out'' small perturbations, making the resulting distributions  statistically indistinguishable.

\begin{lemma}[\cite{BDE18}]\label{smudge}
    Let $\lambda$ be a security parameter, and let $e\in\Z$ satisfy $|e|\leq B$. Suppose that $\sigma\geq B\cdot\lambda^{\omega(1)}$. Then, the following two distributions are statistically indistinguishable:
    $$ \{z: z\leftarrow D_{\Z,\sigma}\}\, \text{ and }\,
    \{z+e: z\leftarrow D_{\Z,\sigma}\}.  $$
\end{lemma}

\iitem{The Gadget Matrix}
Here we recall the definition of the \emph{gadget matrix} introduced in \cite{MP12}. For positive integers $n,q\in\N$, let $\vecg^{\top}=(1,2,\ldots,2^{\lceil\log q\rceil-1})$ denote the \emph{gadget vector}. Define the \emph{gadget matrix} $\matG_n=\matI_n\otimes \vecg^{\top}\in\Z_q^{n\times m}$, where $m=n\lceil \log q\rceil$.

\iitem{Trapdoors for lattices}
 In this work, we adopt the trapdoor frameworks outlined in~\cite{Ajt96,GPV08,MP12}, following the formalization in \cite{BTVW17}.
\begin{lemma}[\cite{GPV08,MP12}]\label{trapdoor}
    Let $n,m,q$ be lattice parameters. Then there exist two efficient algorithms (\trapgen, \samplepre) with the following syntax:
    \begin{itemize}
        \item $\trapgen(1^n,1^m,q)\rightarrow (\matA,\matT)$: On input the lattice dimension $n$,  number of samples $m$, modulus $q$, this randomized algorithm outputs a matrix $\matA\rand\Z_q^{n\times m}$ together with a trapdoor $\matT\in \Z_q^{m\times m}$.
        \item $\samplepre(\matA,\matT,\vecy,\sigma)\rightarrow \vecx$: On input ($\matA,\matT$) from \trapgen, a target vector $\vecy\in\Z_q^n$ and a Gaussian width parameter $\sigma$, this randomized algorithm outputs a vector $\vecx\in\Z^m$.
    \end{itemize}
    
    Moreover, there exists a polynomial $m_0(n,q)=O(n\log q)$ such that for all $m\geq m_0$, the above algorithms satisfy the following properties:
    \begin{itemize}
        \item \textbf{\emph{Trapdoor distribution:}} The matrix $\matA$ output by $\trapgen(1^n,1^m,q)$ is statistically close to a uniformly random matrix. Specifically, if $(\matA,\td_{\matA})\leftarrow \trapgen(1^n,1^m,q)$ and $\matA'\rand \Z_q^{n\times m}$, then the statistical distance between the distributions of $\matA$ and $\matA'$ is at most $2^{-n}$.  The trapdoor $\td_{\matA}$ serves as $\tau$-trapdoor where $\tau=O(\sqrt{n\log q})$. Denote $\chi_0=\tau\cdot\omega(\sqrt{\log n})$.
        
        \item \textbf{\emph{Preimage sampling:}} Suppse $\td_\matA$ is a $\tau$-trapdoor for $\matA$. For all $\sigma\geq \chi_0$ and all target vectors $\vecy\in\Z_q^n$, the statistical distance between the following two distributions is at most $2^{-n}$:
        $$\{\vecx \leftarrow \samplepre(\matA,\matT,\vecy,\sigma) \}\, \text{ and }\, 
       \{\vecx \leftarrow \matA_\sigma^{-1}(\vecy)\}.   $$
    \end{itemize}
\end{lemma}

\begin{corollary}[\cite{WWW22}, adapted]\label{tensortrap}
    Let $\matH\in\Z_q^{k\times t}$ be a full row rank matrix ($k\leq t$). Let $\matA\in\Z_q^{kn\times m'}$ and $\matR\in\Z_q^{m'\times mt}$. Suppose $\matA\matR=\matH\otimes \matG$. Then $(\matR,\matH)$ can be used as a $\tau$-trapdoor for $\matA$, where $\tau\leq\sqrt{kmm'}\cdot mt\|\matR\|$. 
\end{corollary}

\iitem{Succinct $\LWE$ assumption}
Our proposed $\CPABE$ scheme relies on the succinct $\LWE$ assumption \cite{Wee24}.
\begin{assumption}[$\ell$-succinct $\LWE$~\cite{Wee24}]
    Let $\lambda$ be a security parameter. Let $n,m$ be such that $m\geq 2n\log q$. The $(\ell,\sigma)$-succinct $\LWE$ says 
    $$(\matB,\vecs^\top\matB+\vece^\top,\matW,\matT)\overset{c}{\approx}(\matB,\vecc,\matW,\matT)$$ where \begin{align*}
        \matB\rand\Z_q^{n\times m},\vecs\rand\Z_q^n,\vece\leftarrow D_{\Z,\chi}^m,\vecc\rand\Z_q^m\\
        \matW\rand\Z_q^{\ell n\times m},\matT\leftarrow[\matI_{\ell}\otimes \matB\mid \matW]^{-1}_{\sigma}(\matI_{\ell}\otimes \matG)
    \end{align*}
\end{assumption}

\begin{lemma}[Leftover hash lemma,~\cite{ABB10}]\label{left}
    Let $\lambda$ be a security parameter. Let $n=n(\lambda),q=q(\lambda)$ be some lattice parameters. Let $m>(n+1)\log q+\omega(\log n)$, and $k=k(n)$ be some polynomial. Then the following two distributions are statistically indistinguishable:
    \begin{align*}
        \{(\matA,\matA\matR): \matA\rand\Z_q^{n\times m},\matR\rand\{-1,1\}^{m\times k}\}\text{ and } \{(\matA,\matS):\matA\rand\Z_q^{n\times m},\matS\rand\Z_q^{n\times k}\}.
    \end{align*}
\end{lemma}

\begin{lemma}[Leftover hash lemma with trapdoor]\label{lefttrap}
Let $\lambda$ be a security parameter. Let $n=n(\lambda),q=q(\lambda)$ be some lattice parameters. Let $m>2(n+1)\log q+\omega(\log n)$. Let $k=k(n)$ be some polyomial. Let $m>m_0,\chi,\sigma>\chi_0$, where $m_0,\chi_0$ are polynomials given in Lemma \ref{trapdoor}. Then for any adversary $\mathcal{A}$, there exists a negligible function $\negl(\lambda)$, such that for all $n\in \N$,
$$\left|\pr\left[\EXP_{\mathcal{A}}^{\lhltr,q,\sigma,\chi}(\lambda)=1\right]-1/2\right|\leq\negl(\lambda),$$ where the experiment is defined below.
\end{lemma}
    \begin{center}
    \framebox{%
\begin{minipage}[t]{0.5\linewidth}
\raggedright
\gameitem{Setup phase:}
The adversary $\mathcal{A}$ sends $1^n, 1^m$ to the challenger. The challenger flips a random bit $b\rand\{0,1\}$, and proceeds as follows:
\begin{enumerate}[label=\arabic*., itemsep=0.5em]
    \item It samples $(\matB,\td_{\matB})\leftarrow\trapgen(1^n,1^m,q)$, $\matW\rand\Z_q^{2m^2n\times m}$, $\matT\leftarrow \samplepre([\matI_{2m^2}\otimes \matB\mid \matW],\begin{bmatrix} \matI_{2m^{2}}\otimes \td_{\matB} \\ \mathbf{0} \end{bmatrix}, \matI_{2m^2}\otimes \matG,\sigma)$ and sets $\pp:=(\matB,\matW,\matT)$. It then samples $\matR\rand \{-1,1\}^{m\times k}$, and $\matS_0\leftarrow \matB\matR$. It then samples $\matS_1\rand\Z_q^{n\times k}$.
    \item It sends $(\pp,\matS_b)$ to the adversary $\mathcal{A}$.
\end{enumerate}
\end{minipage}

\begin{minipage}[t]{0.5\textwidth}
\gameitem{Query Phase:}
The adversary $\mathcal{A}$ makes polynomially many pre-image queries of the form $\vecz\in \Z_q^n$. The challenger responds to each query by sampling $\vecs\leftarrow \samplepre(\matB,\td_{\matB},\vecz,\chi)$ and sends $\vecs$ to $\mathcal{A}$.

\gameitem{Response Phase:}
The adversary outputs its guess $b'\in\{0,1\}$ for $b$. The experiments outputs 1 if and only if $b=b'$.
\end{minipage}
}
\end{center}
\begin{center}
    Experiment $\EXP_{\mathcal{A}}^{\lhltr,q,\sigma,\chi}$
\end{center}

\begin{proof}
    The proof proceeds via a sequence of hybrid experiments.
We define games $\H_0,\ldots,\H_6$ and show that each
consecutive pair of hybrids is statistically indistinguishable.
The detailed arguments are given in Lemmas~\ref{leftlem1}--\ref{leftlem6}
below.
    
    \gameitem{Game $\H_0$} This corresponds to the original game with $b=0$.
    \begin{enumerate}
        \item \textbf{Setup Phase:} The challenger receives $1^n$ and $1^m$. 
        \begin{enumerate}
            \item It samples $(\matB,\td_{\matB})\leftarrow\trapgen(1^n,1^m,q)$.
            \item It samples $\matW\rand\Z_q^{2m^2n\times m}$.
            \item $\matT\leftarrow \samplepre([\matI_{2m^2}\otimes \matB\mid \matW],\begin{bmatrix} \matI_{2m^{2}}\otimes \td_{\matB} \\ \mathbf{0} \end{bmatrix}, \matI_{2m^2}\otimes \matG,\sigma)$ and sets $\pp:=(\matB,\matW,\matT)$.
            \item It then samples $\matR\rand \{-1,1\}^{m\times k}$, and $\matS\leftarrow \matB\matR$. 
    \item It sends $(\pp,\matS)$ to the adversary $\mathcal{A}$.
        \end{enumerate}
        \item \textbf{Query Phase:} The adversary makes polynomially many queries $\vecz\in\Z_q^n$. The challenger responds to each query as follows:
        \begin{enumerate}
            \item It samples $\vecs\leftarrow \samplepre(\matB,\td_{\matB},\vecz,\chi)$ and sends $\vecs$ to $\mathcal{A}$.
        \end{enumerate}
    \item The adversary outputs a bit $b'$.
    \end{enumerate}

    \gameitem{Game $\H_1$} This game is defined identically to Game $\H_0$ except for the way to generate
the matrices $\matB$, $\matT$ and the responses to the pre-image queries.
    \begin{enumerate}
        \item \textbf{Setup Phase:} The challenger receives $1^n$ and $1^m$. 
        \begin{enumerate}
        \item It samples $\boxed{(\matB_1,\td_{\matB_1})\leftarrow\trapgen(1^n,1^{m/2},q), (\matB_2,\td_{\matB_2})\leftarrow\trapgen(1^n,1^{m/2},q)}$ and sets $\matB=[\matB_1\mid \matB_2]$.
\item It samples $\matW=\left[\begin{smallmatrix}
\matW_1\\
\vdots\\
\matW_{2m^2}
\end{smallmatrix}\right]\rand\Z_q^{2m^2n\times m}$, where $\matW_i\in \Z_q^{n\times m}$ for each $i\in [2m^2]$.
\item  It samples $\matT_{2m^2+1}\leftarrow D_{\Z,\sigma}^{m\times 2m^3}$. For each $i\in [2m^2]$, it samples $\matT_i^{\up}\leftarrow D_{\Z,\sigma}^{m/2\times 2m^3}$, and $\matT_i^{\down}\leftarrow \samplepre(\matB_2,\td_{\matB_2},\vece_i^\top\otimes\matG-\matW_i\matT_{2m^2+1}-\matB_1\matT_i^{\up},\sigma)$, where $\vece_i\in \{0,1\}^{2m^2}$ represents the vector with the $i$-th entry 1 and the other entries 0. It sets $\matT_i=\left[\begin{smallmatrix}
                \matT_i^{\up}\\
                \matT_i^{\down}
            \end{smallmatrix}\right]$ and $\boxed{\matT=\left[\begin{smallmatrix}
                \matT_1\\
                \vdots\\
                \matT_{2m^2}\\
                \matT_{2m^2+1}
            \end{smallmatrix}\right]}$, and sets $\pp:=(\matB,\matW,\matT)$.
            \item It then samples $\matR_1,\matR_2\rand \{-1,1\}^{m/2\times k}$, and $\boxed{\matS\leftarrow \matB_1\matR_1+\matB_2\matR_2}$. 
             \item It sends $(\pp,\matS)$ to the adversary.
        \end{enumerate}
        \item \textbf{Query Phase:} The adversary makes polynomially many queries $\vecz\in\Z_q^n$. The challenger responds to each query as follows:
        \begin{enumerate}
        \item It samples $\vecs_1\leftarrow \chi^{m/2}$, and samples $\vecs_2\leftarrow \samplepre(\matB_2,\td_{\matB_2},\vecz-\matB_1\vecs_1,\chi)$ and sends $\boxed{\vecs=\left[\begin{smallmatrix}
            \vecs_1\\
            \vecs_2
        \end{smallmatrix}\right]}$ to $\mathcal{A}$.
        \end{enumerate}
        \item The adversary outputs a bit $b'$.
    \end{enumerate}

     \gameitem{Game $\H_2$} This game is defined identically to Game $\H_1$ except that the challenger 
samples the matrix $\matB_1$ uniformly at random instead of using $\trapgen$.
     \begin{enumerate}
        \item \textbf{Setup Phase:} The challenger receives $1^n$ and $1^m$. 
        \begin{enumerate}
            \item It samples $\boxed{\matB_1\rand\Z_q^{n\times m/2}}$, $(\matB_2,\td_{\matB_2})\leftarrow\trapgen(1^n,1^{m/2},q)$ and sets $\matB=[\matB_1\mid \matB_2]$.
    \end{enumerate}
    \end{enumerate}

     \gameitem{Game $\H_3$} This game is defined identically to Game $\H_2$ except that we replace the term $\matB_1\matR_1$ with a uniformly random matrix $\matS'$.
     \begin{enumerate}
        \item \textbf{Setup Phase:} The challenger receives $1^n$ and $1^m$. 
        \begin{enumerate}
            \item [d)] It then samples $\matS'\rand \Z_q^{n\times k}$ and $\matR_2\rand \{-1,1\}^{m/2\times k}$, and $\boxed{\matS\leftarrow \matS'+\matB_2\matR_2}$. 
        \end{enumerate}
    \end{enumerate}

\gameitem{Game $\H_4$} This game is defined identically to Game $\H_3$ except that we
do not add the term $\matB_2\matR_2$ to $\matS$, but instead sample $\matS$
uniformly at random.
     \begin{enumerate}
        \item \textbf{Setup Phase:} The challenger receives $1^n$ and $1^m$. 
        \begin{enumerate}
            \item [d)] It then samples $\boxed{\matS\rand \Z_q^{n\times k}}$. 
\end{enumerate}
    \end{enumerate}

\gameitem{Game $\H_5$} This game is defined identically to Game $\H_4$ except that the challenger samples $\matB_1$ using the algorithm $\trapgen$ instead of sampling it uniformly at random.
     \begin{enumerate}
        \item \textbf{Setup Phase:} The challenger receives $1^n$ and $1^m$. 
        \begin{enumerate}
            \item It samples $\boxed{(\matB_1,\td_{\matB_1})\leftarrow\trapgen(1^n,1^{m/2},q)}$, $(\matB_2,\td_{\matB_2})\leftarrow\trapgen(1^n,1^{m/2},q)$ and sets $\matB=[\matB_1\mid \matB_2]$.
        \end{enumerate}
    \end{enumerate}

     \gameitem{Game $\H_6$} This game is defined identically to Game $\H_5$ except that we
sample $(\matB,\td_{\matB})$ and $\matT$ as in Game $\H_0$, and
answer preimage queries directly using $\td_{\matB}$. This game corresponds to the original game with $b=1$.
     \begin{enumerate}
        \item \textbf{Setup Phase:} The challenger receives $1^n$ and $1^m$. 
        \begin{enumerate}
            \item [a)] It samples $\boxed{(\matB,\td_{\matB})\leftarrow\trapgen(1^n,1^{m},q)}$.
           \item [c)] $\boxed{\matT\leftarrow \samplepre([\matI_{2m^2}\otimes \matB\mid \matW],\begin{bmatrix} \matI_{2m^{2}}\otimes \td_{\matB} \\ \mathbf{0} \end{bmatrix}, \matI_{2m^2}\otimes \matG,\sigma)}$ and sets $\pp:=(\matB,\matW,\matT)$.
        \end{enumerate}
        \item \textbf{Query Phase:} The adversary makes polynomially many queries $\vecz\in\Z_q^n$. The challenger responds to each query as follows:
        \begin{enumerate}
            \item It samples $\boxed{\vecs\leftarrow \samplepre(\matB,\td_{\matB},\vecz,\chi)}$ and sends $\vecs$ to $\mathcal{A}$.
        \end{enumerate}
    \end{enumerate}

\begin{lemma}\label{leftlem0}
    Suppose that $\chi,\sigma\geq \chi_0$, where $\chi_0$ is the polynomial given in Lemma \ref{trapdoor}. We have $\H_0\overset{s}{\approx}\H_1$.
\end{lemma}

\begin{proof}
First we consider the matrix $\matB$ in two games. In Game~$\H_0$, the challenger samples
$(\matB,\td_{\matB})\leftarrow\trapgen(1^n,1^m,q)$.
In Game~$\H_1$, it samples
$(\matB_1,\td_{\matB_1}),(\matB_2,\td_{\matB_2})\leftarrow\trapgen(1^n,1^{m/2},q)$
and sets $\matB=[\matB_1\mid\matB_2]$ instead.
Since $\trapgen$ outputs a matrix that is statistically close to a uniform one together with a
trapdoor, the resulting matrix $\matB$ in both games is statistically indistinguishable from each other.

Next we consider the matrix $\matT$. It suffices to consider each block $\matT_i$. In Game $\H_0$, the block $\matT_i$ is distributed as a discrete Gaussian distribution with parameter $\sigma$ conditioned on \begin{align}\label{eq1}
    \matB\matT_i+\matW_i\matT_{2m^2+1}=\vece_i\otimes \matG.
\end{align} In Game $\H_1$, the matrices $\matT_i^{\up}$ and $\matT_i^{\down}$ are sampled from discrete Gaussian distribution with parameter $\sigma$ satisfying $\matB_2\matT_i^{\down}=\vece_i\otimes \matG-\matW_i\matT_{2m^2+1}-\matB_1\matT_i^{\up}$, which aligns with \eqref{eq1} if we write $\matB=[\matB_1\mid\matB_2]$ and $\matT_i=\left[\begin{smallmatrix}
                \matT_i^{\up}\\
                \matT_i^{\down}
            \end{smallmatrix}\right]$. 

We now consider the matrix $\matS$ in two games. In Game~$\H_0$, the challenger samples
$\matR\leftarrow\{-1,1\}^{m\times k}$ and sets $\matS=\matB\matR$.
In Game~$\H_1$, it samples
$\matR_1,\matR_2\leftarrow\{-1,1\}^{m/2\times k}$ independently and sets
$\matS=\matB_1\matR_1+\matB_2\matR_2$.
If we write
$\matR=\begin{bmatrix}\matR_1\\\matR_2\end{bmatrix}\in\{-1,1\}^{m\times k}$,
then the distribution of $\matS=\matB\matR$ in Game~$\H_1$ is identical to
that in Game~$\H_0$.

Finally, it suffices to consider the answers to preimage queries. In Game~$\H_0$, on input $\vecz\in\Z_q^n$ the challenger returns
$\vecs\leftarrow\samplepre(\matB,\td_{\matB},\vecz,\chi)$, which is a
discrete Gaussian preimage of $\vecz$ under $\matB$. 
In Game~$\H_1$, on input the same $\vecz$ it samples
$\vecs_1\leftarrow\chi^{m/2}$ and then
\[
  \vecs_2\leftarrow
  \samplepre(\matB_2,\td_{\matB_2},\vecz-\matB_1\vecs_1,\chi),
\]
and outputs $\vecs=\begin{bmatrix}\vecs_1\\\vecs_2\end{bmatrix}$.
By the preimage-sampling process and the well-sampleness of the preimage with respect to $\samplepre$, this procedure yields exactly
the same discrete Gaussian distribution over all solutions
$\vecs\in\Z_q^m$ to $\matB\vecs=\vecz\bmod q$, which is identical to that in $\H_0$.

Putting everything together, we see that the joint distribution of
$(\pp,\matS)$ and all oracle answers in Game~$\H_1$ is
statistically indistinguishable from that in Game~$\H_0$.
\end{proof}

\begin{lemma}\label{leftlem1}
  Suppose that $m\geq m_0$, where $m_0$ is the  polynomials given in Lemma \ref{trapdoor}. We have $\H_1\overset{s}{\approx}\H_2$.
\end{lemma}

\begin{proof}
This follows directly from the well-sampleness of $\trapgen$: in Game
  $\H_1$ the matrix $\matB_1$ output by $\trapgen$ is statistically
  indistinguishable from a uniform matrix in $\Z_q^{n\times m/2}$, and
  its trapdoor $\td_{\matB_1}$ is neither used nor revealed. All other
  values are generated identically in the two games.
 Moreover, all subsequent steps depend on $\matB_1$ itself or $(\matB_2,\td_{\matB_2})$ in the same way in both games. Hence the lemma holds. 
\end{proof}

\begin{lemma}\label{leftlem2}
    Suppose that $m\geq 2n\log q+\omega(\log n)$. Let $k=k(n)$ be some polynomial. We have $\H_2\overset{s}{\approx}\H_3$.
\end{lemma}

\begin{proof}
    Suppose there exists an adversary $\mathcal{A}$ that distinguishes between Hybrid $\H_2$ and $\H_3$ with non-negligible probability. We construct an adversary $\mathcal{B}$ that violates the leftover hash lemma (Lemma \ref{left}). The algorithm $\mathcal{B}$ proceeds as follows:
    \begin{itemize}
        \item The algorithm $\mathcal{B}$ receives a challenge $(\matB_1,\matS')$, where $\matB_1\rand\Z_q^{n\times m/2}$ and $\matS'$ is either sampled from $\matS'=\matB_1\matR_1$ with $\matR_1\rand\{-1,1\}^{m/2\times k}$ or $\matS'\rand\Z_q^{n\times k}$.  
      \item It samples $(\matB_2,\td_{\matB_2})\leftarrow\trapgen(1^n,1^{m/2},q)$ and sets $\matB=[\matB_1\mid \matB_2]$. 
         \item It samples $\matW=\left[\begin{smallmatrix}
                \matW_1\\
                \vdots\\
                \matW_{2m^2}
            \end{smallmatrix}\right]\rand\Z_q^{2m^2n\times m}$, where $\matW_i\in \Z_q^{n\times m}$ for each $i\in [2m^2]$.
            \item  It samples $\matT_{2m^2+1}\leftarrow D_{\Z,\sigma}^{m\times 2m^3}$. For each $i\in [2m^2]$, it samples $\matT_i^{\up}\leftarrow D_{\Z,\sigma}^{m/2\times 2m^3}$, and $\matT_i^{\down}\leftarrow \samplepre(\matB_2,\td_{\matB_2},\vece_i\otimes\matG-\matW_i\matT_{2m^2+1}-\matB_1\matT_i^{\up},\sigma)$, where $\vece_i\in \{0,1\}^{2m^2}$ represents the vector with the $i$-th entry 1 and the other entries 0. It sets $\matT_i=\left[\begin{smallmatrix}
                \matT_i^{\up}\\
                \matT_i^{\down}
            \end{smallmatrix}\right]$ and $\matT=\left[\begin{smallmatrix}
                \matT_1\\
                \vdots\\
                \matT_{2m^2}
            \end{smallmatrix}\right]$, and sets $\pp:=(\matB,\matW,\matT)$.
            \item It then samples $\matR_2\rand \{-1,1\}^{m/2\times k}$, and $\matS\leftarrow \matS'+\matB_2\matR_2$. 
             \item It sends $(\pp,\matB,\matS)$ to the adversary $\mathcal{A}$.
        \item It responds to each query $\vecz\in \Z_q^n$ given by the adversary $\mathcal{A}$ as follows:
        \begin{itemize}
            \item It samples $\vecs_1\leftarrow\chi^{m/2}$, and samples $\vecs_2\leftarrow\samplepre(\matB_2,\td_{\matB_2},\vecz-\matB_1\vecs_1,\chi)$ and answers with $\vecs=\left[\begin{smallmatrix}
                \vecs_1\\
                \vecs_2
            \end{smallmatrix}\right]$.
        \end{itemize}
        \item It outputs whatever the adversary $\mathcal{A}$ outputs. 
    \end{itemize}
    From the definition of the algorithm $\mathcal{B}$, we notice that when $\matS'=\matB_1\matR_1$ with $\matR_1\rand\{-1,1\}^{m/2\times k}$, the adversary $\mathcal{B}$ simulates the challenger in Game $\H_2$ perfectly, when $\matS'\rand\Z_q^{n\times k}$, the  adversary $\mathcal{B}$ simulates the challenger in Game $\H_3$ perfectly. Hence the adversary $\matB$ in the leftover hash lemma game has the same advantage as that of the adversary $\mathcal{A}$ of distinguish between Game $\H_2$ and $\H_3$. By the leftover hash lemma (Lemma~\ref{left}) and the assumption
  $m/2 \ge 2n\log q + \omega(\log n)$, this advantage must be
  negligible, contradicting our assumption on~$\A$.
  Hence $\H_2 \overset{s}{\approx} \H_3$.
\end{proof}

\begin{lemma}
    We have $\H_3\equiv \H_4$. 
\end{lemma}

\begin{proof}
    The difference between Hybrid $\H_3$ and Hybrid $\H_4$ is merely syntactic. In Game $\H_3$, $\matS' \rand \Z_q^{n\times k}$ is uniform and
  $\matB_2\matR_2$ is independent of $S'$, so
  $\matS = \matS' + \matB_2\matR_2$ is also uniform over $\Z_q^{n\times k}$.
  In Game $\H_4$, $\matS$ is sampled directly and uniformly from
  $\Z_q^{n\times k}$. Thus $\matS$ has exactly the same distribution in
  both games, and all other variables are generated identically.
  Hence $\H_3 \equiv \H_4$.
\end{proof}

\begin{lemma}
    Suppose that $m\geq m_0$, where $m_0$ is the  polynomials given in Lemma \ref{trapdoor}. We have $\H_4\overset{s}{\approx}\H_5$.
\end{lemma}

\begin{proof}
    The lemma follows the same reason as Lemma \ref{leftlem1}.
\end{proof}

\begin{lemma}\label{leftlem6}
    Suppose that $\chi,\sigma\geq \chi_0$, where $\chi_0$ is the polynomial given in Lemma \ref{trapdoor}. We have $\H_5\overset{s}{\approx}\H_6$.
\end{lemma}

\begin{proof}
    The lemma follows the same reason as Lemma \ref{leftlem0}.
\end{proof}

By combining Lemmas \ref{leftlem0}--\ref{leftlem6}, we can conclude that $\H_0\overset{s}{\approx}\H_6$, which completes the proof of the theorem.
\end{proof}

\subsection{Matrix commitment}
Our scheme relies on the matrix commitment presented by \cite{Wee25}. We refer to the formal description of the algorithm to \cite{Wee25}. Roughly, the matrix commitment scheme first gives the public parameters $\pp:=(\matB,\matW,\matT)$ which are sampled according to the same procedure as in the succinct $\LWE$ assumption. It then commits to a matrix $\matM\in\Z_q^{n\times L}$ and the commitment is a matrix $\matC\in \Z_q^{n\times m}$. The verification matrix $\matV_{L}\in \Z_q^{n\times L}$ only depends on the width of the matrix $\matM$. The opening matrix $\matZ\in \Z_q^{n\times L}$ is such that $$\matC\cdot\matV_{L}=\matM-\matB\cdot\matZ.$$
\begin{lemma}[Matrix Commitment \cite{Wee25}]\label{matcommit} Let $n,m,q$ be lattice parameters with $m\geq 2n\log q$. There exists a tuple of deterministic efficient algorithms $(\Com,\ver,\Open)$ with the following syntax:
\begin{itemize}
    \item $\Com(\pp,\matM)\rightarrow\matC$: On input the public parameter $\pp$, and a matrix $\matM\in \Z_q^{n\times L}$, the algorithm outputs a matrix $\matC\in\Z_q^{n\times m}$.
    \item $\ver(\pp,1^L)\rightarrow\matV_L$: On input the public parameters $\pp$ and the length parameter $L$, the algorithm outputs a matrix $\matV_{L}\in\Z_q^{m\times L}$.
    \item $\Open(\pp,\matM)\rightarrow\matZ$: On input the public parameters $\pp$, and a matrix $\matM\in \Z_q^{n\times L}$, the algorithm outputs a matrix $\matZ\in \Z_q^{m\times L}$.
\end{itemize}
For all $\pp=(\matB,\matW,\matT)$, where $\matB\in\Z_q^{n\times m},\matW\in\Z_q^{2m^2n\times m},\matT\in \Z_q^{(2m^2+1)m\times 2m^3}$ are such that $[\matI_{2m^2}\otimes \matB\mid \matW]\cdot\matT=\matI_{2m^2}\otimes \matG$, all $L\in\mathbb{N}$, all matrices $\matM\in \Z_q^{n\times L}$, the matrices $\matC\leftarrow\Com(\pp,\matM),\matV_L\leftarrow\ver(\pp,1^L),\matZ\leftarrow\Open(\pp,\matM)$ satisfy 
\begin{align*}
    &\matC\cdot\matV_L=\matM-\matB\cdot\matZ,\\
    &\|\matV_{L}\|\leq O(\|\matT\|^4\cdot m^4\log q), \quad\|\matZ\|\leq O(\|\matT\|\cdot m^7\log q\log L).
\end{align*}    
\end{lemma}

\section{$\LSSS$ and $\CPABE$ for $\LSSS$-realizable Access Structures}\label{sec:CPABEnotion}
In this section, we first introduce the notion of Linear Secret Sharing Scheme ($\LSSS$) which we use to describe the access structure in our construction, and the ciphertext-policy attribute-based encryption ($\CPABE$).

\begin{definition}[Access Structure,~\cite{Bei96}]
Let $S$ be a set and $2^S$ be the power set of $S$, i.e., the collection of all subsets of $S$. An \emph{access structure} on $S$ is a set $\mathbb{A}\subseteq 2^S\setminus\{\varnothing\}$, consisting of some non-empty subsets of $S$. A subset $A\in 2^{S}$ is called \emph{authorized} if $A\in\mathbb{A}$, and \emph{unauthorized} otherwise. 

An access structure is called \emph{monotone} if it satisfies the following condition: for all subsets $B,C\in 2^{S}$, if $B\in\mathbb{A}$ and $B\subseteq C$, then $C\in\mathbb{A}$. In other words, adding more elements to an authorized subset does not invalidate its authorization.
\end{definition}

In the following, we present a result given by \cite{DKW21a}, which shows that there exists a non-monotone $\{0,1\}$-$\LSSS$ for access structure represented by a Boolean formula.

\begin{lemma}[\cite{DKW21a}]
     For any access structure $\mathbb{A}$ which is described by a Boolean formula, there is a deterministic polynomial time algorithm for generating $(\matM,\rho)$, where $\matM\in\{-1,0,1\}^{\ell\times d}$ and $\rho:[\ell]\to[2n]$ satisfying that  
    \begin{enumerate}
        \item $(\matM,\rho)$ yields a non-monotone $\{0,1\}$-$\LSSS$ for $\mathbb{A}$, namely
        \begin{itemize}
            \item For $S\subseteq [n]$, let $\hat{S}=S\cup\{i\in [n+1,\ldots,2n]\mid i-n\notin S\}\subseteq [2n]$. Let $\matM_{\hat{S}}$ be the submatrix that consists of all the rows of $\matM$ with row indices in $\rho^{-1}(\hat{S})$. For any authorized set of parties $S\subseteq [n]$, there is a linear combination of the rows of $\matM_{\hat{S}}$ that gives $(1,0,\ldots,0)\in \Z_q^d$. Moreover, the coefficients in this linear combination are from $\{0,1\}$.
            \item For any unauthorized set of parties $S\subseteq [n]$, there is no linear combination of the rows of $\matM_{\hat{S}}$ that gives $(1,0,\ldots,0)\in \Z_q^d$. Moreover, there exists a vector $\vecd\in\Z_q^d$, such that its first component $d_1=1$ and $\matM_{\hat{S}}\vecd=\veczero$.
        \end{itemize}
        \item For any unauthorized set of parties $S\subseteq [n]$, all of the rows of $\matM_{\hat{S}}$ are linearly independent.
    \end{enumerate}
\end{lemma}

Here we present the notion of ciphertext-policy attribute-based encryption ($\CPABE$) for $\LSSS$.

\begin{definition}[$\CPABE$ for $\LSSS$]
 Let $\lambda$ be the security parameter, and $n=n(\lambda),q=q(\lambda)$ be lattice parameters. A \emph{ciphertext-policy attribute-based encryption ($\CPABE$) scheme} for access structures captured by a \emph{linear secret sharing scheme (\LSSS)} over some finite field $\Z_q$ is defined as four efficient algorithms $\Pi_{\CPABE}=(\setup,\keygen,\enc,\dec)$. These algorithms proceed as follows: 
\begin{itemize}
\item $\setup(1^\lambda,\mathbb{U})\rightarrow (\pk,\msk)$: The setup algorithm takes as input the security parameter $\lambda$  and an attribute universe $\mathbb{U}$, and outputs the public parameters $\pk$ and a master secret key $\msk$. We assume that $\pk$ specifies the attribute universe $\mathbb{U}$.
\item $\keygen(\msk,U)\rightarrow \sk_{U}$: The key generation algorithm takes as input the master secret key $\msk$ and an attribute set $U\subseteq \mathbb{U}$. It outputs a secret key $\sk_{U}$ associated with attribute set $U$. Without loss of generality, we assume that the secret key $\sk_{U}$ implicitly contains the attribute set $U$.
\item $\enc(\pk,\msg,(\matM,\rho))\rightarrow\ct$: The encryption algorithm takes as input the public parameters $\pk$, a message $\msg$, and an $\LSSS$ access policy $(\matM,\rho)$ such that $\matM$ is a matrix over $\Z_q$ and $\rho$ is a row-labeling function that assigns each row of $\matM$ to an attribute in $\mathbb{U}$. It outputs a ciphertext $\ct$.
\item $\dec(\pk,\sk,(\matM,\rho),\ct)\rightarrow \msg'/\perp$: The decryption algorithm takes as input the public parameters $\pk$, a secret key $\sk$, an access structure $(\matM,\rho)$, and a ciphertext $\ct$.  It outputs a value $\msg'\in\Z_q$ when the attributes in $U$ satisfy the $\LSSS$ structure $(\matM,\rho)$, i.e., the vector $(1,0,\ldots,0)$ lies in the span of the rows of $\matM$ whose labels are in $U$. Otherwise, the decryption outputs $\perp$, indicating the failure of the decryption.
     \end{itemize}
\end{definition}

\iitem{Correctness} A $\CPABE$ scheme for $\LSSS$-realizable access structures is said to be \emph{correct} if for every $\lambda\in \N$, every attribute universe $\mathbb{U}$, every message $\msg$, every $\LSSS$ access policy $(\matM,\rho)$, and every attribute set $U\subseteq \mathbb{U}$ that satisfies the access policy $(\matM,\rho)$, it holds that 
$$\pr\left[\msg'=\msg\left|
\begin{array}{c}(\pk,\msk)\leftarrow\setup(1^{\lambda},\mathbb{U});\\ \sk_{U}\leftarrow\keygen(\msk,U);\\
\ct\leftarrow\enc(\pk,\msg,(\matM,\rho));\\
\msg'\leftarrow\dec(\pk,\sk_{U},(\matM,\rho),\ct)
\end{array}\right.\right]\geq 1-\negl(\lambda).$$

\iitem{Selective Security}
Here we define the selective security of a $\CPABE$ for $\LSSS$-realizable access structures by a game between a challenger and an adversary. For a security parameter $\lambda\in \N$, the game proceeds as follows:
\begin{itemize}
    \item \textbf{Setup Phase:} The adversary first sends an $\LSSS$ access policy $(\matM,\rho)$. The challenger runs algorithm $(\pk,\msk)\leftarrow\setup$ and sends the public parameters $\pk$ to the adversary.
    \item \textbf{Key Query Phase 1:} The adversary can make polynomially many secret-key queries to the challenger. Each secret-key query is of the form $U\subseteq\mathbb{U}$ such that $U$ does not satisfy the $\LSSS$ access policy $(\matM,\rho)$. The challenger then runs $\sk_{U}\leftarrow\keygen(\msk,U)$ and sends $\sk_{U}$ to the adversary.
    \item \textbf{Challenge Phase:} The challenger flips a coin $b\rand\{0,1\}$ and encrypts the message $b$ with respect to the $\LSSS$ access policy $(\matM,\rho)$ as  $\ct\leftarrow \enc(\pk,b,(\matM,\rho))$. Then the challenger sends the ciphertext $\ct$ to the adversary.
    \item \textbf{Key Query Phase 2:} This phase is the same as Key Query Phase 1.
    \item \textbf{Guess Phase:} The adversary outputs a guess $b'\in\{0,1\}$ for the value of $b$.
    \end{itemize}
The \emph{advantage} of the adversary $\mathcal{A}$ in this game is defined by $$\Adv_{\mathcal{A}}^{\CPABE}(\lambda):=|\pr[b=b']-1/2|.$$

\begin{definition}[Selective security for $\CPABE$ for $\LSSS$]\label{def:selectsecurity}
    A $\CPABE$ scheme for $\LSSS$-realizable access structures is \emph{selectively secure} if for any $\ppt$ adversary $\mathcal{A}$, there exists a negligible function $\negl(\cdot)$ such that $\Adv_{\mathcal{A}}^{\CPABE}(\lambda)\leq\negl(\lambda)$ for all $\lambda\in\N$.
\end{definition}

\section{$\CPABE$ for $\LSSS$ from the $\ell$-succinct $\LWE$ Assumption}\label{sec:scheme}
In this section, we present our construction of a $\CPABE$ scheme supporting access structures represented by $\NC^1$ circuits. As discussed in \cite{DKW21a}, we achieve our goal by presenting a $\CPABE$ scheme supporting access structures realized by $\LSSS$. More precisely, we construct a $\CPABE$ scheme for $\LSSS$ access policies $(\matM,\rho)$, where every entry of $\matM$ lies in $\{-1,0,1\}$ and the reconstruction coefficients are in $\{0,1\}$. We then prove our construction achieves selective security as in Definition \ref{def:selectsecurity}, which in turn yields a $\CPABE$ schemes for access structures represented by $\NC^1$ circuits. We assume that all $\LSSS$ access policies in the scheme correspond to matrices $\matM$ with $s_{\max}$ columns and an injective row-labeling function $\rho$. The bound $s_{\max}$ can also be viewed as a bound on the circuit size of the corresponding $\NC^1$ circuits.
\begin{construction}\label{con1}
 Let $\lambda$ be the security parameter, and $n=n(\lambda),q=q(\lambda), m=m(\lambda)$ be lattice parameters. We describe the algorithms $\Pi_{\CPABE}=(\setup,\keygen,\enc,\dec)$ as follows: 
    \begin{itemize}
    \item $\setup(1^\lambda,\mathbb{U}, s_{\max})$: The global setup algorithm takes as input the security parameter $\lambda$, an attribute universe $\mathbb{U}$, and the maximum width of the $\LSSS$ matrix $s_{\max}=s_{\max}(\lambda)$. Let $N:=|\mathbb{U}|$. We fix an arbitrary ordering of $\mathbb{U}$ and henceforth identify it with the index set $[N] := \{1,2,\ldots,N\}$. With a slight abuse of notation, we write $u\in[N]$ to denote the $u$-th attribute in this ordering. First, it chooses the distributions $\sigma,\chi,\chi_1,\chi_s$. Next, it samples $(\matB,\td_{\matB})\leftarrow\trapgen(1^n,1^m,q),\matW\rand\Z_q^{2m^2n\times m},$ and $$\matT\leftarrow \samplepre([\matI_{2m^2}\otimes \matB\mid \matW],\begin{bmatrix} \matI_{2m^{2}}\otimes \td_{\matB} \\ \mathbf{0} \end{bmatrix}, \matI_{2m^2}\otimes \matG,\sigma).$$
    It then samples $\matA\rand\Z_q^{n\times m}$, $\matB_{i}\rand \Z_q^{n\times (m+1)}$ for each $i\in \{2,\ldots,s_{\max}\}$, $\matD_{u}\rand\Z_q^{n\times (m+1)}$ for each $u\in [N]$, $\matQ_u\rand\Z_q^{n\times (m+1)}$ for each $u\in [N]$, and a vector $\vecy\rand\Z_q^n$. It sets $\pp:=(\matB,\matW,\matT)$.
    It outputs the public parameters $$\pk=(\pp,n,m,q,\sigma,\chi,\chi_1,\chi_{s},\matA,\{\matB_{i}\}_{i\in \{2,\ldots,s_{\max}\}},\{\matD_u\}_{u\in [N]},\{\matQ_u\}_{u\in [N]},\vecy)$$ and a master secret key $\msk=(\pk,\td_{\matB})$.

    \item $\keygen(\msk,U)$: The key generation algorithm takes as input the master secret key $\msk$ and an attribute set $U\subseteq \mathbb{U}$. It first computes $\matV\leftarrow \ver(\pp,1^{(m+1)N})\in\Z_q^{m\times (m+1)N}$, and parses it as $\matV=[\matV_1\mid\cdots\mid \matV_{N}]$ with $\matV_u\in \Z_q^{m\times (m+1)}$ for all $u\in [N]$. Then it samples a vector $\hat{\vect}\leftarrow D_{\Z,\chi}^{m}$ and sets the vector $\vect=(1,\hat{\vect}^\top)^\top\in \Z_q^{m+1}$. For each attribute $u\in U$, the algorithm samples $\hat{\veck}_u\leftarrow D_{\Z,\chi_{s}}^m$ and $$\tilde{\veck}_u\leftarrow\samplepre(\matB,\td_{\matB},(\matA\matV_{u}+\matQ_u)\vect-\matB\hat{\veck}_u,\chi_1),$$ and sets $\veck_u=\hat{\veck}_u+\tilde{\veck}_u$. It outputs $\sk=(\{\veck_u\}_{u\in U},\vect)$.

    \item $\enc(\pk,\msg,(\matM,\rho))$: The encryption algorithm takes as input the public parameters $\pk$, a message $\msg\in \{0,1\}$, and an $\LSSS$ access policy $(\matM,\rho)$, where $\matM:=(M_{i,j})\in \{-1,0,1\}^{\ell\times s_{\max}}$ and $\rho:[\ell]\rightarrow [N]$. The function $\rho$ associates rows of $\matM$ to attributes in $\mathbb{U}$. Assume that $\rho$ is an injective function (hence $\ell\leq N=|\mathbb{U}|$). The algorithm first computes \begin{align*}
    \matU_{\rho(i)} &\leftarrow M_{i,1}(\vecy\mid\veczero\mid\cdots\mid\veczero)
       +\sum_{2\leq j\leq s_{\max}}M_{i,j}\matB_j+\matQ_{\rho(i)}
       \in \Z_q^{n\times (m+1)}, \quad &&\forall i\in [\ell],\\
    \matU_u&\leftarrow \matQ_{u}+\matD_{u}\in\Z_q^{n\times (m+1)},\quad &&\forall u\in[N]\setminus \rho([\ell]),
    \end{align*} and sets $\matU\leftarrow [\matU_1\mid \cdots\mid \matU_{N}]\in\Z_q^{n\times (m+1)N}$. It then computes
    $$\matC\leftarrow \Com(\pp,\matU)\in\Z_q^{n\times m}$$
    and 
    samples $\vecs\rand\Z_q^n$,  $\vece_{1}\leftarrow D_{\Z,\chi}^m,\vece_{2}\leftarrow D_{\Z,\chi_{s}}^{m},e_3\leftarrow D_{\Z,\chi_{s}}$. It then computes ciphertext components  $\vecc_{1}\in\Z_q^m, \vecc_{2}\in \Z_q^m$ and $c_3\in \Z_q$ as follows: \begin{align*}
\vecc_{1}^\top&\leftarrow \vecs^\top\matB+\vece_{1}^\top,\\
\vecc_{2}^\top&\leftarrow \vecs^\top(\matA+\matC)+\vece_{2}^\top,\\
c_3&\leftarrow \vecs^\top\vecy+\msg\cdot\lceil q/2\rfloor+e_3.
    \end{align*} 
Then the algorithm outputs the ciphertext $\ct=(\vecc_1,\vecc_{2},c_3)$.

\item $\dec(\pk,\sk,(\matM,\rho),\ct)$: The decryption algorithm takes as input the public parameters $\pk$, a secret key $\sk=(\{\veck_{u}\}_{u\in U}, \vect)$ for an attribute subset $U\subseteq\mathbb{U}$. If $(1,0,\ldots,0)$ is not in the span of the rows of $\matM$ associated with $U$, then the decryption fails and outputs $\perp$. Otherwise, “let $I\subseteq[\ell]$ be the set of row indices such that $I=\{i\in [\ell]\mid \rho(i)\in U\}$. First, the algorithm finds the reconstruction coefficients $\{w_i\}_{i\in I}$ such that $(1,0,\ldots,0)=\sum_{i\in I}w_i \matM_i$, where $\matM_i$ is the $i$-th row vector of $\matM$. It then computes $$\matV\leftarrow \ver(\pp,1^{(m+1)N})$$ and parses it as $\matV=[\matV_1\mid\matV_2\mid\cdots\mid \matV_{N}]$ where $\matV_i\in\Z_q^{m\times (m+1)}$. \begin{align*}
    \matU_{\rho(i)} &\leftarrow M_{i,1}(\vecy\mid\veczero\mid\cdots\mid\veczero)
       +\sum_{2\leq j\leq s_{\max}}M_{i,j}\matB_j+\matQ_{\rho(i)}
       \in \Z_q^{n\times (m+1)}, \quad &&\forall i\in [\ell],\\
    \matU_u&\leftarrow \matQ_{u}+\matD_{u}\in\Z_q^{n\times (m+1)},\quad &&\forall u\in[N]\setminus \rho([\ell]).
    \end{align*} and sets $\matU\leftarrow [\matU_1\mid \cdots\mid \matU_{N}]]\in\Z_q^{n\times (m+1)N}$. It then computes
    $$\matZ\leftarrow \Open(\pp,\matU)\in \Z_q^{m\times (m+1)N},$$
    and 
    parses $\matZ:=[\matZ_1\mid\cdots\mid \matZ_{N}]$ where $\matZ_{i}\in \Z_q^{m\times (m+1)}$.
Finally, it computes $$\mu\leftarrow c_3-\sum_{i\in I}w_i[(\vecc_2^\top\matV_{\rho(i)}+\vecc_1^\top\matZ_{\rho(i)})\vect-\vecc_1^\top \veck_{\rho(i)}]\bmod{q},$$ and it outputs $0$ if $-q/4<\mu<q/4$ and 1 otherwise.
\end{itemize}
\end{construction}

\iitem{Parameters} The parameters are selected as follows:
\begin{itemize}
    \item $q/4> N\cdot\poly(\lambda,\log N,m)\cdot \sigma\cdot (\chi+\chi_{s})$. \hfill (for correctness)
    \item $m>2n\log q+\omega(\log n)$. \hfill (for leftover hash lemma)
    \item $\sigma=\poly (m,\lambda)$. \hfill (for succinct $\LWE$)
    \item $\chi_1 > \poly(m,N,\sigma,\log q)\cdot \omega(\sqrt{\log n})$. \hfill (for preimage sampling)
    \item $\chi>\sqrt{m}\chi_1\cdot\lambda^{\omega(1)}$. \hfill(for noise-smudging/security)
    \item $\chi_{s}\geq \chi\cdot(\chi_1+\sigma)\cdot\log q\cdot\log N\cdot\poly(m)\cdot\lambda^{\omega(1)}$. \hfill(for noise-smudging/security) 
\end{itemize}
    For example, fix $0<\varepsilon<1$, where $2m^2$-succinct $\LWE$ is hard for $2^{n^{\varepsilon}}$ modulus-to-noise ratio. Specifically, we set $$n=\poly(\lambda), \quad m=n\cdot\poly(\lambda),\quad \chi=\poly(n,N,\lambda)\cdot\lambda^{\omega(1)},\quad q=\poly(n,N,\lambda)\cdot\chi\cdot\lambda^{\omega(1)}.$$
This gives the parameters of our $\CPABE$ scheme:
$$|\pk|=O(s_{\max}),\quad |\ct|=O(1),\quad|\sk|=O(N),$$ where $O(\cdot)$ hides the polynomial in $\lambda$. Note that $N$ scales with the input length of the $\NC^1$ circuit and $s_{\max}$ scales with the size of the circuit.

\subsection{Correctness}
\begin{theorem}[Correctness]
Suppose that $q> N\cdot\poly(\lambda,\log N,m)\cdot \sigma\cdot (\chi+\chi_{s})$. The Construction \ref{con1} is correct as a $\CPABE$ scheme.
\end{theorem}

\begin{proof}
    We begin the proof by analyzing the term $\sum_{i\in I}w_i[(\vecc_2^\top\matV_{\rho(i)}+\vecc_1^\top\matZ_{\rho(i)})\vect-\vecc_1^\top \veck_{\rho(i)}]$. First, by the commitment scheme, we know for all $u\in [N]$, we have $\matC\matV_{u}=\matU_{u}-\matB\matZ_u$, then \begin{align}\label{correcteq1}
    \begin{split}
        (\vecc_2^\top\matV_{\rho(i)}+\vecc_1^\top\matZ_{\rho(i)})\vect&\approx \vecs^\top(\matA\matV_{\rho(i)}+\matC\matV_{\rho(i)}+\matB\matZ_{\rho(i)})\vect\\
        &=\vecs^\top(\matA\matV_{\rho(i)}+\matU_{\rho(i)})\vect\\&=\vecs^\top (\matA\matV_{\rho(i)}+M_{i,1}(\vecy\mid\veczero\mid\cdots\mid\veczero)
       +\sum_{2\leq j\leq s_{\max}}M_{i,j}\matB_j+\matQ_{\rho(i)})\vect,
    \end{split}
    \end{align}
     where ``$\approx$" hides the error term. By the $\keygen$ algorithm, we have for all $u\in U$, we have that $\veck_{u}=\hat{\veck}_u+\tilde{\veck}_u$ and $\matB\tilde{\veck}_u=(\matA\matV_{u}+\matQ_u)\vect-\matB\hat{\veck}_u$. Therefore, for all $i\in I$, we have 
     \begin{align}\label{correcteq2}
        \begin{split}
            \vecc_1^\top\veck_{\rho(i)}\approx\vecs^\top\matB\veck_{\rho(i)}=\vecs^\top(\matA\matV_{\rho(i)}+\matQ_{\rho(i)})\vect.
        \end{split}
    \end{align} 
    Combining \eqref{correcteq1} and \eqref{correcteq2}, we obtain that \begin{align}
        \begin{split}
            \sum_{i\in I}w_i[(\vecc_2^\top\matV_{\rho(i)}+\vecc_1^\top\matZ_{\rho(i)})\vect-\vecc_1^\top \veck_{\rho(i)}]\approx\sum_{i\in I}\vecs^\top(w_iM_{i,1}(\vecy\mid\veczero\mid\cdots\mid\veczero)
       +\sum_{2\leq j\leq s_{\max}}w_iM_{i,j}\matB_j)\vect.
        \end{split}
    \end{align}
    From the $\LSSS$ access structure, if the rows of indices in $I$ satisfy the policy $(\matM,\rho)$,  we have $\sum_{i\in I}w_iM_{i,1}=1$ and $\sum_{i\in I}w_iM_{i,j}=0$ for all $j\in \{2,\ldots,s_{\max}\}$. Since $\vect=(1\mid \hat{\vect}^\top)^\top$, we have $$\sum_{i\in I}w_i[(\vecc_2^\top\matV_{\rho(i)}+\vecc_1^\top\matZ_{\rho(i)})\vect-\vecc_1^\top \veck_{\rho(i)}]\approx \vecs^\top\vecy.$$ It follows that $$c_3-\sum_{i\in I}w_i[(\vecc_2^\top\matV_{\rho(i)}+\vecc_1^\top\matZ_{\rho(i)})\vect-\vecc_1^\top \veck_{\rho(i)}]\approx \msg\cdot\lceil q/2\rfloor,$$ where the error term is given by $$e_3-\sum_{i\in I}w_i((\vece_2^\top\matV_{\rho(i)}+\vece_1^\top\matZ_{\rho(i)})\vect-\vece_1^\top\veck_{\rho(i)}).$$
    The following bounds hold except with negligible probability:
    \begin{itemize}
        \item $\|\veck_{\rho(i)}\|\leq \sqrt{m}\chi_{s}$ for all $i\in I$: Since $\tilde{\veck}_u\leftarrow D_{\Z,\chi_1}^m$ conditioned on $(\matA\matV_u+\matQ_u)\vect-\matB\hat{\veck}_u=\matB\tilde{\veck}_u$, $\hat{\veck}_u\leftarrow D_{\Z,\chi_{s}}^m$, and by the noise-smudging lemma (Lemma \ref{smudge}), we obtain that $\hat{\veck}_u+\tilde{\veck}_u\overset{s}{\approx}\hat{\veck}_u$, and hence $\|\veck_u\|\leq \sqrt{m}\chi_{s}$ for all $u\in U$.
        \item $\|\vece_1\|\leq \sqrt{m}\chi, \|\vece_2\|\leq\sqrt{m}\chi_{s},\|e_3\|\leq \chi_{s}$: It follows from the fact that $\vece_1\leftarrow D_{\Z,\chi}^m,\vece_2\leftarrow D_{\Z,\chi_{s}}^m,e_3\leftarrow D_{\Z,\chi_{s}}$ and Lemma \ref{Gaussianbound}.
        \item $\|\matV_u\|\leq O(\sigma\cdot \poly(m)\cdot\log q),\|\matZ_u\|\leq O(\sigma\cdot \poly(m)\cdot \log q\cdot \log N)$: It follows from the matrix commitment (Lemma \ref{matcommit}) and the matrix $\matT$ is sampled as $\matT\leftarrow \samplepre([\matI_{2m^2}\otimes \matB\mid \matW],\begin{bmatrix} \matI_{2m^{2}}\otimes \td_{\matB} \\ \mathbf{0} \end{bmatrix}, \matI_{2m^2}\otimes \matG,\sigma).$
        \item $\|\vect\|\leq \sqrt{m}\chi$: It follows from the fact that $\hat{\vect}\leftarrow D_{\Z,\chi}^{m}$ and $\vect\leftarrow(1\mid \hat{\vect}^\top)^\top$.
    \end{itemize}
    By combining the results above and the fact that $w_i\in \{0,1\}$ by our restriction on $\LSSS$, we have that the error term is bounded by $$N\poly(\lambda,\log N,m)\cdot\sigma\cdot(\chi+\chi_{s}). $$ Therefore, the correctness holds under the condition presented in the theorem.
\end{proof}

\subsection{Selective Security}\label{subsec:security}
\begin{theorem}
    Suppose the $2m^2$-succinct $\LWE$ assumption holds, then Construction \ref{con1} is selectively secure.
\end{theorem}

\begin{proof}
    We prove the selective security of Construction \ref{con1} by giving a sequence of hybrid games. We claim that each adjacent pair of games is indistinguishable. Game $\H_0$ corresponds to the real selective-security game (defined in \ref{def:selectsecurity}), while in the final game the challenge bit is information-theoretically hidden, hence the adversary’s advantage is 0. Therefore, the adversary’s advantage in the real game is negligible.

    Throughout the hybrid games, the adversary $\mathcal{A}$ first commits to an access policy $(\matM,\rho)$, in which $\matM=(M_{i,j})\in\{-1,0,1\}^{\ell\times s_{\max}}$ and $\rho:[\ell]\rightarrow [N]$ is an injective row-labeling function for $\matM$. Then the challenger sends the public parameters to $\mathcal{A}$. In the key query phase, the adversary can query secret keys corresponding to the attribute sets $U$ such that the row indices in $\rho^{-1}(U)$ are not authorized with respect to the access policy $(\matM,\rho)$. In the challenge phase, the challenger tosses a coin $b\rand\{0,1\}$ and provides the ciphertext of the bit $b$ under the committed policy $(\matM,\rho)$. Finally, in the guess phase, the adversary outputs a guess bit $\hat{b}$ for $b$.   
\newpage
\noindent
\iitem{Game $\H_0$}
This game corresponds to the real selective security game
for our $\CPABE$ scheme.
\begin{center}
\small
    \framebox{%
\begin{minipage}[t]{0.5\linewidth}
\raggedright
\gameitem{Setup phase:}
\begin{enumerate}[label=\arabic*., itemsep=0.5em]
    \item $(\matB,\td_{\matB})\leftarrow\trapgen(1^n,1^m,q)$.
    \item $\matW\rand\Z_q^{2m^2n\times m}$.
    \item $\matT\leftarrow \samplepre([\matI_{2m^2}\otimes \matB\mid \matW],\begin{bmatrix} \matI_{2m^{2}}\otimes \td_{\matB} \\ \mathbf{0} \end{bmatrix}, \matI_{2m^2}\otimes \matG,\sigma)$. \item $\pp:=(\matB,\matW,\matT)$.
    \item $\matA\rand\Z_q^{n\times m}$.
    \item $\vecy\rand\Z_q^n$.
    \item $\{\matB_{i}\}_{i\in \{2,\ldots,s_{\max}\}}\rand\Z_q^{n\times (m+1)}$.
    \item $\{\matD_u\}_{u\in [N]}\rand\Z_q^{n\times (m+1)}$.
    \item $\{\matQ_{u}\}_{u\in [N]}\rand \Z_q^{n\times (m+1)}$.
    
    \item $\begin{aligned}&\pk=(\pp,n,m,q,\sigma,\chi,\chi_1,\chi_{s},\\  &\matA,\{\matB_{i}\}_{i\in\{2,\ldots,s_{\max}\}},\{\matD_u\}_{u\in [N]},\{\matQ_u\}_{u\in [N]},\vecy).\end{aligned}$
\end{enumerate}

\gameitem{Secret-key query for attribute set $U$:}
\begin{enumerate}[label=\arabic*., itemsep=0.5em]
   \item $\matV=[\matV_1\mid\cdots\mid \matV_{N}]\leftarrow \ver(\pp,1^{(m+1)N})$.
    \item $\hat{\vect}\leftarrow D_{\Z,\chi}^{m}$. 
    \item $\vect=(1,\hat{\vect}^\top)^\top\in \Z_q^{m+1}$.
    \item $\{\hat{\veck}_u\}_{u\in U}\leftarrow D_{\Z,\chi_{s}}^m$.
    \item $\forall u\in U$: $\tilde{\veck}_{u}\leftarrow \samplepre(\matB,\td_{\matB},(\matA\matV_{u}+\matQ_u)\vect-\matB\hat{\veck}_u,\chi_1)$.
    \item $\forall u\in U$: $\veck_{u}\leftarrow\hat{\veck}_u+\tilde{\veck}_{u}$.
    \item $\sk\leftarrow (\{\veck_u\}_{u\in U},\vect)$.
\end{enumerate}
\end{minipage}

\begin{minipage}[t]{0.5\textwidth}
\gameitem{Challenge phase:}
\begin{enumerate}[label=\arabic*., itemsep=0.5em]
\raggedright
\item $\msg\rand\{0,1\}$.
\item $\forall i\in [\ell]:\matU_{\rho(i)}\leftarrow M_{i,1}(\vecy\mid \veczero\mid\cdots\mid\veczero)+\sum_{2\leq j\leq s_{\max}}M_{i,j}\matB_j+\matQ_{\rho(i)}$.
\item $\forall u\in [N]\setminus\rho([\ell]):\matU_u\leftarrow \matQ_{u}+\matD_u$.
\item $\matU\leftarrow[\matU_1\mid\cdots\mid\matU_{N}]$.
\item $\matC\leftarrow \Com(\pp,\matU)$.
    \item $\vecs\rand\Z_q^n$.
    \item $\vece_{1}\leftarrow D_{\Z,\chi}^m$.
    \item $\vece_{2}\leftarrow D_{\Z,\chi_{s}}^{m}$.
    \item $e_3\leftarrow D_{\Z,\chi_{s}}$.
    \item $\vecc_{1}^\top\leftarrow\vecs^\top\matB+\vece_{1}^\top$.
    \item $\vecc_{2}^\top\leftarrow \vecs^\top(\matA+\matC)+\vece_{2}^\top$.
    \item $c_3\leftarrow \vecs^\top\vecy+\msg\cdot\lceil q/2\rfloor+e_3$.
    \item $\ct\leftarrow(\vecc_{1},\vecc_{2},c_3)$.
\end{enumerate}
\end{minipage}
}
\end{center}

\newpage
\iitem{Game $\H_1$}
This game is defined identically to Game $\H_0$, except for how the matrix $\matA$ is sampled and how the matrices $\{\matU_u\}_{u\in [N]}$ and $\{\matQ_u\}_{u\in [N]}$ are generated during the setup phase. The indistinguishability between Game $\H_0$ and Game $\H_1$ (Lemma \ref{lem0}) follows from the fact that this change is merely syntactic.
\begin{center}
\small
    \framebox{%
\begin{minipage}[t]{0.5\linewidth}
\raggedright
\gameitem{Setup phase:}
\begin{enumerate}[label=\arabic*., itemsep=0.5em]
    \item $(\matB,\td_{\matB})\leftarrow\trapgen(1^n,1^m,q)$.
    \item $\matW\rand\Z_q^{2m^2n\times m}$.
    \item $\matT\leftarrow \samplepre([\matI_{2m^2}\otimes \matB\mid \matW], \begin{bmatrix} \matI_{2m^{2}}\otimes \td_{\matB} \\ \mathbf{0} \end{bmatrix}, \matI_{2m^2}\otimes \matG,\sigma)$. \item $\pp:=(\matB,\matW,\matT)$.
    \item $\boxed{\matA'\rand\Z_q^{n\times m}}$.
    \item $\vecy\rand\Z_q^n$.
    \item $\{\matB_{i}\}_{i\in \{2,\ldots,s_{\max}\}}\rand\Z_q^{n\times (m+1)}$.
    \item $\boxed{\{\matU_u\}_{u\in [N]}\rand\Z_q^{n\times (m+1)}}$.
    \item $\{\matD_u\}_{u\in [N]}\rand\Z_q^{n\times (m+1)}$.
    \item $\forall u\in \rho([\ell]):\boxed{\begin{aligned}\matQ_{u}&\leftarrow \matU_{u}-(M_{\rho^{-1}(u),1}(\vecy\mid \veczero\mid\cdots\mid\veczero)\\
    &+\sum_{2\leq j\leq s_{\max}}M_{\rho^{-1}(u),j}\matB_j)\end{aligned}}$.
    \item $\forall u\in [N]\setminus\rho([\ell]):\boxed{\matQ_u\leftarrow \matU_{u}-\matD_{u}}$.    
    \item $\matU\leftarrow[\matU_1\mid\cdots\mid\matU_{N}]$.
    \item $\matC\leftarrow \Com(\pp,\matU)$.
    \item $\boxed{\matA\leftarrow \matA'-\matC\in \Z_q^{n\times m}}$.
    \item $\begin{aligned}&\pk=(\pp,n,m,q,\sigma,\chi,\chi_1,\chi_{s},\\
    &\matA,\{\matB_{i}\}_{i\in\{2,\ldots,s_{\max}\}},\{\matD_u\}_{u\in [N]},\{\matQ_u\}_{u\in [N]},\vecy)\end{aligned}$.
\end{enumerate}
\end{minipage}

\begin{minipage}[t]{0.5\textwidth}
\gameitem{Secret-key query for attribute set $U$:}
\begin{enumerate}[label=\arabic*., itemsep=0.5em]
\raggedright
   \item $\matV=[\matV_1\mid\cdots\mid \matV_{N}]\leftarrow \ver(\pp,1^{(m+1)N})$.
    \item $\hat{\vect}\leftarrow D_{\Z,\chi}^{m}$. 
    \item $\vect=(1,\hat{\vect}^\top)^\top\in \Z_q^{m+1}$.
    \item $\{\hat{\veck}_u\}_{u\in U}\leftarrow D_{\Z,\chi_{s}}^m$.
    \item $\forall u\in U$: $\tilde{\veck}_{u}\leftarrow \samplepre(\matB,\td_{\matB},(\matA\matV_{u}+\matQ_u)\vect-\matB\hat{\veck}_u,\chi_1)$.
    \item $\forall u\in U$: $\veck_{u}\leftarrow\hat{\veck}_u+\tilde{\veck}_{u}$.
    \item $\sk\leftarrow (\{\veck_u\}_{u\in U},\vect)$.
\end{enumerate}

\gameitem{Challenge phase:}
\begin{enumerate}[label=\arabic*., itemsep=0.5em]
\raggedright
    \item $\msg\rand\{0,1\}$.
    \item $\vecs\rand\Z_q^n$.
    \item $\vece_{1}\leftarrow D_{\Z,\chi}^m$.
    \item $\vece_{2}\leftarrow D_{\Z,\chi_{s}}^{m}$.
    \item $e_3\leftarrow D_{\Z,\chi_{s}}$.
    \item $\vecc_{1}^\top\leftarrow\vecs^\top\matB+\vece_{1}^\top$.
    \item $\vecc_{2}^\top\leftarrow \vecs^\top(\matA+\matC)+\vece_{2}^\top$.
    \item $c_3\leftarrow \vecs^\top\vecy+\msg\cdot\lceil q/2\rfloor+e_3$.
    \item $\ct\leftarrow(\vecc_{1},\vecc_{2},c_3)$.
\end{enumerate}
\end{minipage}
}
\end{center}

\newpage
\iitem{Game $\H_2$} This game is defined identically to Game $\H_1$, except for the way the
challenger generates the matrices $\matA'$, $\{\matU_u\}_{u\in [N]}$ and
the vector $\vecy$ during the setup phase. The indistinguishability
between Games $\H_1$ and $\H_2$ (Lemma \ref{lem1}) follows from the
leftover hash lemma with trapdoor (Lemma~\ref{lefttrap}).
\begin{center}
\small
    \framebox{%
\begin{minipage}[t]{0.5\linewidth}
\raggedright
\gameitem{Setup phase:}
\begin{enumerate}[label=\arabic*., itemsep=0.5em]
    \item $(\matB,\td_{\matB})\leftarrow\trapgen(1^n,1^m,q)$.
    \item $\matW\rand\Z_q^{2m^2n\times m}$.
    \item $\matT\leftarrow \samplepre([\matI_{2m^2}\otimes \matB\mid \matW],\begin{bmatrix} \matI_{2m^{2}}\otimes \td_{\matB} \\ \mathbf{0} \end{bmatrix}, \matI_{2m^2}\otimes \matG,\sigma)$. \item $\pp:=(\matB,\matW,\matT)$.
        \item $\boxed{\begin{aligned}&\matR\rand\{-1,1\}^{m\times m},\{\matR_u\}_{u\in [N]}\rand\{-1,1\}^{m\times (m+1)},\\
        &\vecr\rand\{-1,1\}^{m}\end{aligned}}$.
    \item $\boxed{\matA'\leftarrow\matB\matR}$.
    \item $\boxed{\vecy\leftarrow\matB\vecr}$.
    \item $\{\matB_{i}\}_{i\in \{2,\ldots,s_{\max}\}}\rand\Z_q^{n\times (m+1)}$.
    \item $\forall u\in [N]:\boxed{\matU_{u}\leftarrow\matB\matR_u}$.
    \item $\{\matD_{u}\}_{u\in [N]}\rand \Z_q^{n\times (m+1)}$.
    \item $\forall u\in \rho([\ell]): \matQ_{u}\leftarrow \matU_{u}-(M_{\rho^{-1}(u),1}(\vecy\mid \veczero\mid\cdots\mid\veczero)+\sum_{2\leq j\leq s_{\max}}M_{\rho^{-1}(u),j}\matB_j)$.
    \item $\forall u\in [N]\setminus\rho([\ell]):\matQ_u\leftarrow \matU_{u}-\matD_{u}$. 
    \item $\matU\leftarrow[\matU_1\mid\cdots\mid\matU_{N}]$.
\item $\matC\leftarrow \Com(\pp,\matU)$.
    \item $\matA\leftarrow \matA'-\matC\in \Z_q^{n\times m}$.
    \item $\begin{aligned}&\pk=(\pp,n,m,q,\sigma,\chi,\chi_1,\chi_{s},\\  &\matA,\{\matB_{i}\}_{i\in\{2,\ldots,s_{\max}\}},\{\matD_u\}_{u\in [N]},\{\matQ_u\}_{u\in [N]},\vecy).\end{aligned}$
\end{enumerate}
\end{minipage}

\begin{minipage}[t]{0.5\textwidth}
\raggedright
\gameitem{Secret-key query for attribute set $U$:}
\begin{enumerate}[label=\arabic*., itemsep=0.5em]
   \item $\matV=[\matV_1\mid\cdots\mid \matV_{N}]\leftarrow \ver(\pp,1^{(m+1)N})$.
    \item $\hat{\vect}\leftarrow D_{\Z,\chi}^{m}$. 
    \item $\vect=(1,\hat{\vect}^\top)^\top\in \Z_q^{m+1}$.
    \item $\{\hat{\veck}_u\}_{u\in U}\leftarrow D_{\Z,\chi_{s}}^m$.
    \item $\forall u\in U$: $\tilde{\veck}_{u}\leftarrow \samplepre(\matB,\td_{\matB},(\matA\matV_{u}+\matQ_u)\vect-\matB\hat{\veck}_u,\chi_1)$.
    \item $\forall u\in U$: $\veck_{u}\leftarrow\hat{\veck}_u+\tilde{\veck}_{u}$.
    \item $\sk\leftarrow (\{\veck_u\}_{u\in U},\vect)$.
\end{enumerate}

\gameitem{Challenge phase:}
\begin{enumerate}[label=\arabic*., itemsep=0.5em]
\raggedright
\item $\msg\rand\{0,1\}$.
    \item $\vecs\rand\Z_q^n$.
    \item $\vece_{1}\leftarrow D_{\Z,\chi}^m$.
    \item $\vece_{2}\leftarrow D_{\Z,\chi_{s}}^{m}$.
    \item $e_3\leftarrow D_{\Z,\chi_{s}}$.
    \item $\vecc_{1}^\top\leftarrow\vecs^\top\matB+\vece_{1}^\top$.
    \item $\vecc_{2}^\top\leftarrow \vecs^\top(\matA+\matC)+\vece_{2}^\top$.
    \item $c_3\leftarrow \vecs^\top\vecy+\msg\cdot\lceil q/2\rfloor+e_3$.
    \item $\ct\leftarrow(\vecc_{1},\vecc_{2},c_3)$.
\end{enumerate}
\end{minipage}
}
\end{center}

\newpage
\iitem{Game $\H_3$} This game is identical to Game~$\H_2$, except for how the challenger
generates the responses to secret-key queries.
The indistinguishability between Games~$\H_2$ and~$\H_3$ (Lemma~\ref{lem2})
follows from the noise-smudging lemma (Lemma~\ref{smudge}).
\begin{center}
    \small\framebox{%
\begin{minipage}[t]{0.5\linewidth}
\raggedright
\gameitem{Setup phase:}
\begin{enumerate}[label=\arabic*., itemsep=0.5em]
    \item $(\matB,\td_{\matB})\leftarrow\trapgen(1^n,1^m,q)$.
    \item $\matW\rand\Z_q^{2m^2n\times m}$.
    \item $\matT\leftarrow \samplepre([\matI_{2m^2}\otimes \matB\mid \matW],\begin{bmatrix} \matI_{2m^{2}}\otimes \td_{\matB} \\ \mathbf{0} \end{bmatrix}, \matI_{2m^2}\otimes \matG,\sigma)$. \item $\pp:=(\matB,\matW,\matT)$.
        \item $\matR\rand\{-1,1\}^{m\times m},\{\matR_u\}_{u\in [N]}\rand\{-1,1\}^{m\times (m+1)},
        \vecr\rand\{-1,1\}^{m}$.
    \item $\matA'\leftarrow\matB\matR$.
    \item $\vecy\leftarrow\matB\vecr$.
    \item $\{\matB_{i}\}_{i\in \{2,\ldots,s_{\max}\}}\rand\Z_q^{n\times (m+1)}$.
    \item $\forall i\in \{2,\ldots,s_{\max}\}$, $\matB_i:=[\vecb_i'\mid \matB_i']$ where $\vecb_i'\in\Z_q^n$, $\matB_i'\in\Z_q^{n\times m}$.
    \item $\forall u\in \rho([\ell])$: $\matN_u\leftarrow\sum_{2\leq j\leq s_{\max}}M_{\rho^{-1}(u),j}\matB'_j$.
    \item $\forall u\in [N]:\matU_{u}\leftarrow\matB\matR_u$.
    \item $\{\matD_{u}\}_{u\in [N]}\rand \Z_q^{n\times (m+1)}$.
    \item $\forall u\in [N]$, $\matD_u:=[\vecd_u'\mid \matD_u']$ where $\vecd_u'\in\Z_q^n$, $\matD_u'\in\Z_q^{n\times m}$.
    \item $\forall u\in \rho([\ell]): \matQ_{u}\leftarrow \matU_{u}-(M_{\rho^{-1}(u),1}(\vecy\mid \veczero\mid\cdots\mid\veczero)+\sum_{2\leq j\leq s_{\max}}M_{\rho^{-1}(u),j}\matB_j)$.
    \item $\forall u\in [N]\setminus\rho([\ell]):\matQ_u\leftarrow \matU_{u}-\matD_{u}$. 
    \item $\matU\leftarrow[\matU_1\mid\cdots\mid\matU_{N}]$.
\item $\matC\leftarrow \Com(\pp,\matU)$.
    \item $\matA\leftarrow \matA'-\matC\in \Z_q^{n\times m}$.
    \item $\begin{aligned}&\pk=(\pp,n,m,q,\sigma,\chi,\chi_1,\chi_{s},\\  &\matA,\{\matB_{i}\}_{i\in\{2,\ldots,s_{\max}\}},\{\matD_u\}_{u\in [N]},\{\matQ_u\}_{u\in [N]},\vecy).\end{aligned}$
\end{enumerate}
\end{minipage}

\begin{minipage}[t]{0.5\textwidth}
\gameitem{Secret-key query for attribute set $U$:}
\begin{enumerate}[label=\arabic*., itemsep=0.5em]
\raggedright
\item $\matV=[\matV_1\mid\cdots\mid \matV_{N}]\leftarrow \ver(\pp,1^{(m+1)N})$.
\item $\boxed{\matZ=[\matZ_1\mid\cdots\mid \matZ_{N}]\leftarrow \Open(\pp,\matU)}$.
    \item $\hat{\vect}\leftarrow D_{\Z,\chi}^{m}$. 
    \item $\vect=(1,\hat{\vect}^\top)^\top\in \Z_q^{m+1}$.
    \item $\{\hat{\veck}_u\}_{u\in U}\leftarrow D_{\Z,\chi_{s}}^m$.
    \item $\forall u\in U\cap\rho([\ell])$: $\boxed{\begin{aligned}\tilde{\veck}_{u}\leftarrow \samplepre(\matB,\td_{\matB},-M_{\rho^{-1}(u),1}\vecy\\-\sum_{2\leq j\leq s_{\max}}M_{\rho^{-1}(u),j}\vecb_j'
        -\matN_u\hat{\vect}-\matB\hat{\veck}_u,\chi_1)\end{aligned}}$.
    
    \item  $\forall u\in U\setminus\rho([\ell])$: $\boxed{\begin{aligned}\tilde{\veck}_{u}\leftarrow \samplepre(\matB,\td_{\matB},-\vecd_u'-\matD_{u}'\hat{\vect}-\matB\hat{\veck}_u,\chi_1)\end{aligned}}$.
    \item $\forall u\in U$: $\boxed{\veck_{u}\leftarrow\hat{\veck}_u+\tilde{\veck}_u+\matR\matV_{u}\vect+\matZ_{u}\vect}$.
    \item $\sk\leftarrow (\{\veck_u\}_{u\in U},\vect)$.
\end{enumerate}

\gameitem{Challenge phase:}
\begin{enumerate}[label=\arabic*., itemsep=0.5em]
\raggedright
   \item $\msg\rand\{0,1\}$.
    \item $\vecs\rand\Z_q^n$.
    \item $\vece_{1}\leftarrow D_{\Z,\chi}^m$.
    \item $\vece_{2}\leftarrow D_{\Z,\chi_{s}}^{m}$.
    \item $e_3\leftarrow D_{\Z,\chi_{s}}$.
    \item $\vecc_{1}^\top\leftarrow\vecs^\top\matB+\vece_{1}^\top$.
    \item $\vecc_{2}^\top\leftarrow \vecs^\top(\matA+\matC)+\vece_{2}^\top$.
    \item $c_3\leftarrow \vecs^\top\vecy+\msg\cdot\lceil q/2\rfloor+e_3$.
    \item $\ct\leftarrow(\vecc_{1},\vecc_{2},c_3)$.
\end{enumerate}
\end{minipage}
}
\end{center}

\newpage
\iitem{Game $\H_4$} This game is identical to Game~$\H_3$, except for how the challenger
generates the responses to secret-key queries.
More precisely, for each secret-key query in this game, the challenger samples
$\hat{\vect}\leftarrow D_{\Z,\chi}^{m}$ and
$\tilde{\vect}\leftarrow D_{\Z,\chi_1}^m$, and sets the last $m$ entries
of $\vect$ to be $\hat{\vect}+\tilde{\vect}$. The indistinguishability
between Game $\H_3$ and Game $\H_4$ (Lemma \ref{lem3}) follows from the noise-smudging lemma
(Lemma~\ref{smudge}).
\begin{center}
    \small\framebox{%
\begin{minipage}[t]{0.5\linewidth}
\raggedright
\gameitem{Setup phase:}
\begin{enumerate}[label=\arabic*., itemsep=0.5em]
    \item $(\matB,\td_{\matB})\leftarrow\trapgen(1^n,1^m,q)$.
    \item $\matW\rand\Z_q^{2m^2n\times m}$.
    \item $\matT\leftarrow \samplepre([\matI_{2m^2}\otimes \matB\mid \matW],\begin{bmatrix} \matI_{2m^{2}}\otimes \td_{\matB} \\ \mathbf{0} \end{bmatrix}, \matI_{2m^2}\otimes \matG,\sigma)$. \item $\pp:=(\matB,\matW,\matT)$.
        \item $\matR\rand\{-1,1\}^{m\times m},\{\matR_u\}_{u\in [N]}\rand\{-1,1\}^{m\times (m+1)},
        \vecr\rand\{-1,1\}^{m}$.
    \item $\matA'\leftarrow\matB\matR$.
    \item $\vecy\leftarrow\matB\vecr$.
    \item $\{\matB_{i}\}_{i\in \{2,\ldots,s_{\max}\}}\rand\Z_q^{n\times (m+1)}$.
    \item $\forall i\in \{2,\ldots,s_{\max}\}$, $\matB_i:=[\vecb_i'\mid \matB_i']$ where $\vecb_i'\in\Z_q^n$, $\matB_i'\in\Z_q^{n\times m}$.
    \item $\forall u\in \rho([\ell])$: $\matN_u\leftarrow\sum_{2\leq j\leq s_{\max}}M_{\rho^{-1}(u),j}\matB'_j$.
    \item $\forall u\in [N]:\matU_{u}\leftarrow\matB\matR_u$.
    \item $\{\matD_{u}\}_{u\in [N]}\rand \Z_q^{n\times (m+1)}$.
    \item $\forall u\in [N]$, $\matD_u:=[\vecd_u'\mid \matD_u']$ where $\vecd_u'\in\Z_q^n$, $\matD_u'\in\Z_q^{n\times m}$.
    \item $\forall u\in \rho([\ell]): \matQ_{u}\leftarrow \matU_{u}-(M_{\rho^{-1}(u),1}(\vecy\mid \veczero\mid\cdots\mid\veczero)+\sum_{2\leq j\leq s_{\max}}M_{\rho^{-1}(u),j}\matB_j)$.
    \item $\forall u\in [N]\setminus\rho([\ell]):\matQ_u\leftarrow \matU_{u}-\matD_{u}$. 
    \item $\matU\leftarrow[\matU_1\mid\cdots\mid\matU_{N}]$.
\item $\matC\leftarrow \Com(\pp,\matU)$.
    \item $\matA\leftarrow \matA'-\matC\in \Z_q^{n\times m}$.
    \item $\begin{aligned}&\pk=(\pp,n,m,q,\sigma,\chi,\chi_1,\chi_{s},\\  &\matA,\{\matB_{i}\}_{i\in\{2,\ldots,s_{\max}\}},\{\matD_u\}_{u\in [N]},\{\matQ_u\}_{u\in [N]},\vecy).\end{aligned}$
\end{enumerate}
\end{minipage}

\begin{minipage}[t]{0.5\textwidth}
\gameitem{Secret-key query for attribute set $U$:}
\begin{enumerate}[label=\arabic*., itemsep=0.5em]
\raggedright
\item $\matV=[\matV_1\mid\cdots\mid \matV_{N}]\leftarrow \ver(\pp,1^{(m+1)N})$.
\item $\matZ=[\matZ_1\mid\cdots\mid \matZ_{N}]\leftarrow \Open(\pp,\matU)$.
    \item $\hat{\vect}\leftarrow D_{\Z,\chi}^{m}$. 
    \item $\boxed{\tilde{\vect}\leftarrow D_{\Z,\chi_1}^m}$.
    \item $\boxed{\vect=(1,\hat{\vect}^\top+\tilde{\vect}^\top)^\top\in \Z_q^{m+1}}$.
    \item $\{\hat{\veck}_u\}_{u\in U}\leftarrow D_{\Z,\chi_{s}}^m$.
    \item $\forall u\in U\cap\rho([\ell])$: $\boxed{\begin{aligned}\tilde{\veck}_{u}&\leftarrow \samplepre(\matB,\td_{\matB},-M_{\rho^{-1}(u),1}\vecy\\
    &-\sum_{2\leq j\leq s_{\max}}M_{\rho^{-1}(u),j}\vecb_j'
        -\matN_u\hat{\vect}-\matN_u\tilde{\vect}-\matB\hat{\veck}_u,\chi_1)\end{aligned}}$.
    
    \item  $\forall u\in U\setminus\rho([\ell])$: $\boxed{\begin{aligned}\tilde{\veck}_{u}&\leftarrow \samplepre(\matB,\td_{\matB},-\vecd_u'\\
    &-\matD_{u}'\hat{\vect}-\matD_u'\tilde{\vect}-\matB\hat{\veck}_u,\chi_1)\end{aligned}}$.
    \item $\forall u\in U$: $\veck_{u}\leftarrow\hat{\veck}_u+\tilde{\veck}_u+\matR\matV_{u}\vect+\matZ_{u}\vect$.
    \item $\sk\leftarrow (\{\veck_u\}_{u\in U},\vect)$.
\end{enumerate}

\gameitem{Challenge phase:}
\begin{enumerate}[label=\arabic*., itemsep=0.5em]
\raggedright
   \item $\msg\rand\{0,1\}$.
    \item $\vecs\rand\Z_q^n$.
    \item $\vece_{1}\leftarrow D_{\Z,\chi}^m$.
    \item $\vece_{2}\leftarrow D_{\Z,\chi_{s}}^{m}$.
    \item $e_3\leftarrow D_{\Z,\chi_{s}}$.
    \item $\vecc_{1}^\top\leftarrow\vecs^\top\matB+\vece_{1}^\top$.
    \item $\vecc_{2}^\top\leftarrow \vecs^\top(\matA+\matC)+\vece_{2}^\top$.
    \item $c_3\leftarrow \vecs^\top\vecy+\msg\cdot\lceil q/2\rfloor+e_3$.
    \item $\ct\leftarrow(\vecc_{1},\vecc_{2},c_3)$.
\end{enumerate}
\end{minipage}
}
\end{center}

\newpage
\iitem{Game $\H_5$}
This game is identical to Game~$\H_4$, except for how the challenger
generates the matrices $\{\matB_i\}_{i\in \{2,\ldots,s_{\max}\}}$ and
$\{\matD_u\}_{u\in[N]}$, which are no longer sampled uniformly at random. In this game, the vector $\vecu_{i}\in\{0,1\}^{N+s_{\max}-1}$ represents the unit vector whose $i$-th entry is 1. The indistinguishability between Game $\H_4$ and Game $\H_5$ (Lemma \ref{lem4}) follows from the leftover hash lemma with trapdoor (Lemma \ref{lefttrap}).
\begin{center}
   \small \framebox{%
\begin{minipage}[t]{0.5\linewidth}
\raggedright
\gameitem{Setup phase:}
\begin{enumerate}[label=\arabic*., itemsep=0.5em]
    \item $(\matB,\td_{\matB})\leftarrow\trapgen(1^n,1^m,q)$.
    \item $\matW\rand\Z_q^{2m^2n\times m}$.
    \item $\matT\leftarrow \samplepre([\matI_{2m^2}\otimes \matB\mid \matW],\begin{bmatrix} \matI_{2m^{2}}\otimes \td_{\matB} \\ \mathbf{0} \end{bmatrix}, \matI_{2m^2}\otimes \matG,\sigma)$. \item $\pp:=(\matB,\matW,\matT)$.
        \item $\matR\rand\{-1,1\}^{m\times m},\{\matR_u\}_{u\in [N]}\rand\{-1,1\}^{m\times (m+1)},
        \vecr\rand\{-1,1\}^{m}$, $\boxed{\{\matR_i'\}_{i\in \{2,\ldots,s_{\max}\}}\rand\{-1,1\}^{m\times m}}$, $\boxed{\{\matR_u''\}_{u\in [N]}\rand\{-1,1\}^{m\times m}}$.
    \item $\matA'\leftarrow\matB\matR$.
    \item $\vecy\leftarrow\matB\vecr$.
    \item $\forall i\in\{2,\ldots,s_{\max}\}$, $\boxed{\matC_i\leftarrow \Com(\pp,\vecu_{i-1}\otimes \matG)}$.
    \item $\forall u\in [N]$, $\boxed{\matC_u'\leftarrow \Com(\pp,\vecu_{u+s_{\max}-1}\otimes \matG)}$.
    \item $\{\vecb_{i}'\}_{i\in \{2,\ldots,s_{\max}\}}\rand\Z_q^{n}$.
    \item $\forall i\in \{2,\ldots,s_{\max}\}$: $\boxed{\matB_i'\leftarrow \matB\matR_i'+\matC_i}$.
    \item $\forall i\in \{2,\ldots,s_{\max}\}$, $\matB_i:=[\vecb_i'\mid \matB_i']\in\Z_q^{n\times (m+1)}$.
    \item $\forall u\in \rho([\ell])$: $\matN_u\leftarrow\sum_{2\leq j\leq s_{\max}}M_{\rho^{-1}(u),j}\matB'_j$.
    \item $\forall u\in [N]:\matU_{u}\leftarrow\matB\matR_u$.
\item $\{\vecd_u'\}_{u\in [N]}\rand\Z_q^n$.
    \item $\forall u\in [N]$: $\boxed{\matD_{u}'\leftarrow \matB\matR_u''+\matC'_u}$.
    \item $\forall u\in [N]$, $\matD_u:=[\vecd_u'\mid \matD_u']\in\Z_q^{n\times (m+1)}$.
    \item $\forall u\in \rho([\ell]): \matQ_{u}\leftarrow \matU_{u}-(M_{\rho^{-1}(u),1}(\vecy\mid \veczero\mid\cdots\mid\veczero)+\sum_{2\leq j\leq s_{\max}}M_{\rho^{-1}(u),j}\matB_j)$.
    \item $\forall u\in [N]\setminus\rho([\ell]):\matQ_u\leftarrow \matU_{u}-\matD_{u}$. 
    \item $\matU\leftarrow[\matU_1\mid\cdots\mid\matU_{N}]$.
\item $\matC\leftarrow \Com(\pp,\matU)$.
    \item $\matA\leftarrow \matA'-\matC\in \Z_q^{n\times m}$.
    \item $\begin{aligned}&\pk=(\pp,n,m,q,\sigma,\chi,\chi_1,\chi_{s},\\  &\matA,\{\matB_{i}\}_{i\in\{2,\ldots,s_{\max}\}},\{\matD_u\}_{u\in [N]},\{\matQ_u\}_{u\in [N]},\vecy).\end{aligned}$
\end{enumerate}
\end{minipage}

\begin{minipage}[t]{0.5\textwidth}
\gameitem{Secret-key query for attribute set $U$:}
\begin{enumerate}[label=\arabic*., itemsep=0.5em]
\raggedright
\item $\matV=[\matV_1\mid\cdots\mid \matV_{N}]\leftarrow \ver(\pp,1^{(m+1)N})$.
\item $\matZ=[\matZ_1\mid\cdots\mid \matZ_{N}]\leftarrow \Open(\pp,\matU)$.
    \item $\hat{\vect}\leftarrow D_{\Z,\chi}^{m}$. 
    \item $\tilde{\vect}\leftarrow D_{\Z,\chi_1}^m$.
    \item $\vect=(1,\hat{\vect}^\top+\tilde{\vect}^\top)^\top\in \Z_q^{m+1}$.
    \item $\{\hat{\veck}_u\}_{u\in U}\leftarrow D_{\Z,\chi_{s}}^m$.
    \item $\forall u\in U\cap\rho([\ell])$: $\tilde{\veck}_{u}\leftarrow \samplepre(\matB,\td_{\matB},-M_{\rho^{-1}(u),1}\vecy-\sum_{2\leq j\leq s_{\max}}M_{\rho^{-1}(u),j}\vecb_j'
        -\matN_u\hat{\vect}-\matN_u\tilde{\vect}-\matB\hat{\veck}_u,\chi_1).$
    
    \item  $\forall u\in U\setminus\rho([\ell])$: $\tilde{\veck}_{u}\leftarrow \samplepre(\matB,\td_{\matB},-\vecd_u'-\matD_{u}'\hat{\vect}-\matD_u'\tilde{\vect}-\matB\hat{\veck}_u,\chi_1)$.
    \item $\forall u\in U$: $\veck_{u}\leftarrow\hat{\veck}_u+\tilde{\veck}_u+\matR\matV_{u}\vect+\matZ_{u}\vect$.
    \item $\sk\leftarrow (\{\veck_u\}_{u\in U},\vect)$.
\end{enumerate}

\gameitem{Challenge phase:}
\begin{enumerate}[label=\arabic*., itemsep=0.5em]
\raggedright
   \item $\msg\rand\{0,1\}$.
    \item $\vecs\rand\Z_q^n$.
    \item $\vece_{1}\leftarrow D_{\Z,\chi}^m$.
    \item $\vece_{2}\leftarrow D_{\Z,\chi_{s}}^{m}$.
    \item $e_3\leftarrow D_{\Z,\chi_{s}}$.
    \item $\vecc_{1}^\top\leftarrow\vecs^\top\matB+\vece_{1}^\top$.
    \item $\vecc_{2}^\top\leftarrow \vecs^\top(\matA+\matC)+\vece_{2}^\top$.
    \item $c_3\leftarrow \vecs^\top\vecy+\msg\cdot\lceil q/2\rfloor+e_3$.
    \item $\ct\leftarrow(\vecc_{1},\vecc_{2},c_3)$.
\end{enumerate}
\end{minipage}
}
\end{center}

\newpage
\iitem {Game $\H_6$}
This game is identical to Game~$\H_5$, except for how the challenger
generates the responses to secret-key queries. More precisely, the challenger
answers secret-key queries without using the trapdoor $\td_\matB$. Instead, in this game we use a trapdoor $\td$ for the matrix
$[\matI_{g+h}\otimes \matB \mid \matW']$ in order to invoke
$\samplepre$ in the key-generation algorithm. This trapdoor is obtained
from $(\matB,\matW,\matT)$ by the procedure described in the proof of
Lemma~\ref{lem5}, which is based on the Wee25 commitment scheme.
The indistinguishability between Games~$\H_5$ and~$\H_6$
(Lemma~\ref{lem5}) follows from the well-sampleness of the preimage
distribution for $\samplepre$ (Lemma~\ref{trapdoor}) and the
noise-smudging lemma (Lemma~\ref{smudge}).
\begin{center}
   \small \framebox{%
\begin{minipage}[t]{0.5\linewidth}
\raggedright
\gameitem{Setup phase:}
\begin{enumerate}[label=\arabic*., itemsep=0.5em]
    \item $(\matB,\td_{\matB})\leftarrow\trapgen(1^n,1^m,q)$.
    \item $\matW\rand\Z_q^{2m^2n\times m}$.

    \item $\matT\leftarrow \samplepre([\matI_{2m^2}\otimes \matB\mid \matW],\begin{bmatrix} \matI_{2m^{2}}\otimes \td_{\matB} \\ \mathbf{0} \end{bmatrix}, \matI_{2m^2}\otimes \matG,\sigma)$. \item $\pp:=(\matB,\matW,\matT)$.
        \item $\matR\rand\{-1,1\}^{m\times m},\{\matR_u\}_{u\in [N]}\rand\{-1,1\}^{m\times (m+1)},
        \vecr\rand\{-1,1\}^{m}$, $\{\matR_i'\}_{i\in \{2,\ldots,s_{\max}\}}\rand\{-1,1\}^{m\times m}$, $\{\matR_u''\}_{u\in [N]}\rand\{-1,1\}^{m\times m}$.
    \item $\matA'\leftarrow\matB\matR$.
    \item $\vecy\leftarrow\matB\vecr$.
    \item $\forall i\in\{2,\ldots,s_{\max}\}$, $\matC_i\leftarrow \Com(\pp,\vecu_{i-1}\otimes \matG)$.
    \item $\forall u\in [N]$, $\matC_u'\leftarrow \Com(\pp,\vecu_{u+s_{\max}-1}\otimes \matG)$.
   \item $\{\vecb_{i}'\}_{i\in \{2,\ldots,s_{\max}\}}\rand\Z_q^{n}$.
    \item $\forall i\in \{2,\ldots,s_{\max}\}$: $\matB_i'\leftarrow \matB\matR_i'+\matC_i$.
    \item $\forall i\in \{2,\ldots,s_{\max}\}$, $\matB_i:=[\vecb_i'\mid \matB_i']\in\Z_q^{n\times (m+1)}$.
    \item $\forall u\in \rho([\ell])$: $\matN_u\leftarrow\sum_{2\leq j\leq s_{\max}}M_{\rho^{-1}(u),j}\matB'_j$.
    \item $\forall u\in \rho([\ell])$: $\vecn_u\leftarrow M_{\rho^{-1}(u),1}\vecy+\sum_{2\leq j\leq s_{\max}}M_{\rho^{-1}(u),j}\vecb_j'$.
    \item $\forall u\in [N]:\matU_{u}\leftarrow\matB\matR_u$.
\item $\{\vecd_u'\}_{u\in [N]}\rand\Z_q^n$.
    \item $\forall u\in [N]$: $\matD_{u}'\leftarrow \matB\matR_u''+\matC'_u$.
    \item $\forall u\in [N]$, $\matD_u:=[\vecd_u'\mid \matD_u']\in\Z_q^{n\times (m+1)}$.
    \item $\forall u\in \rho([\ell]): \matQ_{u}\leftarrow \matU_{u}-(M_{\rho^{-1}(u),1}(\vecy\mid \veczero\mid\cdots\mid\veczero)+\sum_{2\leq j\leq s_{\max}}M_{\rho^{-1}(u),j}\matB_j)$.
    \item $\forall u\in [N]\setminus\rho([\ell]):\matQ_u\leftarrow \matU_{u}-\matD_{u}$. 
    \item $\matU\leftarrow[\matU_1\mid\cdots\mid\matU_{N}]$.
\item $\matC\leftarrow \Com(\pp,\matU)$.
    \item $\matA\leftarrow \matA'-\matC\in \Z_q^{n\times m}$.
    \item $\begin{aligned}&\pk=(\pp,n,m,q,\sigma,\chi,\chi_1,\chi_{s},\\  &\matA,\{\matB_{i}\}_{i\in\{2,\ldots,s_{\max}\}},\{\matD_u\}_{u\in [N]},\{\matQ_u\}_{u\in [N]},\vecy).\end{aligned}$
\end{enumerate}
\end{minipage}

\begin{minipage}[t]{0.5\textwidth}
\gameitem{Secret-key query for attribute set $U$:}
\begin{enumerate}[label=\arabic*., itemsep=0.5em]
\raggedright
\item $\matV=[\matV_1\mid\cdots\mid \matV_{N}]\leftarrow \ver(\pp,1^{(m+1)N})$.
\item $\matZ=[\matZ_1\mid\cdots\mid \matZ_{N}]\leftarrow \Open(\pp,\matU)$.
    \item $\hat{\vect}\leftarrow D_{\Z,\chi}^{m}$. 
    \item Denote $U\cap\rho([\ell])=\{u_1,\ldots,u_g\}$ and $U\setminus \rho([\ell])=\{u_1',\ldots,u_h'\}$.
    \item $\matW'\leftarrow [\matN_{u_1}^\top\mid \cdots\mid\matN_{u_g}^\top\mid\matD_{u_1'}'^\top\mid\cdots\mid\matD_{u_h'}'^\top]^\top$.
    \item $\{\hat{\veck}_u\}_{u\in U}\leftarrow D_{\Z,\chi_{s}}^m$.
     \item $\boxed{\begin{aligned}\left[\begin{smallmatrix}
        \tilde{\veck}_{u_1}\\
        \vdots\\
        \tilde{\veck}_{u_g}\\
        \tilde{\veck}_{u_1'}\\
        \vdots\\
        \tilde{\veck}_{u_h'}\\
        \tilde{\vect}
    \end{smallmatrix}\right]\leftarrow\samplepre(\left[\begin{matrix}
        \matI_{g+h}\otimes \matB\mid\matW'
    \end{matrix}\right],\td,\\
    \left[\begin{smallmatrix}
        -\vecn_{u_1}-\matN_{u_1}\hat{\vect}-\matB\hat{\veck}_{u_1}\\
        \vdots\\
        -\vecn_{u_g}
        -\matN_{u_g}\hat{\vect}-\matB\hat{\veck}_{u_g}\\
        -\vecd_{u_1'}'-\matD_{u_1'}'\hat{\vect}-\matB\hat{\veck}_{u_1'}\\
        \vdots\\
        -\vecd_{u_h'}'-\matD_{u_h'}'\hat{\vect}-\matB\hat{\veck}_{u_h'}
    \end{smallmatrix}\right],\chi_1)\end{aligned}}$.
    \item $\vect=(1,\hat{\vect}^\top+\tilde{\vect}^\top)^\top\in \Z_q^{m+1}$.
    \item $\forall u\in U$: $\veck_{u}\leftarrow\hat{\veck}_u+\tilde{\veck}_{u}+\matR\matV_{u}\vect+\matZ_u\vect$.

    \item $\sk\leftarrow (\{\veck_u\}_{u\in U},\vect)$.
\end{enumerate}

\gameitem{Challenge phase:}
\begin{enumerate}[label=\arabic*., itemsep=0.5em]
\raggedright
    \item $\msg\rand\{0,1\}$.
    \item $\vecs\rand\Z_q^n$.
    \item $\vece_{1}\leftarrow D_{\Z,\chi}^m$.
    \item $\vece_{2}\leftarrow D_{\Z,\chi_{s}}^{m}$.
    \item $e_3\leftarrow D_{\Z,\chi_{s}}$.
    \item $\vecc_{1}^\top\leftarrow\vecs^\top\matB+\vece_{1}^\top$.
    \item $\vecc_{2}^\top\leftarrow \vecs^\top(\matA+\matC)+\vece_{2}^\top$.
    \item $c_3\leftarrow \vecs^\top\vecy+\msg\cdot\lceil q/2\rfloor+e_3$.
    \item $\ct\leftarrow(\vecc_{1},\vecc_{2},c_3)$.
\end{enumerate}
\end{minipage}
}
\end{center}

\newpage
\iitem {Game $\H_7$}
This game is defined identically to Game $\H_6$ except for the way the challenger generates the challenge ciphertext components $\vecc_2$ and $c_3$. The indistinguishability between Game $\H_6$ and Game $\H_7$ (Lemma \ref{lem6}) follows from the noise-smudging lemma (Lemma \ref{smudge}).
\begin{center}
   \small \framebox{%
\begin{minipage}[t]{0.5\linewidth}
\raggedright
\gameitem{Setup phase:}
\begin{enumerate}[label=\arabic*., itemsep=0.5em]
    \item $(\matB,\td_{\matB})\leftarrow\trapgen(1^n,1^m,q)$.
    \item $\matW\rand\Z_q^{2m^2n\times m}$.

    \item $\matT\leftarrow \samplepre([\matI_{2m^2}\otimes \matB\mid \matW],\begin{bmatrix} \matI_{2m^{2}}\otimes \td_{\matB} \\ \mathbf{0} \end{bmatrix}, \matI_{2m^2}\otimes \matG,\sigma)$. \item $\pp:=(\matB,\matW,\matT)$.
        \item $\matR\rand\{-1,1\}^{m\times m},\{\matR_u\}_{u\in [N]}\rand\{-1,1\}^{m\times (m+1)},
        \vecr\rand\{-1,1\}^{m}$, $\{\matR_i'\}_{i\in \{2,\ldots,s_{\max}\}}\rand\{-1,1\}^{m\times m}$, $\{\matR_u''\}_{u\in [N]}\rand\{-1,1\}^{m\times m}$.
    \item $\matA'\leftarrow\matB\matR$.
    \item $\vecy\leftarrow\matB\vecr$.
    \item $\forall i\in\{2,\ldots,s_{\max}\}$, $\matC_i\leftarrow \Com(\pp,\vecu_{i-1}\otimes \matG)$.
    \item $\forall u\in [N]$, $\matC_u'\leftarrow \Com(\pp,\vecu_{u+s_{\max}-1}\otimes \matG)$.
   \item $\{\vecb_{i}'\}_{i\in \{2,\ldots,s_{\max}\}}\rand\Z_q^{n}$.
    \item $\forall i\in \{2,\ldots,s_{\max}\}$: $\matB_i'\leftarrow \matB\matR_i'+\matC_i$.
    \item $\forall i\in \{2,\ldots,s_{\max}\}$, $\matB_i:=[\vecb_i'\mid \matB_i']\in\Z_q^{n\times (m+1)}$.
    \item $\forall u\in \rho([\ell])$: $\matN_u\leftarrow\sum_{2\leq j\leq s_{\max}}M_{\rho^{-1}(u),j}\matB'_j$.
    \item $\forall u\in \rho([\ell])$: $\vecn_u\leftarrow M_{\rho^{-1}(u),1}\vecy+\sum_{2\leq j\leq s_{\max}}M_{\rho^{-1}(u),j}\vecb_j'$.
    \item $\forall u\in [N]:\matU_{u}\leftarrow\matB\matR_u$.
\item $\{\vecd_u'\}_{u\in [N]}\rand\Z_q^n$.
    \item $\forall u\in [N]$: $\matD_{u}'\leftarrow \matB\matR_u''+\matC'_u$.
    \item $\forall u\in [N]$, $\matD_u:=[\vecd_u'\mid \matD_u']\in\Z_q^{n\times (m+1)}$.
    \item $\forall u\in \rho([\ell]): \matQ_{u}\leftarrow \matU_{u}-(M_{\rho^{-1}(u),1}(\vecy\mid \veczero\mid\cdots\mid\veczero)+\sum_{2\leq j\leq s_{\max}}M_{\rho^{-1}(u),j}\matB_j)$.
    \item $\forall u\in [N]\setminus\rho([\ell]):\matQ_u\leftarrow \matU_{u}-\matD_{u}$. 
    \item $\matU\leftarrow[\matU_1\mid\cdots\mid\matU_{N}]$.
\item $\matC\leftarrow \Com(\pp,\matU)$.
    \item $\matA\leftarrow \matA'-\matC\in \Z_q^{n\times m}$.
    \item $\begin{aligned}&\pk=(\pp,n,m,q,\sigma,\chi,\chi_1,\chi_{s},\\  &\matA,\{\matB_{i}\}_{i\in\{2,\ldots,s_{\max}\}},\{\matD_u\}_{u\in [N]},\{\matQ_u\}_{u\in [N]},\vecy).\end{aligned}$
\end{enumerate}
\end{minipage}

\begin{minipage}[t]{0.5\textwidth}
\gameitem{Secret-key query for attribute set $U$:}
\begin{enumerate}[label=\arabic*., itemsep=0.5em]
\raggedright
\item $\matV=[\matV_1\mid\cdots\mid \matV_{N}]\leftarrow \ver(\pp,1^{(m+1)N})$.
\item $\matZ=[\matZ_1\mid\cdots\mid \matZ_{N}]\leftarrow \Open(\pp,\matU)$.
    \item $\hat{\vect}\leftarrow D_{\Z,\chi}^{m}$. 
    \item Denote $U\cap\rho([\ell])=\{u_1,\ldots,u_g\}$ and $U\setminus \rho([\ell])=\{u_1',\ldots,u_h'\}$.
    \item $\matW'\leftarrow [\matN_{u_1}^\top\mid \cdots\mid\matN_{u_g}^\top\mid\matD_{u_1'}'^\top\mid\cdots\mid\matD_{u_h'}'^\top]^\top$.
    \item $\{\hat{\veck}_u\}_{u\in U}\leftarrow D_{\Z,\chi_{s}}^m$.
     \item $\begin{aligned}\left[\begin{smallmatrix}
        \tilde{\veck}_{u_1}\\
        \vdots\\
        \tilde{\veck}_{u_g}\\
        \tilde{\veck}_{u_1'}\\
        \vdots\\
        \tilde{\veck}_{u_h'}\\
        \tilde{\vect}
    \end{smallmatrix}\right]\leftarrow\samplepre(\left[\begin{matrix}
        \matI_{g+h}\otimes \matB\mid\matW'
    \end{matrix}\right],\td,\\
    \left[\begin{smallmatrix}
        -\vecn_{u_1}-\matN_{u_1}\hat{\vect}-\matB\hat{\veck}_{u_1}\\
        \vdots\\
        -\vecn_{u_g}
        -\matN_{u_g}\hat{\vect}-\matB\hat{\veck}_{u_g}\\
        -\vecd_{u_1'}'-\matD_{u_1'}'\hat{\vect}-\matB\hat{\veck}_{u_1'}\\
        \vdots\\
        -\vecd_{u_h'}'-\matD_{u_h'}'\hat{\vect}-\matB\hat{\veck}_{u_h'}
    \end{smallmatrix}\right],\chi_1)\end{aligned}$.
    \item $\vect=(1,\hat{\vect}^\top+\tilde{\vect}^\top)^\top\in \Z_q^{m+1}$.
    \item $\forall u\in U$: $\veck_{u}\leftarrow\hat{\veck}_u+\tilde{\veck}_{u}+\matR\matV_{u}\vect+\matZ_u\vect$.

    \item $\sk\leftarrow (\{\veck_u\}_{u\in U},\vect)$.
\end{enumerate}

\gameitem{Challenge phase:}
\begin{enumerate}[label=\arabic*., itemsep=0.5em]
\raggedright
    \item $\msg\rand\{0,1\}$.
    \item $\vecs\rand\Z_q^n$.
    \item $\vece_{1}\leftarrow D_{\Z,\chi}^m$.
    \item $\vece_{2}\leftarrow D_{\Z,\chi_{s}}^{m}$.
    \item $e_3\leftarrow D_{\Z,\chi_{s}}$.
    \item $\vecc_{1}^\top\leftarrow\vecs^\top\matB+\vece_{1}^\top$.
    \item $\boxed{\vecc_{2}^\top\leftarrow \vecc_1^\top\matR+\vece_2^\top}$.
    \item $\boxed{c_3\leftarrow \vecc_1^\top\vecr+\msg\cdot\lceil q/2\rfloor+e_3}$.
\end{enumerate}
\end{minipage}
}
\end{center}

\newpage
\iitem {Game $\H_8$}
This game is defined identically to Game~$\H_7$, except for the way the
challenger generates the ciphertext component $\vecc_1$. In Game~$\H_8$,
$\vecc_1$ is generated uniformly and independently at random. The
indistinguishability between Games~$\H_7$ and~$\H_8$ (Lemma~\ref{lem7})
follows from the $(2m^2,\sigma)$-succinct $\LWE$ assumption.
\begin{center}
   \small \framebox{%
\begin{minipage}[t]{0.5\linewidth}
\raggedright
\gameitem{Setup phase:}
\begin{enumerate}[label=\arabic*., itemsep=0.5em]
    \item $(\matB,\td_{\matB})\leftarrow\trapgen(1^n,1^m,q)$.
    \item $\matW\rand\Z_q^{2m^2n\times m}$.

    \item $\matT\leftarrow \samplepre([\matI_{2m^2}\otimes \matB\mid \matW],\begin{bmatrix} \matI_{2m^{2}}\otimes \td_{\matB} \\ \mathbf{0} \end{bmatrix}, \matI_{2m^2}\otimes \matG,\sigma)$. \item $\pp:=(\matB,\matW,\matT)$.
        \item $\matR\rand\{-1,1\}^{m\times m},\{\matR_u\}_{u\in [N]}\rand\{-1,1\}^{m\times (m+1)},
        \vecr\rand\{-1,1\}^{m}$, $\{\matR_i'\}_{i\in \{2,\ldots,s_{\max}\}}\rand\{-1,1\}^{m\times m}$, $\{\matR_u''\}_{u\in [N]}\rand\{-1,1\}^{m\times m}$.
    \item $\matA'\leftarrow\matB\matR$.
    \item $\vecy\leftarrow\matB\vecr$.
    \item $\forall i\in\{2,\ldots,s_{\max}\}$, $\matC_i\leftarrow \Com(\pp,\vecu_{i-1}\otimes \matG)$.
    \item $\forall u\in [N]$, $\matC_u'\leftarrow \Com(\pp,\vecu_{u+s_{\max}-1}\otimes \matG)$.
   \item $\{\vecb_{i}'\}_{i\in \{2,\ldots,s_{\max}\}}\rand\Z_q^{n}$.
    \item $\forall i\in \{2,\ldots,s_{\max}\}$: $\matB_i'\leftarrow \matB\matR_i'+\matC_i$.
    \item $\forall i\in \{2,\ldots,s_{\max}\}$, $\matB_i:=[\vecb_i'\mid \matB_i']\in\Z_q^{n\times (m+1)}$.
    \item $\forall u\in \rho([\ell])$: $\matN_u\leftarrow\sum_{2\leq j\leq s_{\max}}M_{\rho^{-1}(u),j}\matB'_j$.
    \item $\forall u\in \rho([\ell])$: $\vecn_u\leftarrow M_{\rho^{-1}(u),1}\vecy+\sum_{2\leq j\leq s_{\max}}M_{\rho^{-1}(u),j}\vecb_j'$.
    \item $\forall u\in [N]:\matU_{u}\leftarrow\matB\matR_u$.
\item $\{\vecd_u'\}_{u\in [N]}\rand\Z_q^n$.
    \item $\forall u\in [N]$: $\matD_{u}'\leftarrow \matB\matR_u''+\matC'_u$.
    \item $\forall u\in [N]$, $\matD_u:=[\vecd_u'\mid \matD_u']\in\Z_q^{n\times (m+1)}$.
    \item $\forall u\in \rho([\ell]): \matQ_{u}\leftarrow \matU_{u}-(M_{\rho^{-1}(u),1}(\vecy\mid \veczero\mid\cdots\mid\veczero)+\sum_{2\leq j\leq s_{\max}}M_{\rho^{-1}(u),j}\matB_j)$.
    \item $\forall u\in [N]\setminus\rho([\ell]):\matQ_u\leftarrow \matU_{u}-\matD_{u}$. 
    \item $\matU\leftarrow[\matU_1\mid\cdots\mid\matU_{N}]$.
\item $\matC\leftarrow \Com(\pp,\matU)$.
    \item $\matA\leftarrow \matA'-\matC\in \Z_q^{n\times m}$.
    \item $\begin{aligned}&\pk=(\pp,n,m,q,\sigma,\chi,\chi_1,\chi_{s},\\  &\matA,\{\matB_{i}\}_{i\in\{2,\ldots,s_{\max}\}},\{\matD_u\}_{u\in [N]},\{\matQ_u\}_{u\in [N]},\vecy).\end{aligned}$
\end{enumerate}
\end{minipage}

\begin{minipage}[t]{0.5\textwidth}
\gameitem{Secret-key query for attribute set $U$:}
\begin{enumerate}[label=\arabic*., itemsep=0.5em]
\raggedright
\item $\matV=[\matV_1\mid\cdots\mid \matV_{N}]\leftarrow \ver(\pp,1^{(m+1)N})$.
\item $\matZ=[\matZ_1\mid\cdots\mid \matZ_{N}]\leftarrow \Open(\pp,\matU)$.
    \item $\hat{\vect}\leftarrow D_{\Z,\chi}^{m}$. 
    \item Denote $U\cap\rho([\ell])=\{u_1,\ldots,u_g\}$ and $U\setminus \rho([\ell])=\{u_1',\ldots,u_h'\}$.
    \item $\matW'\leftarrow [\matN_{u_1}^\top\mid \cdots\mid\matN_{u_g}^\top\mid\matD_{u_1'}'^\top\mid\cdots\mid\matD_{u_h'}'^\top]^\top$.
    \item $\{\hat{\veck}_u\}_{u\in U}\leftarrow D_{\Z,\chi_{s}}^m$.
     \item $\begin{aligned}\left[\begin{smallmatrix}
        \tilde{\veck}_{u_1}\\
        \vdots\\
        \tilde{\veck}_{u_g}\\
        \tilde{\veck}_{u_1'}\\
        \vdots\\
        \tilde{\veck}_{u_h'}\\
        \tilde{\vect}
    \end{smallmatrix}\right]\leftarrow\samplepre(\left[\begin{matrix}
        \matI_{g+h}\otimes \matB\mid\matW'
    \end{matrix}\right],\td,\\
    \left[\begin{smallmatrix}
        -\vecn_{u_1}-\matN_{u_1}\hat{\vect}-\matB\hat{\veck}_{u_1}\\
        \vdots\\
        -\vecn_{u_g}
        -\matN_{u_g}\hat{\vect}-\matB\hat{\veck}_{u_g}\\
        -\vecd_{u_1'}'-\matD_{u_1'}'\hat{\vect}-\matB\hat{\veck}_{u_1'}\\
        \vdots\\
        -\vecd_{u_h'}'-\matD_{u_h'}'\hat{\vect}-\matB\hat{\veck}_{u_h'}
    \end{smallmatrix}\right],\chi_1)\end{aligned}$.
    \item $\vect=(1,\hat{\vect}^\top+\tilde{\vect}^\top)^\top\in \Z_q^{m+1}$.
    \item $\forall u\in U$: $\veck_{u}\leftarrow\hat{\veck}_u+\tilde{\veck}_{u}+\matR\matV_{u}\vect+\matZ_u\vect$.

    \item $\sk\leftarrow (\{\veck_u\}_{u\in U},\vect)$.
\end{enumerate}

\gameitem{Challenge phase:}
\begin{enumerate}[label=\arabic*., itemsep=0.5em]
\raggedright
    \item $\msg\rand\{0,1\}$.
    \item $\vece_{2}\leftarrow D_{\Z,\chi_{s}}^{m}$.
    \item $e_3\leftarrow D_{\Z,\chi_{s}}$.
    \item $\boxed{\vecc_{1}\rand\Z_q^m}$.
    \item $\vecc_{2}^\top\leftarrow \vecc_1^\top\matR+\vece_2^\top$.
    \item $c_3\leftarrow \vecc_1^\top\vecr+\msg\cdot\lceil q/2\rfloor+e_3$.
\end{enumerate}
\end{minipage}
}
\end{center}

\newpage
\iitem {Game $\H_9$}
The game is defined identically to Game $\H_8$ except for the process for generating the ciphertext component $c_3$. In this game, the element $c_3$ is generated uniformly and independently at random instead. The advantage of the adversary $\mathcal{A}$ in this game is $0$,
since the message $\msg$ is information-theoretically hidden. The indistinguishability between Game $\H_8$ and $\H_9$ (Lemma \ref{lem8}) follows from the leftover hash lemma (Lemma \ref{left}).
\begin{center}
   \small \framebox{%
\begin{minipage}[t]{0.5\linewidth}
\raggedright
\gameitem{Setup phase:}
\begin{enumerate}[label=\arabic*., itemsep=0.5em]
    \item $(\matB,\td_{\matB})\leftarrow\trapgen(1^n,1^m,q)$.
    \item $\matW\rand\Z_q^{2m^2n\times m}$.

    \item $\matT\leftarrow \samplepre([\matI_{2m^2}\otimes \matB\mid \matW],\begin{bmatrix} \matI_{2m^{2}}\otimes \td_{\matB} \\ \mathbf{0} \end{bmatrix}, \matI_{2m^2}\otimes \matG,\sigma)$. \item $\pp:=(\matB,\matW,\matT)$.
        \item $\matR\rand\{-1,1\}^{m\times m},\{\matR_u\}_{u\in [N]}\rand\{-1,1\}^{m\times (m+1)},
        \vecr\rand\{-1,1\}^{m}$, $\{\matR_i'\}_{i\in \{2,\ldots,s_{\max}\}}\rand\{-1,1\}^{m\times m}$, $\{\matR_u''\}_{u\in [N]}\rand\{-1,1\}^{m\times m}$.
    \item $\matA'\leftarrow\matB\matR$.
    \item $\vecy\leftarrow\matB\vecr$.
    \item $\forall i\in\{2,\ldots,s_{\max}\}$, $\matC_i\leftarrow \Com(\pp,\vecu_{i-1}\otimes \matG)$.
    \item $\forall u\in [N]$, $\matC_u'\leftarrow \Com(\pp,\vecu_{u+s_{\max}-1}\otimes \matG)$.
   \item $\{\vecb_{i}'\}_{i\in \{2,\ldots,s_{\max}\}}\rand\Z_q^{n}$.
    \item $\forall i\in \{2,\ldots,s_{\max}\}$: $\matB_i'\leftarrow \matB\matR_i'+\matC_i$.
    \item $\forall i\in \{2,\ldots,s_{\max}\}$, $\matB_i:=[\vecb_i'\mid \matB_i']\in\Z_q^{n\times (m+1)}$.
    \item $\forall u\in \rho([\ell])$: $\matN_u\leftarrow\sum_{2\leq j\leq s_{\max}}M_{\rho^{-1}(u),j}\matB'_j$.
    \item $\forall u\in \rho([\ell])$: $\vecn_u\leftarrow M_{\rho^{-1}(u),1}\vecy+\sum_{2\leq j\leq s_{\max}}M_{\rho^{-1}(u),j}\vecb_j'$.
    \item $\forall u\in [N]:\matU_{u}\leftarrow\matB\matR_u$.
\item $\{\vecd_u'\}_{u\in [N]}\rand\Z_q^n$.
    \item $\forall u\in [N]$: $\matD_{u}'\leftarrow \matB\matR_u''+\matC'_u$.
    \item $\forall u\in [N]$, $\matD_u:=[\vecd_u'\mid \matD_u']\in\Z_q^{n\times (m+1)}$.
    \item $\forall u\in \rho([\ell]): \matQ_{u}\leftarrow \matU_{u}-(M_{\rho^{-1}(u),1}(\vecy\mid \veczero\mid\cdots\mid\veczero)+\sum_{2\leq j\leq s_{\max}}M_{\rho^{-1}(u),j}\matB_j)$.
    \item $\forall u\in [N]\setminus\rho([\ell]):\matQ_u\leftarrow \matU_{u}-\matD_{u}$. 
    \item $\matU\leftarrow[\matU_1\mid\cdots\mid\matU_{N}]$.
\item $\matC\leftarrow \Com(\pp,\matU)$.
    \item $\matA\leftarrow \matA'-\matC\in \Z_q^{n\times m}$.
    \item $\begin{aligned}&\pk=(\pp,n,m,q,\sigma,\chi,\chi_1,\chi_{s},\\  &\matA,\{\matB_{i}\}_{i\in\{2,\ldots,s_{\max}\}},\{\matD_u\}_{u\in [N]},\{\matQ_u\}_{u\in [N]},\vecy).\end{aligned}$
\end{enumerate}
\end{minipage}

\begin{minipage}[t]{0.5\textwidth}
\gameitem{Secret-key query for attribute set $U$:}
\begin{enumerate}[label=\arabic*., itemsep=0.5em]
\raggedright
\item $\matV=[\matV_1\mid\cdots\mid \matV_{N}]\leftarrow \ver(\pp,1^{(m+1)N})$.
\item $\matZ=[\matZ_1\mid\cdots\mid \matZ_{N}]\leftarrow \Open(\pp,\matU)$.
    \item $\hat{\vect}\leftarrow D_{\Z,\chi}^{m}$. 
    \item Denote $U\cap\rho([\ell])=\{u_1,\ldots,u_g\}$ and $U\setminus \rho([\ell])=\{u_1',\ldots,u_h'\}$.
    \item $\matW'\leftarrow [\matN_{u_1}^\top\mid \cdots\mid\matN_{u_g}^\top\mid\matD_{u_1'}'^\top\mid\cdots\mid\matD_{u_h'}'^\top]^\top$.
    \item $\{\hat{\veck}_u\}_{u\in U}\leftarrow D_{\Z,\chi_{s}}^m$.
     \item $\begin{aligned}\left[\begin{smallmatrix}
        \tilde{\veck}_{u_1}\\
        \vdots\\
        \tilde{\veck}_{u_g}\\
        \tilde{\veck}_{u_1'}\\
        \vdots\\
        \tilde{\veck}_{u_h'}\\
        \tilde{\vect}
    \end{smallmatrix}\right]\leftarrow\samplepre(\left[\begin{matrix}
        \matI_{g+h}\otimes \matB\mid\matW'
    \end{matrix}\right],\td,\\
    \left[\begin{smallmatrix}
        -\vecn_{u_1}-\matN_{u_1}\hat{\vect}-\matB\hat{\veck}_{u_1}\\
        \vdots\\
        -\vecn_{u_g}
        -\matN_{u_g}\hat{\vect}-\matB\hat{\veck}_{u_g}\\
        -\vecd_{u_1'}'-\matD_{u_1'}'\hat{\vect}-\matB\hat{\veck}_{u_1'}\\
        \vdots\\
        -\vecd_{u_h'}'-\matD_{u_h'}'\hat{\vect}-\matB\hat{\veck}_{u_h'}
    \end{smallmatrix}\right],\chi_1)\end{aligned}$.
    \item $\vect=(1,\hat{\vect}^\top+\tilde{\vect}^\top)^\top\in \Z_q^{m+1}$.
    \item $\forall u\in U$: $\veck_{u}\leftarrow\hat{\veck}_u+\tilde{\veck}_{u}+\matR\matV_{u}\vect+\matZ_u\vect$.

    \item $\sk\leftarrow (\{\veck_u\}_{u\in U},\vect)$.
\end{enumerate}

\gameitem{Challenge phase:}
\begin{enumerate}[label=\arabic*., itemsep=0.5em]
\raggedright
    \item $\vece_{2}\leftarrow D_{\Z,\chi_{s}}^{m}$.
    \item $\vecc_{1}\rand\Z_q^m$.
    \item $\vecc_{2}^\top\leftarrow \vecc_1^\top\matR+\vece_2^\top.$
    \item $\boxed{c_3\rand\Z_q.}$
\end{enumerate}
\end{minipage}
}
\end{center}

\begin{lemma}\label{lem0}
    We have $\H_0\equiv\H_1$.
\end{lemma}

\begin{proof}
    The only difference between Games $\H_0$ and $\H_1$ is the way the
    matrices $\matA$, $\{\matU_u\}_{u\in[N]}$ and
    $\{\matQ_u\}_{u\in [N]}$ are generated. In Game $\H_0$, the matrices $\matA,\{\matQ_{u}\}_{u\in [N]}$ are generated uniformly and independently at random. Then, for all $i\in [\ell]$, it sets $\matU_{\rho(i)}\leftarrow M_{i,1}(\vecy\mid \veczero\mid\cdots\mid\veczero)+\sum_{2\leq j\leq s_{\max}}M_{i,j}\matB_j+\matQ_{\rho(i)}$, and for all $u\in [N]\setminus\rho([\ell])$, it sets $\matU_u\leftarrow\matQ_{u}+\matD_{u}$. In Game $\H_1$, the matrix $\matA$ is obtained by first sampling
$\matA'\rand\Z_q^{n\times m}$ and then setting
$\matA\leftarrow \matA'-\matC$, where $\matC$ is a matrix independent of $\matA'$ ($\matC$ is not sent to the adversary, and in both games it is computed as $\Com(\pp,\matU)$ from the same distribution of $\matU$). Then the challenger instead samples $\{\matU_u\}_{u\in [N]}\rand\Z_q^{n\times (m+1)}$ and sets $\matQ_{u}\leftarrow \matU_{u}-(M_{\rho^{-1}(u),1}(\vecy\mid \veczero\mid\cdots\mid\veczero)+\sum_{2\leq j\leq s_{\max}}M_{\rho^{-1}(u),j}\matB_j)$ for all $u\in \rho([\ell])$ and $\matQ_{u}\leftarrow \matU_{u}-\matD_{u}$ for all $u\in [N]\setminus\rho([\ell])$.

In Game~$\H_1$
the matrix $\matA$ is obtained as a fixed shift of the uniform matrix
$\matA'$, so $\matA$ itself is uniform and independent of
$\{\matU_u\}$ and $\{\matQ_u\}$, exactly as in Game~$\H_0$.
For each $u\in\rho([\ell])$, the pair
$(\matQ_u,\matU_{u})$ in Game~$\H_0$ is uniformly
distributed over $\Z_q^{n\times (m+1)}\times\Z_q^{n\times (m+1)}$ subject to
\[
  \matU_{u}-\matQ_u
  = M_{\rho^{-1}(u),1}(\vecy\mid\veczero\mid\cdots\mid\veczero)
  + \sum_{2\le j\le s_{\max}} M_{\rho^{-1}(u),j}\matB_j,
\]
and exactly the same relation holds in Game~$\H_1$.  Similarly, for
each $u\in [N]\setminus\rho([\ell])$, the triple
$(\matQ_u,\matD_u,\matU_u)$ is uniformly distributed subject to
$\matU_u=\matQ_u+\matD_u$ in both games. Hence the joint distribution of
    $(\matA,\{\matQ_u\},\{\matU_u\})$ is exactly the same in
    Games~$\H_0$ and~$\H_1$.
    Observe that all the changes between Game $\H_0$ and $\H_1$ are merely syntactic. Since all subsequent uses of these matrices in the query phase and challenge phase are identical in both games, we conclude
    that $\H_0 \equiv \H_1$.
\end{proof}

\begin{lemma}\label{lem1}
    Let $q$ be a prime, and let $n,m$ be such that $m>2n\log q +\omega(\log n)$. Let $m\geq m_0(n,q)$ and $\sigma,\chi_1\geq \chi_0(n,q)$, where $m_0$ and $\chi_0$ are polynomials given in Lemma \ref{trapdoor}. We have $\H_1\overset{s}{\approx}\H_2$.
\end{lemma}

\begin{proof}
    This lemma follows from the leftover hash lemma with trapdoor (Lemma \ref{lefttrap}). Suppose there exists an adversary $\mathcal{A}$ that can distinguish between Game~$\H_1$ and Game~$\H_2$ with non-negligible advantage. Then we can construct an adversary $\mathcal{B}$ that can win the game $\EXP_{\mathcal{B}}^{\lhltr,q,\sigma,\chi_1}$ with non-negligible probability, which leads to a contradiction. The algorithm $\mathcal{B}$ proceeds as follows:
    
    \iitem{Setup Phase} The algorithm $\mathcal{B}$ receives $1^\lambda,q,\sigma,\chi_1$ from its challenger. Then it invokes $\mathcal{A}$ and receives an access policy $(\matM,\rho)$, where $\matM\in\{-1,0,1\}^{\ell\times s_{\max}}$ and $\rho: [\ell]\rightarrow [N]$ is an injective function. Then it proceeds as follows:
    \begin{enumerate}
        \item It sends $1^n,1^m,1^{s_{\max}}$ to its challenger and receives a challenge $\{\pp=(\matB,\matW,\matT),\matS\}$. It then parses $\matS:=[\matS_1\mid\cdots\mid\matS_{N}\mid \matS_0\mid \vecs']$ where $\matS_i\in \Z_q^{n\times (m+1)}$ for each $i\in [N]$, $\matS_0\in \Z_q^{n\times m}$, and $\vecs'\in \Z_q^n$.
        \item It sets $\matA'\leftarrow \matS_0$ and $\vecy\leftarrow\vecs'$.
        \item It samples $\{\matB_i\}_{i\in \{2,\ldots,s_{\max}\}}\rand\Z_q^{n\times (m+1)}$.
        \item It sets $\matU_u\leftarrow\matS_u$ for all $u\in [N]$.
        \item It samples $\{\matD_u\}_{u\in [N]}\rand\Z_q^{n\times (m+1)}$.
        \item For all $u\in \rho([\ell])$, it sets $\matQ_u\leftarrow \matU_{u}-(M_{\rho^{-1}(u),1}(\vecy\mid \veczero\mid\cdots\mid\veczero)
        +\sum_{2\leq j\leq s_{\max}}M_{\rho^{-1}(u),j}\matB_j)$. 
        \item For all $u\in [N]\setminus\rho([\ell])$, it sets $\matQ_{u}\leftarrow\matU_{u}-\matD_{u}$.
        \item It sets $\matU\leftarrow[\matU_1\mid\cdots\mid\matU_{N}]$ and computes $\matC\leftarrow\Com(\pp,\matU)$.
        \item It sets $\matA\leftarrow \matA'-\matC$.
        \item It sends the public parameters $$\pk=(\pp,n,m,q,\sigma,\chi,\chi_1,\chi_{s},\matA,\{\matB_{i}\}_{i\in\{2,\ldots,s_{\max}\}},\{\matD_u\}_{u\in [N]},\{\matQ_u\}_{u\in [N]},\vecy)$$ to the adversary $\mathcal{A}$.
    \end{enumerate}
    \iitem{Query Phase} To respond to a secret-key query for an attribute set $U\subseteq\mathbb{U}$, the algorithm $\mathcal{B}$ proceeds as follows:
    \begin{enumerate}
        \item It first computes $\matV=[\matV_1\mid\cdots\mid \matV_{N}]\leftarrow \ver(\pp,1^{(m+1)N})$.
    \item It samples $\hat{\vect}\leftarrow D_{\Z,\chi}^{m}$ and sets $\vect=(1,\hat{\vect}^\top)^\top\in \Z_q^{m+1}$.
    \item It samples $\{\hat{\veck}_u\}_{u\in U}\leftarrow D_{\Z,\chi_{s}}^m$.
    \item For all $u\in U$, $\mathcal{B}$ sends $(\matA\matV_{u}+\matQ_u)\vect-\matB\hat{\veck}_u$ to its challenger in $\EXP_{\mathcal{B}}^{\lhltr,q,\sigma,\chi_1}$, and receives $\tilde{\veck}_u$.
    \item For all $u\in U$, it sets $\veck_{u}\leftarrow\hat{\veck}_u+\tilde{\veck}_{u}$.
    \item Then it provides $\sk\leftarrow (\{\veck_u\}_{u\in U},\vect)$ to the adversary $\mathcal{A}$.
    \end{enumerate}

    \iitem{Challenge Phase} In this phase, $\mathcal{B}$ proceeds exactly as the challenger in Games $\H_1$ and $\H_2$. Precisely,
    \begin{enumerate}
    \item It samples $\msg\rand\{0,1\}$.
    \item It samples $\vecs\rand\Z_q^n$, $\vece_{1}\leftarrow D_{\Z,\chi}^m$, $\vece_{2}\leftarrow D_{\Z,\chi_{s}}^{m}$, and $e_3\leftarrow D_{\Z,\chi_{s}}$.
    \item It computes $\vecc_{1}^\top\leftarrow\vecs^\top\matB+\vece_{1}^\top$,  $\vecc_{2}^\top\leftarrow \vecs^\top(\matA+\matC)+\vece_{2}^\top$, and $c_3\leftarrow \vecs^\top\vecy+\msg\cdot\lceil q/2\rfloor+e_3$.
    \item Finally, it provides the challenge ciphertext  $\ct\leftarrow(\vecc_{1},\vecc_{2},c_3)$ to the adversary $\mathcal{A}$.
    \end{enumerate}

    \iitem{Guess Phase} The algorithm $\mathcal{B}$ outputs whatever the adversary $\mathcal{A}$ outputs.

    It is straightforward that the algorithm $\mathcal{B}$ simulates either the game $\H_1$ or $\H_2$ perfectly depending on whether $\matS\rand \Z_q^{n\times (m+1)(N+1)}$ or $\matS\leftarrow \matB\matR$ for some $\matR\rand\{-1,1\}^{m\times (m+1)(N+1)}$. Therefore, the advantage of $\mathcal{B}$ in $\EXP_{\mathcal{B}}^{\lhltr,q,\sigma,\chi_1}$ is at least the advantage of $\mathcal{A}$ in distinguishing between the game $\H_1$ and $\H_2$. This completes the proof.
\end{proof}

\begin{lemma}\label{lem2}
 Suppose that $\chi_{s}\geq \chi\cdot(\chi_1+\sigma)\cdot\log q\cdot\log N\cdot\poly(m)\cdot\lambda^{\omega(1)}$. Then we have $\H_2\overset{s}{\approx}\H_3$.  
\end{lemma}

\begin{proof}
    The only difference between the two games lies in the responses to the secret-key queries. In Game $\H_2$, the secret-key responses $\{\veck_u\}_{u\in U}$ are computed as follows: the challenger sets $\veck_u\leftarrow\hat{\veck}_u+\tilde{\veck}_u$, where $\hat{\veck}_u\leftarrow D_{\Z,\chi_{s}}^m$ and samples $\tilde{\veck}_u\leftarrow \samplepre(\matB,\td_{\matB},(\matA\matV_{u}+\matQ_u)\vect-\matB\hat{\veck}_u,\chi_1),$ where $\vect=(1,\hat{\vect}^\top)^\top$ with $\hat{\vect}\leftarrow D_{\Z,\chi}^{m}$. By the construction above, we have \begin{align}\label{keyeq} \matB\veck_u=\matB\hat{\veck}_u+\matB\tilde{\veck}_u=(\matA\matV_{u}+\matQ_u)\vect, \quad\forall u\in U. \end{align} Moreover, $\tilde{\veck}_u$ is sampled from a discrete Gaussian
$D_{\Z,\chi_1}^m$ conditioned on equation \ref{keyeq}. Since
$\chi_{s}\geq \sqrt{m}\cdot\chi_1\cdot\lambda^{\omega(1)}$ by our choice
of parameters, the noise-smudging lemma (Lemma~\ref{smudge}) implies
that in Game~$\H_2$ each $\veck_u$ is statistically
indistinguishable from a fresh sample from $D_{\Z,\chi_{s}}^m$ conditioned on \eqref{keyeq}.

In Game $\H_3$, the secret-key responses $\{\veck_u\}_{u\in U}$ are generated as follows. First the challenger samples $\hat{\vect}\leftarrow D_{\Z,\chi}^{m}$ and sets $\vect\leftarrow (1,\hat{\vect}^\top)^\top$. For each $u\in U\cap\rho([\ell])$, the challenger sets $\veck_u\leftarrow \hat{\veck}_u+\tilde{\veck}_u+\matR\matV_{u}\vect+\matZ_u\vect$, where $\hat{\veck}_u\leftarrow D_{\Z,\chi_{s}}^m$ and $$\tilde{\veck}_u\leftarrow \samplepre(\matB,\td_{\matB},-M_{\rho^{-1}(u),1}\vecy-\sum_{2\leq j\leq s_{\max}}M_{\rho^{-1}(u),j}\vecb_j'
        -\matN_u\hat{\vect}-\matB\hat{\veck}_u,\chi_1).$$ By bounding the norm of $\tilde{\veck}_u+\matR\matV_{u}\vect+\matZ_u\vect$, we obtain $$\|\tilde{\veck}_u+\matR\matV_{u}\vect+\matZ_{u}\vect\|\leq \chi\cdot(\chi_1+\sigma)\log q\cdot\log N\cdot\poly(m).$$ Hence, by the noise-smudging lemma (Lemma~\ref{smudge}), the distribution of
    $\veck_u$ is again statistically indistinguishable from that of
    $\hat{\veck}_u$. Moreover, in Game $\H_3$, we have \begin{align*}
            \matB\veck_u&=\matB(\hat{\veck}_u+\tilde{\veck}_u+\matR\matV_{u}\vect+\matZ_u\vect)\\
            &=\matB\hat{\veck}_u-M_{\rho^{-1}(u),1}\vecy-\sum_{2\leq j\leq s_{\max}}M_{\rho^{-1}(u),j}\vecb_j'
        -\matN_u\hat{\vect}-\matB\hat{\veck}_u+\matB\matR\matV_{u}\vect+\matB\matZ_u\vect\\
        &=(-M_{\rho^{-1}(u),1}(\vecy\mid \veczero\mid\cdots\mid\veczero)
        -\sum_{2\leq j\leq s_{\max}}M_{\rho^{-1}(u),j}\matB_j+\matA'\matV_{u}+\matB\matZ_u)\vect\\
         &=(-M_{\rho^{-1}(u),1}(\vecy\mid \veczero\mid\cdots\mid\veczero)
        -\sum_{2\leq j\leq s_{\max}}M_{\rho^{-1}(u),j}\matB_j+\matA\matV_{u}+\matC\matV_u+\matB\matZ_u)\vect\\
        &=(-M_{\rho^{-1}(u),1}(\vecy\mid \veczero\mid\cdots\mid\veczero)
        -\sum_{2\leq j\leq s_{\max}}M_{\rho^{-1}(u),j}\matB_j+\matA\matV_{u}+\matU_u)\vect\\
        &=(\matA\matV_{u}+\matQ_u)\vect,
        \end{align*} which matches the equation \eqref{keyeq}. For each $u\in U\setminus\rho([\ell])$, the challenger sets $\veck_{u}\leftarrow \hat{\veck}_u+\tilde{\veck}_u+\matR\matV_{u}\vect+\matZ_{u}\vect$, where $\hat{\veck}_u\leftarrow D_{\Z,\chi_{s}}^m$ and $$\tilde{\veck}_u\leftarrow\samplepre(\matB,\td_{\matB},-\vecd_u'-\matD_u'\hat{\vect}-\matB\hat{\veck}_u,\chi_1).$$ By bounding the norm of $\matR\matV_{u}\vect+\matZ_{u}\vect$ in the same way as above, we obtain that the distribution of $\veck_u$ is statistically indistinguishable from that of $\hat{\veck}_u$. Furthermore, we have \begin{align*} \matB\veck_u&=\matB(\hat{\veck}_u+\tilde{\veck}_u+\matR\matV_{u}\vect+\matZ_{u}\vect)\\
        &=\matB\hat{\veck}_u-\vecd_u'-\matD_u'\hat{\vect}-\matB\hat{\veck}_u+\matB\matR\matV_u\vect+\matB\matZ_u\vect\\
            &=\matB\hat{\veck}_u-\matD_{u}\vect-\matB\hat{\veck}_u+\matA\matV_{u}\vect+\matU_{u}\vect\\
            &=(\matA\matV_{u}+\matQ_{u})\vect,
        \end{align*} which matches the key equation \eqref{keyeq} as well.
        In summary, for every attribute $u$ and every secret-key query, the
distribution of $\veck_u$ in Game~$\H_3$ is statistically
indistinguishable from a sample from $D_{\Z,\chi_s}^m$ (and hence from that of
$\veck_u$ in Game~$\H_2$), and the key equation
$\matB\veck_u=(\matA\matV_u+\matQ_u)\vect$ holds in both games.

Since the adversary makes only polynomially many secret-key queries,
a union bound shows that the joint distribution of all secret-key
answers in $\H_2$ and $\H_3$ is statistically indistinguishable.
All other parts of the two games are identical. Therefore,
    $\H_2\overset{s}{\approx}\H_3$.
\end{proof}

\begin{lemma}\label{lem3}
    Suppose that $\chi>\sqrt{m}\chi_1\cdot\lambda^{\omega(1)}$. We have $\H_3\overset{s}{\approx}\H_4$.
\end{lemma}

\begin{proof}
    This lemma follows directly from the noise-smudging lemma (Lemma \ref{smudge}). The only difference between Games~$\H_3$ and~$\H_4$ lies in the way that the challenger generates the vector component $\vect$ of each secret-key queries. In Game~$\H_3$, for each secret-key query, the challenger samples $\hat{\vect}\leftarrow D_{\Z,\chi}^m$ and sets $\vect=(1,\hat{\vect}^\top)^\top$. In Game~$\H_4$, for each secret-key query, the challenger samples $\hat{\vect}\leftarrow D_{\Z,\chi}^m, \tilde{\vect}\leftarrow D_{\Z,\chi_1}^m$, and sets $\vect\leftarrow (1,\hat{\vect}^\top+\tilde{\vect}^\top)^\top$ instead. Since $\chi>\sqrt{m}\chi_1\cdot\lambda^{\omega(1)}$, the noise-smudging lemma
    (Lemma~\ref{smudge}) implies that in Game~$\H_4$ we have
    $\hat{\vect}+\tilde{\vect}\overset{s}{\approx}\hat{\vect}$.
    Therefore, the distributions of the secret keys in Games~$\H_3$ and~$\H_4$
    are statistically indistinguishable, and hence
    $\H_3\overset{s}{\approx}\H_4$.
\end{proof}

\begin{lemma}\label{lem4}
    Let $q$ be a prime, and let $n,m$ be such that $m>2n\log q +\omega(\log n)$. Let $m\geq m_0(n,q)$ and $\sigma,\chi_1\geq \chi_0(n,q)$, where $m_0$ and $\chi_0$ are polynomials given in Lemma \ref{trapdoor}. We have $\H_4\overset{s}{\approx}\H_5$.
\end{lemma}

\begin{proof}
    This lemma follows from the leftover hash lemma with trapdoor (Lemma \ref{lefttrap}). Suppose there exists an adversary $\mathcal{A}$ that can distinguish between Game $\H_4$ and Game $\H_5$ with non-negligible advantage, we can construct an adversary $\mathcal{B}$ that can win the game $\EXP_{\mathcal{B}}^{\lhltr,q,\sigma,\chi_1}$ with non-negligible advantage, which leads to a contradiction. The algorithm $\mathcal{B}$ proceeds as follows:
    
    \iitem{Setup Phase} The algorithm $\mathcal{B}$ receives $1^\lambda,q,\sigma,\chi_1$ from its challenger. Then it invokes $\mathcal{A}$ and receives an access policy $(\matM,\rho)$, where $\matM\in\{-1,0,1\}^{\ell\times s_{\max}}$ and $\rho: [\ell]\rightarrow [N]$ is an injective function. Then it proceeds as follows:
    \begin{enumerate}
        \item It sends $1^n,1^m,1^{s_{\max}}$ to its challenger and receives a challenge $\{\pp=(\matB,\matW,\matT),\matS\}$. It then parses $\matS:=[\matS_2'\mid\cdots\mid\matS_{s_{\max}}'\mid \matS_{1}''\mid\cdots\mid\matS_{N}'']$ where $\matS_i'\in \Z_q^{n\times m}$ for each $i\in \{2,\ldots,s_{\max}\}$ and $\matS_{i}''\in\Z_q^{n\times m}$ for each $i\in [N]$.
        \item It samples $\matR\rand\{-1,1\}^{m\times m},\{\matR_u\}_{u\in [N]}\rand\{-1,1\}^{m\times (m+1)}$, and $
        \vecr\rand\{-1,1\}^{m}$.
        \item It sets $\matA'\leftarrow \matB\matR$ and $\vecy\leftarrow\matB\vecr$.
        \item It computes $\matC_i\leftarrow \Com(\pp,\vecu_{i-1}\otimes \matG)$ for all $i\in\{2,\ldots,s_{\max}\}$, and $\matC_u'\leftarrow \Com(\pp,\vecu_{u+s_{\max}-1}\otimes \matG)$ for all $u\in [N]$.
   
        \item It samples $\{\vecb_i'\}_{i\in \{2,\ldots,s_{\max}\}}\rand\Z_q^{n}$.
        \item For each $i\in \{2,\ldots,s_{\max}\}$, it sets $\matB_i'\leftarrow \matS_i'+\matC_i$ and sets $\matB_i\leftarrow[\vecb_i'\mid\matB_i']$.
        \item It computes $\matN_u\leftarrow\sum_{2\leq j\leq s_{\max}}M_{\rho^{-1}(u),j}\matB'_j$ for each $u\in \rho([\ell])$.
        \item It sets $\matU_u\leftarrow\matB\matR_u$ for all $u\in [N]$.
        \item It samples $\{\vecd_u'\}_{u\in [N]}\rand\Z_q^{n}$, computes $\matD_u'\leftarrow\matS_u''+\matC_u'$ and sets $\matD_u:=[\vecd_u'\mid\matD_u']$ for all $u\in [N]$.
        \item For all $u\in \rho([\ell])$, it sets $\matQ_u\leftarrow \matU_{u}-(M_{\rho^{-1}(u),1}(\vecy\mid \veczero\mid\cdots\mid\veczero)
        +\sum_{2\leq j\leq s_{\max}}M_{\rho^{-1}(u),j}\matB_j)$. 
        \item For all $u\in [N]\setminus\rho([\ell])$, it sets $\matQ_{u}\leftarrow\matU_{u}-\matD_{u}$.
        \item It sets $\matU\leftarrow[\matU_1\mid\cdots\mid\matU_{N}]$ and computes $\matC\leftarrow\Com(\pp,\matU)$.
        \item It sets $\matA\leftarrow \matA'-\matC$.
        \item It sends the public parameters $\begin{aligned}&\pk=(\pp,n,m,q,\sigma,\chi,\chi_1,\chi_{s},\\  &\matA,\{\matB_{i}\}_{i\in\{2,\ldots,s_{\max}\}},\{\matD_u\}_{u\in [N]},\{\matQ_u\}_{u\in [N]},\vecy)\end{aligned}$ to the adversary $\mathcal{A}$.
    \end{enumerate}
    \iitem{Query Phase} To respond to a secret-key query for an attribute set $U\subseteq\mathbb{U}$, the algorithm $\mathcal{B}$ proceeds as follows:
    \begin{enumerate}
        \item It first computes $\matV=[\matV_1\mid\cdots\mid \matV_{N}]\leftarrow \ver(\pp,1^{(m+1)N})$ and $\matZ=[\matZ_1\mid\cdots\mid\matZ_N]\leftarrow \Open(\pp,\matU)$.
    \item It samples $\hat{\vect}\leftarrow D_{\Z,\chi}^{m}$, $\tilde{\vect}\leftarrow D_{\Z,\chi_1}^m$ and sets $\vect=(1,\hat{\vect}^\top+\tilde{\vect}^\top)^\top\in \Z_q^{m+1}$.
    \item It samples $\{\hat{\veck}_u\}_{u\in U}\leftarrow D_{\Z,\chi_{s}}^m$.
    \item For all $u\in U\cap \rho([\ell])$, $\mathcal{B}$ sends $-M_{\rho^{-1}(u),1}\vecy-\sum_{2\leq j\leq s_{\max}}M_{\rho^{-1}(u),j}\vecb_j'
        -\matN_u\hat{\vect}-\matN_u\tilde{\vect}-\matB\hat{\veck}_u$ to its challenger in $\EXP_{\mathcal{B}}^{\lhltr,q,\sigma,\chi_1}$, and receives $\tilde{\veck}_u$.
        \item For all $u\in U\setminus \rho([\ell])$, $\mathcal{B}$ sends $-\vecd_u'-\matD_{u}'\hat{\vect}-\matD_u'\tilde{\vect}-\matB\hat{\veck}_u$ to its challenger in $\EXP_{\mathcal{B}}^{\lhltr,q,\sigma,\chi_1}$, and receives $\tilde{\veck}_u$.
    \item For all $u\in U$, it sets $\veck_{u}\leftarrow\hat{\veck}_u+\tilde{\veck}_u+\matR\matV_{u}\vect+\matZ_{u}\vect$.
    \item Then it provides $\sk\leftarrow (\{\veck_u\}_{u\in U},\vect)$ to the adversary $\mathcal{A}$.
    \end{enumerate}

    \iitem{Challenge Phase} In this phase, $\mathcal{B}$ proceeds exactly as the challenger in Game $\H_4$ and Game $\H_5$ does. Precisely,
    \begin{enumerate}
    \item It samples $\msg\rand\{0,1\}$, $\vecs\rand\Z_q^n$, $\vece_{1}\leftarrow D_{\Z,\chi}^m$, $\vece_{2}\leftarrow D_{\Z,\chi_{s}}^{m}$, and $e_3\leftarrow D_{\Z,\chi_{s}}$.
    \item It computes $\vecc_{1}^\top\leftarrow\vecs^\top\matB+\vece_{1}^\top$,  $\vecc_{2}^\top\leftarrow \vecs^\top(\matA+\matC)+\vece_{2}^\top$, and $c_3\leftarrow \vecs^\top\vecy+\msg\cdot\lceil q/2\rfloor+e_3$.
    \item Finally, it provides the challenge ciphertext  $\ct\leftarrow(\vecc_{1},\vecc_{2},c_3)$ to the adversary $\mathcal{A}$.
    \end{enumerate}

    \iitem{Guess Phase} The algorithm $\mathcal{B}$ outputs whatever the adversary $\mathcal{A}$ outputs.

    It is straightforward that the algorithm $\mathcal{B}$ simulates either the game $\H_4$ or $\H_5$ perfectly depending on whether $\matS\rand \Z_q^{n\times m(N+ s_{\max}-1)}$ or $\matS\leftarrow \matB\matR$ for some $\matR\rand\{-1,1\}^{m\times  m(N+ s_{\max}-1)}$. Therefore, the advantage of $\mathcal{B}$ in $\EXP_{\mathcal{B}}^{\lhltr,q,\sigma,\chi_1}$ is no less than the advantage of $\mathcal{A}$ in distinguishing the game $\H_4$ and $\H_5$. This completes the proof.
\end{proof}

\begin{lemma}\label{lem5}
     Suppose that $\chi_{s}\geq \chi\cdot(\chi_1+\sigma)\log q\cdot\log N\cdot\poly(m)$ and $\chi_1 > \poly(m,N,\sigma,\log q)\cdot \omega(\sqrt{\log n})$. We have $\H_5\overset{s}{\approx}\H_6$.
\end{lemma}

\begin{proof}
The only difference between Games~$\H_5$ and~$\H_6$ lies in the way the challenger
responds to secret-key queries. In Game $\H_5$, the challenger first samples $\hat{\vect}\leftarrow D_{\Z,\chi}^m$ and $\tilde{\vect}\leftarrow D_{\Z,\chi_1}^m$, and sets $\vect\leftarrow (1\mid \hat{\vect}^\top+\tilde{\vect}^\top)^\top$. For each $u\in U\cap\rho([\ell])$, it samples $$\tilde{\veck}_{u}\leftarrow \samplepre(\matB,\td_{\matB},-M_{\rho^{-1}(u),1}\vecy-\sum_{2\leq j\leq s_{\max}}M_{\rho^{-1}(u),j}\vecb_j'
        -\matN_u\hat{\vect}-\matN_u\tilde{\vect}-\matB\hat{\veck}_u,\chi_1).$$ 
For each $u\in U\setminus\rho([\ell])$, it samples $$\tilde{\veck}_{u}\leftarrow \samplepre(\matB,\td_{\matB},-\vecd_u'-\matD_{u}'\hat{\vect}-\matD_u'\tilde{\vect}-\matB\hat{\veck}_u,\chi_1).$$ For each $u\in U$, it sets $\veck_u\leftarrow \hat{\veck}_u+\tilde{\veck}_u+\matR\matV_u\vect+\matZ_u\vect$. In Game $\H_6$, the challenger first samples $\hat{\vect}\leftarrow D_{\Z,\chi}^{m}$ and samples $$\begin{aligned}\left(\begin{matrix}
        \tilde{\veck}_{u_1}\\
        \vdots\\
        \tilde{\veck}_{u_g}\\
        \tilde{\veck}_{u_1'}\\
        \vdots\\
        \tilde{\veck}_{u_h'}\\
        \tilde{\vect}
    \end{matrix}\right)=\samplepre(\left[\begin{matrix}
        \matI_{g+h}\otimes \matB\mid\matW'
    \end{matrix}\right],\td,\left[\begin{matrix}
        -\vecn_{u_1}-\matN_{u_1}\hat{\vect}-\matB\hat{\veck}_{u_1}\\
        \vdots\\
        -\vecn_{u_g}
        -\matN_{u_g}\hat{\vect}-\matB\hat{\veck}_{u_g}\\
        -\vecd_{u_1'}'-\matD_{u_1'}'\hat{\vect}-\matB\hat{\veck}_{u_1'}\\
        \vdots\\
        -\vecd_{u_h'}'-\matD_{u_h'}'\hat{\vect}-\matB\hat{\veck}_{u_h'}
    \end{matrix}\right],\chi_1)\end{aligned}$$ and sets $\vect\leftarrow (1, \hat{\vect}^\top+\tilde{\vect}^\top)^\top$. 

    In Game $\H_6$, for each $u_i\in U\cap\rho([\ell])$, from the preimage sampling process, we obtain that $$\matB\tilde{\veck}_{u_i}+\matN_{u_i}\tilde{\vect}=-\vecn_{u_i}-\matN_{u_i}\hat{\vect}-\matB\hat{\veck}_{u_i},$$ which matches the secret-key equation in Game $\H_5$. For each $u_j'\in U\setminus \rho([\ell])$, we have $$\matB\tilde{\veck}_{u_j'}+\matD_{u_j'}'\tilde{\vect}=-\vecd_{u_j'}'-\matD_{u_j'}'\hat{\vect}-\matB\hat{\veck}_{u_j'},$$ which also matches the corresponding secret-key equation in Game $\H_5$.

    By the noise-smudging lemma, in either game, since $\chi>\sqrt{m}\chi_1\cdot\lambda^{\omega(1)}$, the distribution of $\hat{\vect}+\tilde{\vect}$ is statistically indistinguishable from a vector sampled from $D_{\Z,\chi}^m$. By bounding the norm of $\tilde{\veck}_u+\matR\matV_{u}\vect+\matZ_u\vect$, we have $$\|\tilde{\veck}_{u}+\matR\matV_u\vect+\matZ_{u}\vect\|\leq \chi\cdot(\chi_1+\sigma)\log q\cdot\log N\cdot\poly(m).$$ According to the selection of the parameters, we know each $\veck_u$ in either game is statistically indistinguishable from a sample from $D_{\Z,\chi_s}^m$ conditioned on $\matB\veck_u=(\matA\matV_u+\matQ_u)\vect$.
    
    It suffices to clarify that the challenger can construct the trapdoor for the matrix $[\matI_{g+h}\otimes \matB\mid \matW']$ efficiently. Denote $s=N+s_{\max}-1$. From the correctness of the commitment scheme, we have \begin{align*}
       & \matC_i\matV_{ns}=\vecu_{i-1}\otimes \matG-\matB\matZ_i,\quad \forall i\in \{2,\ldots,s_{\max}\},\\
        &\matC_i'\matV_{ns}=\vecu_{i+s_{\max}-1}\otimes \matG-\matB\matZ_i',\quad \forall i\in [N].
    \end{align*}
    If we write $\matC'=[\matC_2^\top\mid\cdots\mid\matC_{s_{\max}}^\top\mid\matC_1'^\top\mid\cdots\mid\matC_{N}'^\top]^\top$ and $\matZ'=[\matZ_2^\top\mid\cdots\mid\matZ_{s_{\max}}^\top\mid\matZ_1'^\top\mid\cdots\mid\matZ_{N}'^\top]^\top$, we have \begin{align}\label{succincttrap}
        [\matI_s\otimes \matB\mid\matC']\left[\begin{matrix}
           \matZ'\\
           \matV_{ns}
\end{matrix}\right]=\matI_{s}\otimes\matG
    \end{align}
    Then by the norm bound of the matrices $\matZ'$ and $\matV_{ns}$, the matrix $\left[\begin{matrix}
           \matZ'\\
           \matV_{ns}
\end{matrix}\right]$ serves as a trapdoor for $[\matI_s\otimes \matB\mid\matC']$. Denote the set $I=\rho^{-1}(U)\cap [\ell]$ and $\matM_{I}\in\{-1,0,1\}^{|I|\times (s_{\max}-1)}$ the submatrix of $\matM$ containing all the rows with indices in $I$ and removing the first column. Let $\vece_i\in\{0,1\}^N$ be the unit vector whose $i$-th entry is 1. 

Let $\matM_{U}=\left[\begin{array}{cc}
    \matM_I &  \\
     & \begin{array}{c}
         \vece_{u_1'}^\top\\
      \vece_{u_2'}^\top\\
     \vdots\\
     \vece_{u_h'}^\top
     \end{array}
\end{array}\right]\in\{-1,0,1\}^{(g+h)\times s}$. Now we prove the matrix $\matM_U$ has full row rank. By the restriction on secret-key queries submitted by the adversary, the rows
of the access matrix in $\rho^{-1}(U)$ are required to be linearly independent
and unauthorized. This means that no non-zero linear combination over $\Z_q$
of the vectors obtained by removing the first entries of the rows of $\matM$
with indices in $\rho^{-1}(U)$ can span the zero vector in dimension
$s_{\max}-1$, i.e., $\matM_I$ has full row rank. Since $u_1',\ldots, u_h'$ are different from each other, it also follows that $\vece_{u_1'},\ldots,\vece_{u_h'}$ are linearly independent. Therefore, we can conclude $\matM_U$ has full row rank.

Let $\matR'=[\matR_2'^\top\mid\cdots\mid\matR_{s_{\max}}'^\top\mid\matR_1''^\top\mid\cdots\mid\matR_{N}''^\top]^\top$. We prove $$\matW'=(\matM_U\otimes\matI_{n})((\matI_s\otimes \matB)\matR'+\matC').$$ It follows from that \begin{align*}
    (\matM_U\otimes\matI_{n})((\matI_s\otimes \matB)\matR'+\matC')&=(\matM_U\otimes\matI_n)\left[\begin{array}{c}
        \matB\matR_2'+\matC_2\\
        \vdots\\
    \matB\matR_{s_{\max}}'+\matC_{s_{\max}}\\
        \matB\matR_1''+\matC_1'\\
        \vdots\\
        \matB\matR_{N}''+\matC_N'
    \end{array}\right]=(\matM_U\otimes\matI_n)\left[\begin{array}{c}
        \matB_2'\\
        \vdots\\
    \matB_{s_{\max}}'\\
        \matD_1'\\
        \vdots\\
        \matD_{N}'
    \end{array}\right]\\&=\left[\begin{array}{c}
        \sum_{2\leq j\leq s_{\max}}M_{\rho^{-1}(u_1),j}\matB_j'\\
        \vdots\\
    \sum_{2\leq j\leq s_{\max}}M_{\rho^{-1}(u_g),j}\matB'_j\\
        \matD_{u_1'}'\\
        \vdots\\
        \matD_{u_h'}'
    \end{array}\right]=\left[\begin{array}{c}
       \matN_{u_1}\\
        \vdots\\
    \matN_{u_g}\\
        \matD_{u_1'}'\\
        \vdots\\
        \matD_{u_h'}'
    \end{array}\right]=\matW'
\end{align*}

We claim that $\td:=\left[\begin{array}{c}
     (\matM_U\otimes \matI_m)(\matZ'-\matR'\matV_{ns})\\
     \matV_{ns}
\end{array}\right]$ can serve as a trapdoor for the matrix $[\matI_{g+h}\otimes \matB\mid\matW']$. 
\begin{align*}
    &[\matI_{g+h}\otimes \matB\mid\matW']\left[\begin{array}{c}
     (\matM_U\otimes \matI_m)(\matZ'-\matR'\matV_{ns})\\
     \matV_{ns}
\end{array}\right]\\
&=(\matI_{g+h}\otimes \matB) (\matM_U\otimes \matI_m)(\matZ'-\matR'\matV_{ns})+(\matM_U\otimes\matI_{n})((\matI_s\otimes \matB)\matR'+\matC')\matV_{ns}\\
& =(\matM_U\otimes \matI_n)(\matI_{s}\otimes \matB)(\matZ'-\matR'\matV_{ns})+(\matM_U\otimes\matI_{n})((\matI_s\otimes \matB)\matR'+\matC')\matV_{ns}\\
&=(\matM_U\otimes \matI_n)[\matI_s\otimes \matB\mid\matC']\left[\begin{array}{c}
     \matZ'  \\
     \matV_{ns}
\end{array}\right]=(\matM_U\otimes \matI_n)(\matI_s\otimes \matG)=\matM_U\otimes \matG.
\end{align*}
Since $\matM_U$ has full row rank, we know $\td$ serves as a trapdoor for $[\matI_{g+h}\otimes \matB\mid\matW']$ by Lemma \ref{tensortrap}.
\end{proof}

\begin{lemma}\label{lem6}
    Suppose that $\chi_{s}>m^2\chi\cdot\lambda^{\omega(1)}$. We have $\H_6\overset{s}{\approx}\H_7$.
\end{lemma}

\begin{proof}
    The only difference between Games~$\H_6$ and~$\H_7$ lies in the way the
    ciphertext components $\vecc_2$ and $c_3$ are sampled. In Game $\H_6$, the challenger first samples $\vecs\rand\Z_q^n,\vece_1\leftarrow D_{\Z,\chi}^m,\vece_2\leftarrow D_{\Z,\chi_{s}}^m, e_3\leftarrow D_{\Z,\chi_{s}}$, and computes $\vecc_1^\top\leftarrow\vecs^\top\matB+\vece_1^\top$, $\vecc_2^\top\leftarrow\vecs^\top(\matA+\matC)+\vece_2^\top$ and $c_3\leftarrow\vecs^\top\vecy+\msg\cdot\lceil q/2\rfloor+e_3$. In Game $\H_7$, the challenger generates $\vecc_1$ exactly the same as in Game $\H_6$, samples $\vece_2\leftarrow D_{\Z,\chi_{s}}^m, e_3\leftarrow D_{\Z,\chi_{s}}$, and computes $\vecc_2^\top\leftarrow\vecc_1^\top\matR+\vece_2^\top$ and $c_3\leftarrow\vecc_1^\top\vecr+\msg\cdot\lceil q/2\rfloor+e_3$. In Game $\H_7$, using the relations $\matA'=\matB\matR$ and $\matA=\matA'-\matC$, we have that $$\vecc_2^\top\leftarrow\vecc_1^\top\matR+\vece_2^\top=\vecs^\top\matB\matR+\vece_1^\top\matR+\vece_2^\top=\vecs^\top\matA'+\vece_1^\top\matR+\vece_2^\top=\vecs^\top(\matA+\matC)+\vece_1^\top\matR+\vece_2^\top.$$
    By bounding the error term, we have $\|\vece_1^\top\matR\|\leq m^{3/2}\cdot\chi$ with overwhelming probabiliry. By the noise-smudging lemma, we obtain that $\vece_1^\top\matR+\vece_2^\top$ is statistically indistinguishable from $\vece_2^\top$. Therefore, the ciphertexts $\vecc_2$ generated in two games are statistically indistinguishable. An analogous argument applies to $c_3$. Since $(\vecc_2,c_3)$ are the only ciphertext components that differ
    between the two games, the full ciphertext distributions in
    Games~$\H_6$ and~$\H_7$ are statistically indistinguishable. Therefore,
    $\H_6\overset{s}{\approx}\H_7$.
\end{proof}

\begin{lemma}\label{lem7}
    Suppose that $(2m^2,\sigma)$-succinct $\LWE$ assumption holds. We have $\H_7\overset{c}{\approx}\H_8$.
\end{lemma}
\begin{proof}
 Suppose that there exists an \ppt adversary $\mathcal{A}$ that can distinguish between Game $\H_7$ and $\H_8$ with non-negligible advantage. We construct an \ppt adversary $\mathcal{B}$ that can break the $(2m^2,\sigma)$-succinct $\LWE$ assumption with non-negligible advantage. The adversary $\mathcal{B}$ proceeds as follows:
 
\iitem{Setup Phase} The algorithm $\mathcal{B}$ receives from its challenger a succinct $\LWE$ challenge $(\matB,\matW,\matT,\vecc)$, where $\matB\in\Z_q^{n\times m},\matW\in\Z_q^{2m^2n\times m},\matT\in\Z_q^{(2m^2+1)m\times 2m^3}$, and $\vecc$ is sampled as $\vecc^\top\leftarrow \vecs^\top\matB+\vece^\top$ for some $\vecs\in\Z_q^n,\vece\leftarrow D_{\Z,\chi}^m$ or $\vecc\rand\Z_q^m$. It sets $\pp := (\matB,\matW,\matT)$.
\begin{enumerate}
    \item It samples $\vecr\rand\{-1,1\}^{m},\{\matR_u\}_{u\in [N]}\rand\{-1,1\}^{m\times (m+1)}$, and $\matR,\{\matR_i'\}_{i\in \{2,\ldots,s_{\max}\}},\{\matR_u''\}_{u\in [N]}\rand\{-1,1\}^{m\times m}.$
    \item It sets $\matA'\leftarrow\matB\matR$ and $\vecy\leftarrow\matB\vecr$.
    \item For all $i\in\{2,\ldots,s_{\max}\}$, it computes $\matC_i\leftarrow\Com(\pp,\vecu_{i-1}\otimes \matG)$. For all $u\in [N]$, it computes $\matC_u'\leftarrow \Com(\pp,\vecu_{u+s_{\max}-1}\otimes\matG)$.
    \item It samples $\{\vecb_i'\}_{2,\ldots,s_{\max}}\rand \Z_q^n$ and $\{\vecd_u'\}_{u\in [N]}\rand\Z_q^n$.
    \item For all $i\in \{2,\ldots,s_{\max}\}$, it sets $\matB_i'\leftarrow \matB\matR_i'+\matC_i$ and $\matB_i\leftarrow[\vecb_i'\mid\matB_i']$.
    \item  For all $u\in [N]$, it sets $\matD_u'\leftarrow \matB\matR_u''+\matC_u'$ and $\matD_u\leftarrow [\vecd_u'\mid\matD_u']$.
    \item For all $u\in \rho([\ell])$, it sets $\matN_u\leftarrow\sum_{2\leq j\leq s_{\max}}M_{\rho^{-1}(u),j}\matB'_j$ and $\vecn_u\leftarrow M_{\rho^{-1}(u),1}\vecy+\sum_{2\leq j\leq s_{\max}}M_{\rho^{-1}(u),j}\vecb_j'$.
    
    \item For all $u\in [N]$ set $\matU_{u}\leftarrow\matB\matR_u$, and set $\matU\leftarrow[\matU_1\mid\cdots\mid\matU_{N}]$.
    \item For all $u\in \rho([\ell])$, it sets $\matQ_{u}\leftarrow \matU_{u}-(M_{\rho^{-1}(u),1}(\vecy\mid \veczero\mid\cdots\mid\veczero)+\sum_{2\leq j\leq s_{\max}}M_{\rho^{-1}(u),j}\matB_j)$.
    \item For all $u\in [N]\setminus\rho([\ell])$, it sets $\matQ_u\leftarrow \matU_{u}-\matD_{u}$. 
\item It computes $\matC\leftarrow \Com(\pp,\matU)$ and sets $\matA\leftarrow \matA'-\matC\in \Z_q^{n\times m}$.
   \item It sets the public key $$\begin{aligned}\pk=(\pp,n,m,q,\sigma,\chi,\chi_1,\chi_{s},\matA,\{\matB_{i}\}_{i\in\{2,\ldots,s_{\max}\}},\{\matD_u\}_{u\in [N]},\{\matQ_u\}_{u\in [N]},\vecy)\end{aligned}$$ and sends it to $\mathcal{A}$.
\end{enumerate}

\iitem{Query Phase} To respond to a secret-key query for an attribute set
$U\subseteq\mathbb{U}=[N]$, the adversary $\mathcal{B}$ proceeds
exactly as in Games~$\H_7$ and~$\H_8$:
    \begin{enumerate}
    \item It first computes $\matV=[\matV_1\mid\cdots\mid \matV_{N}]\leftarrow \ver(\pp,1^{(m+1)N})$ and $\matZ=[\matZ_1\mid\cdots\mid\matZ_N]\leftarrow\Open(\pp,\matU)$.
    \item It samples $\hat{\vect}\leftarrow D_{\Z,\chi}^m$.
    \item Denote $U\cap\rho([\ell])=\{u_1,\ldots,u_g\}$ and $U\setminus \rho([\ell])=\{u_1',\ldots,u_h'\}$ and set $\matW'\leftarrow [\matN_{u_1}^\top\mid \cdots\mid\matN_{u_g}^\top\mid\matD_{u_1'}'^\top\mid\cdots\mid\matD_{u_h'}'^\top]^\top$.
    \item It samples $\{\hat{\veck}_u\}_{u\in U}\leftarrow D_{\Z,\chi_{s}}^m$.
     \item It computes $\begin{aligned}\left(\begin{matrix}
        \tilde{\veck}_{u_1}\\
        \vdots\\
        \tilde{\veck}_{u_g}\\
        \tilde{\veck}_{u_1'}\\
        \vdots\\
        \tilde{\veck}_{u_h'}\\
        \tilde{\vect}
    \end{matrix}\right)=\samplepre(\left[\begin{matrix}
        \matI_{g+h}\otimes \matB\mid\matW'
    \end{matrix}\right],\td,\left[\begin{matrix}
        -\vecn_{u_1}-\matN_{u_1}\hat{\vect}-\matB\hat{\veck}_{u_1}\\
        \vdots\\
        -\vecn_{u_g}
        -\matN_{u_g}\hat{\vect}-\matB\hat{\veck}_{u_g}\\
        -\vecd_{u_1'}'-\matD_{u_1'}'\hat{\vect}-\matB\hat{\veck}_{u_1'}\\
        \vdots\\
        -\vecd_{u_h'}'-\matD_{u_h'}'\hat{\vect}-\matB\hat{\veck}_{u_h'}
    \end{matrix}\right],\chi_1)\end{aligned}$ using the trapdoor $\td$ for $[\matI_{g+h}\otimes\matB\mid\matW']$.
The construction of $\td$ is given in the proof of Lemma~\ref{lem5}.
    \item It sets $\vect=(1,\hat{\vect}^\top+\tilde{\vect}^\top)^\top\in \Z_q^{m+1}$.
    
    \item For all $u\in U$, it computes  $\veck_{u}\leftarrow\hat{\veck}_u+\tilde{\veck}_{u}+\matR\matV_{u}\vect+\matZ_u\vect$.
     
    \item Then it provides the secret-key response $\sk\leftarrow (\{\veck_u\}_{u\in U},\vect)$ to the adversary $\mathcal{A}$.
    \end{enumerate}

    \iitem{Challenge Phase} 
    \begin{enumerate}
    \item It sets $\vecc_1\leftarrow\vecc$.
    \item It samples $\msg\rand\{0,1\}$, $\vece_{2}\leftarrow D_{\Z,\chi_{s}}^{m}$, and $e_3\leftarrow D_{\Z,\chi_{s}}$.
    \item It computes $\vecc_{2}^\top\leftarrow \vecc_1^\top\matR+\vece_2^\top$, and $c_3\leftarrow \vecc_1^\top\vecr+\msg\cdot\lceil q/2\rfloor+e_3$.
    \item Finally, it provides the challenge ciphertext  $\ct\leftarrow(\vecc_{1},\vecc_{2},c_3)$ to the adversary $\mathcal{A}$.
    \end{enumerate}

    \iitem{Guess Phase} The algorithm $\mathcal{B}$ outputs whatever the adversary $\mathcal{A}$ outputs.

  It is straightforward that the algorithm $\mathcal{B}$ simulates either
Game~$\H_7$ or Game~$\H_8$ perfectly, depending on whether
$\vecc^\top=\vecs^\top\matB+\vece_1^\top$ for some
$\vecs\in\Z_q^n$ and $\vece_1\leftarrow D_{\Z,\chi}^m$, or
$\vecc\rand\Z_q^m$. Therefore, the advantage of $\mathcal{B}$ in breaking the succinct $\LWE$ assumption is no less than the advantage of $\mathcal{A}$ in distinguishing Games~$\H_7$ and $\H_8$, which leads to a contradiction. Hence, we complete the proof.
\end{proof}

\begin{lemma}\label{lem8}
    Suppose that $m>2n\log q+\omega(\log n)$. We have $\H_8\overset{s}{\approx}\H_9$.
\end{lemma}

\begin{proof}
    This lemma follows from the leftover hash lemma (Lemma \ref{left}). Suppose there exists an adversary $\mathcal{A}$ that can distinguish between Game $\H_8$ and Game $\H_9$ with non-negligible advantage, we can construct an adversary $\mathcal{B}$ that can break leftover hash lemma with non-negligible advantage, which leads to a contradiction. The algorithm $\mathcal{B}$ proceeds as follows:

    \iitem{Setup Phase} The algorithm $\mathcal{B}$ first invokes $\mathcal{A}$ and receives an access policy $(\matM,\rho)$, where $\matM\in\{-1,0,1\}^{\ell\times s_{\max}}$ and $\rho: [\ell]\rightarrow \mathbb{U}$ is an injective function. Then it proceeds as follows:
    \begin{enumerate}
        \item It receives a challenge $(\vecb_0^\top,s_0)$, where $\vecb_0\in\Z_q^m, s_0\in \Z_q$.
        \item It samples $(\matB,\td_{\matB})\leftarrow\trapgen(1^n,1^m,q)$ and $\matW\rand\Z_q^{2m^2n\times m}$.
        \item Then it samples $\matT\leftarrow \samplepre([\matI_{2m^2}\otimes \matB\mid \matW],\begin{bmatrix} \matI_{2m^{2}}\otimes \td_{\matB} \\ \mathbf{0} \end{bmatrix}, \matI_{2m^2}\otimes \matG,\sigma)$.
\item It sets $\pp:=(\matB,\matW,\matT)$.        
         \item It samples $\vecr\rand\{-1,1\}^{m},\{\matR_u\}_{u\in [N]}\rand\{-1,1\}^{m\times (m+1)}$, and $\matR,\{\matR_i'\}_{i\in \{2,\ldots,s_{\max}\}},\{\matR_u''\}_{u\in [N]}\rand\{-1,1\}^{m\times m}.$
    \item It sets $\matA'\leftarrow\matB\matR$ and $\vecy\leftarrow\matB\vecr$.
    \item For all $i\in\{2,\ldots,s_{\max}\}$, it computes $\matC_i\leftarrow\Com(\pp,\vecu_{i-1}\otimes \matG)$. For all $u\in [N]$, it computes $\matC_u'\leftarrow \Com(\pp,\vecu_{u+s_{\max}-1}\otimes\matG)$.
    \item It samples $\{\vecb_i'\}_{i=2}^{s_{\max}}\rand\Z_q^n$ and
$\{\vecd_u'\}_{u\in [N]}\rand\Z_q^n$.
    \item For all $i\in \{2,\ldots,s_{\max}\}$, it sets $\matB_i'\leftarrow \matB\matR_i'+\matC_i$ and $\matB_i\leftarrow[\vecb_i'\mid\matB_i']$.
    \item  For all $u\in [N]$, it sets $\matD_u'\leftarrow \matB\matR_u''+\matC_u'$ and $\matD_u\leftarrow [\vecd_u'\mid\matD_u']$.
    \item For all $u\in \rho([\ell])$, it sets $\matN_u\leftarrow\sum_{2\leq j\leq s_{\max}}M_{\rho^{-1}(u),j}\matB'_j$ 
    and $\vecn_u\leftarrow M_{\rho^{-1}(u),1}\vecy+\sum_{2\leq j\leq s_{\max}}M_{\rho^{-1}(u),j}\vecb_j'$.
    
    \item For all $u\in [N]:\matU_{u}\leftarrow\matB\matR_u$ and sets $\matU\leftarrow[\matU_1\mid\cdots\mid\matU_{N}]$.
    \item For all $u\in \rho([\ell])$, it sets $\matQ_{u}\leftarrow \matU_{u}-(M_{\rho^{-1}(u),1}(\vecy\mid \veczero\mid\cdots\mid\veczero)+\sum_{2\leq j\leq s_{\max}}M_{\rho^{-1}(u),j}\matB_j)$.
    \item For all $u\in [N]\setminus\rho([\ell])$, it sets $\matQ_u\leftarrow \matU_{u}-\matD_{u}$. 
\item It computes $\matC\leftarrow \Com(\pp,\matU)$ and sets $\matA\leftarrow \matA'-\matC\in \Z_q^{n\times m}$.
   \item The challenger sets the public key $$\begin{aligned}\pk=(\pp,n,m,q,\sigma,\chi,\chi_1,\chi_{s},\matA,\{\matB_{i}\}_{i\in\{2,\ldots,s_{\max}\}},\{\matD_u\}_{u\in [N]},\{\matQ_u\}_{u\in [N]},\vecy)\end{aligned}$$ and sends it to $\mathcal{A}$.
   \end{enumerate}

   \iitem{Query Phase} To respond to a secret-key query for an attribute set
$U\subseteq\mathbb{U}=[N]$, the adversary $\mathcal{B}$ proceeds
exactly as in Games~$\H_8$ and~$\H_9$.
   
    \iitem{Challenge Phase} Precisely,
    \begin{enumerate}
    \item It samples $\msg\rand\{0,1\}$,  $\vece_{2}\leftarrow D_{\Z,\chi_{s}}^{m},e_3\leftarrow D_{\Z,\chi_s}$.
    \item It samples $\vecc_1\leftarrow \vecb_0$ and computes   $\vecc_{2}^\top\leftarrow \vecc_1^\top\matR+\vece_2^\top$, and $c_3\leftarrow s_0+\msg\cdot\lceil q/2\rfloor+e_3$.
    \item Finally, it provides the challenge ciphertext  $\ct\leftarrow(\vecc_{1},\vecc_{2},c_3)$ to the adversary $\mathcal{A}$.
    \end{enumerate}

    \iitem{Guess Phase} The algorithm $\mathcal{B}$ outputs whatever the adversary $\mathcal{A}$ outputs.

    It is straightforward that the algorithm $\mathcal{B}$ simulates either the game $\H_8$ or $\H_9$ depending on whether $s_0=
        \vecb_0^\top\vecr$ for some $\vecr\rand\{-1,1\}^{m}$ or $s_0\rand\Z_q$. We know $\mathcal{B}$ simulates either Game $\H_8$ or Game $\H_9$ up to negligible statistical distance. Therefore, the advantage of the algorithm $\mathcal{B}$ of breaking leftover hash lemma has negligible difference from the advantage of the adversary $\mathcal{A}$ of distinguishing between Game $\H_8$ and $\H_9$.

    Hence $\H_8\overset{s}{\approx}\H_9$.
\end{proof}
\end{proof}

\subsection{Broadcast Encryption}
As discussed in \cite{BV22,Wee22,Wee24}, $\CPABE$ for circuits gives a broadcast encryption for $N=\ell$ users. In the broadcast encryption, we can use a circuit of size $O(N\log N)$ and depth $O(\log N)$ to check the membership of the users. More precisely, 
\begin{corollary}[Broadcast encryption]
Under the $\poly(\lambda)$-succinct $\LWE$ assumption, we have a broadcast encryption scheme for $N$ users with parameters 
$$|\pk|=N\cdot \poly(\lambda,\log N),\quad |\sk|=\poly (\lambda,\log N),\quad |\ct|=\poly(\lambda).$$
\end{corollary}
\bibliographystyle{alpha}
\bibliography{reference}
\end{document}